\documentclass[11pt,letterpaper]{article}

\usepackage{times,amssymb,amsthm,amsmath,multirow,amsfonts,enumerate,bbm,setspace,pdflscape,lscape,mathtools,float}
\usepackage[T1]{fontenc}
\usepackage[margin=1in]{geometry}
\usepackage[small,compact]{titlesec}
\usepackage[normalem]{ulem}
\usepackage[round,authoryear]{natbib}
\usepackage[colorlinks=true,linkcolor=red,urlcolor=blue,citecolor=blue,anchorcolor=blue,
pdftex,breaklinks,pdfencoding=auto,psdextra,
bookmarksopenlevel=1,bookmarksopen=true]{hyperref}

\usepackage{xr}
\usepackage{booktabs}
\usepackage{graphicx,epstopdf,color,xcolor}
\graphicspath{{graphics/}}
\usepackage{subcaption} 
\makeatletter
\DeclareRobustCommand\citepos
{\begingroup\def\NAT@nmfmt##1{{\NAT@up##1's}}%
	\NAT@swafalse\let\NAT@ctype\z@\NAT@partrue
	\@ifstar{\NAT@fulltrue\NAT@citetp}{\NAT@fullfalse\NAT@citetp}}
\makeatother

\makeatletter
\renewcommand*{\eqref}[1]{\hyperref[{#1}]{\textup{\tagform@{\ref*{#1}}}}}
\makeatother

\definecolor{pink}{RGB}{219, 48, 122} 
\usepackage[inline,shortlabels]{enumitem}

\expandafter
\def \expandafter \normalsize \expandafter{\normalsize \setlength \abovedisplayskip{4pt plus 1pt minus 3pt}}
\expandafter
\def \expandafter \normalsize \expandafter{\normalsize \setlength \abovedisplayshortskip{0pt plus 2pt}}
\expandafter
\def \expandafter \normalsize \expandafter{\normalsize \setlength \belowdisplayskip{4pt plus 1pt minus 3pt}}
\expandafter
\def \expandafter \normalsize \expandafter{\normalsize \setlength \belowdisplayshortskip{4pt plus 1pt minus 3pt}}

\numberwithin{equation}{section}

\makeatletter

\newcommand{\bt}{\textcolor{black}} 

\newtheorem{theorem}{Theorem}[section]
\newtheorem{lemma}{Lemma}[section]

\newtheorem{proposition}{Proposition}[section]
\newtheorem{remark}{Remark}[section]

\newtheorem{assumpK}{Assumption}

\newtheorem{assumpWW}{Assumption}

\newtheorem{assumpMM}{Assumption}

\newtheorem*{assumption*}{Assumption}
\newtheorem*{proof*}{Proof of}

\newenvironment{claim}[1]{\par\noindent\underline{Claim 1:}\space#1}{}
\newenvironment{claim2}[1]{\par\noindent\underline{Claim 2:}\space#1}{}
\newenvironment{claim3}[1]{\par\noindent\underline{Claim 3:}\space#1}{}

\mathtoolsset{mathic=true}
\DeclareMathOperator{\ran}{ran}

\DeclareMathOperator{\LRV}{LRV}
\DeclareMathOperator{\spn}{span}

\DeclareMathOperator{\tr}{tr}

\DeclareMathOperator{\rank}{rank}

\DeclareMathOperator{\Tr}{tr}
\DeclareMathOperator{\ii}{i}
\DeclareMathOperator{\sgn}{sgn}
\DeclareMathOperator{\fdd}{fdd}
\numberwithin{equation}{section}

\begin{document}
	\date{}
	\title{\vspace{-1.5em}\Large Functional principal component analysis for cointegrated functional time series}
	\author{	\large	Won-Ki Seo\thanks{I am grateful to anonymous referees, the coeditor, and the editor of the journal for helpful comments. I also thank Brendan Beare and Morten Nielsen for their comments.  This note supersedes the manuscript ``Fully modified functional principal component analysis for cointegrated functional time series'', posted at \url{https://arxiv.org/abs/2011.12781}.  Email : won-ki.seo@sydney.edu.au 
		}\\  University of Sydney
	}
	\maketitle \vspace{-2.1em}
	\begin{abstract}
		\vspace{-.3em}	Functional principal component analysis (FPCA) has played an important role in the development of functional time series analysis. This note investigates how FPCA can be used to analyze cointegrated functional time series and proposes a modification of FPCA as a novel statistical tool. Our modified FPCA not only provides an asymptotically more efficient estimator of the cointegrating vectors, but also leads to novel FPCA-based tests for examining  essential properties of cointegrated functional time series. \\[2pt]
				 \vspace{-.1em}	
	\end{abstract}

	
	\section{Introduction}
	\vspace{-0.15em}
	
	Functional principal component analysis (FPCA) has been a central tool for functional time series (FTS) in various contexts encompassing  FTS regression \citep{Park2012397,seong2022}, prediction \citep{hyndman2007robust,aue2015prediction} and long memory FTS \citep{li2020long}  to name a few.  In the recent work by \cite{Chang2016152},  FPCA is applied to 
	nonstationary cointegrated FTS. Given that many economic time series are nonstationary but allow a long run stable relationship, their analysis paves the way for potential  applications and extensions (see e.g., \citealp{chang2020evaluating}, \citealp{LRS}, \citealp{seoshang2022}). 
	
	This note further investigates how FPCA can be used in statistical analysis of cointegrated FTS $\{X_t\}_{t\geq 1}$ taking values in a Hilbert space $\mathcal H$. Cointegrated FTS considered in this paper 
	can be decomposed by orthogonal projections $P^N$ and $P^S = I-P^N$ (where $I$ denotes the identity map) into a finite dimensional unit root process $\{P^NX_t\}_{t \geq 1}$ and a potentially infinite dimensional stationary process $\{P^SX_t\}_{t \geq 1}$. When such a time series is given, it is important to estimate $P^N$ or $P^S$ since either of them characterizes the cointegrating behavior. 
	\cite{Chang2016152} earlier studied how the eigenelements of the sample covariance operator  can be used to construct a consistent estimator of $P^N$. 
	In this note, we first add some novel asymptotic results, including the convergence rate and the asymptotic limit of the existing estimator of $P^N$ (or $P^S$) obtained from FPCA. Based on the results, we propose a modified FPCA methodology that guarantees (i) a more asymptotically efficient estimator and (ii) a more convenient asymptotic limit which leads to novel statistical tests for examining hypotheses about cointegration.  
	\bt{Our modified FPCA is motivated from semiparametric modifications to obtain asymptotically efficient estimators of the cointegrating vectors in a multivariate cointegrated system as in \cite{phillips1990statistical} and \cite{harris1997principal}, where the latter article is more related to the present paper (this will be discussed in detail in Section \ref{sec:harris} of the Supplementary Material to this article}). 
	
	The rest of this paper is organized as follows. Section \ref{sec:prelim} provides main assumptions. In Section \ref{sec:fpca}, we discuss on the ordinary FPCA and our proposed modification of it as statistical tools to estimate $P^N$ or $P^S$, and provide the relevant asymptotic theory. Based on our modified FPCA, Section  \ref{sec:test} gives statistical tests to examine various hypotheses on cointegration. In Section \ref{sec_conclusion}, we conclude with some cautions about using our methodology in practice and present some idea how the methodology can be further extended. The Supplementary Material to this article contains six sections (Sections \ref{sec:harris}-\ref{appproof}) on some extensions of the theoretical results to be given, Monte Carlo simulations, empirical applications and mathematical proofs.  
	
	We review notation for the subsequent sections. Let $\mathcal H$ denote a real separable Hilbert space with inner product $\langle \cdot, \cdot \rangle$ and norm $\|\cdot\| = \langle \cdot,\cdot \rangle^{1/2}$. The space of bounded linear operator is denoted by  $\mathcal{L}_{\mathcal H}$, which is assumed to be equipped with the operator norm $\|\cdot\|_{\mathcal L_{\mathcal H}}$. For any $A\in \mathcal L_{\mathcal H}$, its adjoint ($A^\ast$), range ($\ran A$), kernel ($\ker A$), rank ($\rank A$), and Moore-Penrose inverse ($A^\dag$) are introduced in Section \ref{appintro01}. Moreover, various properties of $A \in \mathcal L_{\mathcal H}$ such as compactness, positive (semi)definiteness and self-adjointness are also defined in detail. As discussed in Section \ref{appintro01},  if $A$ is self-adjoint, positive semidefinite, and compact, we may define its $m$-regularized inverse  $A{\mid_{m}^\dagger}= \sum_{j=1}^m a_j^{-1} u_j \otimes u_j$ where $\{a_j\}_{j=1}^m$ are positive eigenvalues of $A$ and $\{u_j\}_{j=1}^m$ are the corresponding eigenvectors; $A{\mid_{m}^\dagger}$ is understood as the partial inverse of $A$ on the restricted domain $\spn(\{u_j\}_{j=1}^m)$ {(for convenience, we let $A{\mid_{m}^\dagger} = 0$ if $A=0$)}.  
	Random elements taking values in $\mathcal H$ and $\mathcal L_{\mathcal H}$ 
	are introduced in Section \ref{appintro}. We let $L^2_{\mathcal H}$ denote the space of $\mathcal H$-valued random variables $X$ satisfying $\mathbb{E}X=0$ and $\mathbb{E}\|X\|^2 < \infty$. For any $X \in L^2_{\mathcal H}$, its covariance operator is defined by $C_X = \mathbb{E} \left[  X \otimes X  \right]$. If $A$ is a random  linear operator taking values in $\mathcal L_{\mathcal H}$ and $x_1,\ldots,x_n,y_1,\ldots,y_n \in \mathcal H$ for $n \in \mathbb{N}$, we let the distribution of $(\langle Ax_1,y_1 \rangle, \ldots, \langle Ax_n,y_n\rangle)'$ be called the finite dimensional distribution (fdd)  of $A$ with respect to the choice of vectors; if $n=m^2$ for $m \in \mathbb{N}$, $x_j=e_j$ and $y_k=e_k$ for all $j,k\leq m$, and $\{e_j\}_{j=1}^m$ is an orthonormal basis of a subspace $\mathcal H_{m}$, then the associated fdd can be viewed as the distribution of $A$ as a map from $\mathcal H_{m}$ to $\mathcal H_{m}$ 
	If two random bounded linear operators $A$ and $B$ have the same fdd regardless of $n$ and $x_1,\ldots,x_n,y_1,\ldots,y_n \in \mathcal H$, we write $A=_{\fdd}B$. Moreover, if $\{A_j\}_{j\geq 1}$ is a sequence of random elements in $\mathcal L_{\mathcal H}$ and $\|A_j-A\|_{\mathcal L_{\mathcal H}}\to_p 0$ for some $A$ taking its value in $\mathcal L_{\mathcal H}$, we write $A_j \to_{\mathcal L_{\mathcal H}} A$.  
	\section{Cointegrated FTS in Hilbert space} \label{sec:prelim}
	We let $\{\varepsilon_t\}_{t \in \mathbb{Z}}$ be an iid sequence in $L^2_{\mathcal H}$ with a positive definite covariance and let $\{\Phi_j\}_{j \geq 0}$ be a sequence in $\mathcal L_{\mathcal H}$ satisfying $\sum_{j=0}^\infty j \|\Phi_j\|_{\mathcal L_{\mathcal H}} < \infty$. We define $\{\zeta_t\}_{t \in \mathbb{Z}}$ and $\{\eta_t\}_{t \in \mathbb{Z}}$ as follows,
	\begin{equation*}
		\zeta_t = \sum_{j\geq 0}{\Phi}_j \varepsilon_{t-j}, \quad\quad \eta_t = \sum_{j\geq 0} \widetilde{\Phi}_j \varepsilon_{t-j}, \quad\quad t \in \mathbb{Z},
	\end{equation*}
	where  $\widetilde{\Phi}_j = - \sum_{k\geq j+1} \Phi_k$. 
	We consider a sequence $\{X_t\}_{t\geq 0}$ of which first differences $\Delta X_t =X_t - X_{t-1}$ satisfy  the equations $\Delta X_t = \zeta_t$ for $t \geq 1$.  
	{Then by the Phillips-Solo decomposition (and ignoring the initial values $X_0$ and $\eta_0$ that are unimportant in the development of our asymptotic theory), we find that }   
	\begin{equation}
		X_t =   \Phi(1) \varepsilon_t^c + \eta_t, \quad t \geq 1, \label{eqcointeg2}
	\end{equation}
	\bt{where $\varepsilon_t^c = \sum_{s=1}^t \varepsilon_s$ and $\Phi(1) = \sum_{j\geq 0} \Phi_j$.}	
	Note that $\{X_t\}_{t \geq 1}$ is a unit-root-type nonstationary process unless $\Phi(1) = 0$;  however it may have an element $x$ such that $\{\langle X_t,x \rangle\}_{t \geq 1}$ is stationary. Such an element is called a cointegrating vector, and the collection of the cointegrating vectors, denoted by $\mathcal H^S$, is called the cointegrating space. The attractor space, denoted by $\mathcal H^N$, is defined by the orthogonal complement to $\mathcal H^S$.  
	When $X_t$ satisfies \eqref{eqcointeg2}, $\mathcal H^S$ (resp.\ $\mathcal H^N$) is given by $ [\ran \Phi(1)]^\perp$ (resp.\ the closure of  $\ran \Phi(1)$), in this case, $\mathcal H = \mathcal H^N \oplus \mathcal H^S$ and the orthogonal projection $P^N$ onto $\mathcal H^N$ is uniquely defined  (see Proposition 3.3 of \citealp{BSS2017}). Then $P^S=I-P^N$ is the orthogonal projection onto $\mathcal H^S$. Note that $\{P^NX_t\}_{t\geq 1}$ is a unit root process which does not allow any cointegrating vector while $\{P^SX_t\}_{t\geq 1}$ is stationary.  
	
	$\{X_t\}_{t\geq 1}$ in \eqref{eqcointeg2} has zero mean by construction, but which is not essentially required in this paper. A deterministic component such as a nonzero intercept or a linear trend may be allowed in our analysis with a proper modification, which will be discussed in Section \ref{sec:deterministic} of the Supplementary Material.
	
	For our asymptotic analysis, we apply the following additional conditions throughout. 
	\begin{assumpMM}  \label{assumm0} $\mathrm{(i)}$ $\sum_{j=0}^\infty j\|\widetilde{\Phi}_j\|_{\mathcal L_{\mathcal H}} < \infty$ and $\mathrm{(ii)}$ $\varphi = \rank \Phi(1) \in [0, \infty)$.
	\end{assumpMM}
	\noindent Assumption \ref{assumm0}-(i) is employed for mathematical proofs, which may not be restrictive in practice. Assumption \ref{assumm0}-(ii) implies that  $\dim(\mathcal H^N) <\infty$, which is commonly assumed in the recent literature on cointegrated FTS for statistical analysis, see e.g., \cite{Chang2016152} and \cite{SS2019}. 
	\begin{remark} \normalfont
		Finiteness of $\varphi$ seems to be reasonable in many empirical examples; see \citet[Section 5]{Chang2016152} and \citet[Section 5]{SS2019}. Moreover, it is known that functional autoregressive (AR) processes with unit roots (including functional ARMA processes that are considered in e.g.\ \citeauthor{klepsch2017prediction}, \citeyear{klepsch2017prediction}) with compact AR operators always satisfy  Assumption \ref{assumm0}-(ii); see \cite{BS2018}, \cite{Franchi2017b} and \cite{seo_2022}. It is common to assume compactness of AR operators in statistical analysis of such time series, and thus Assumption \ref{assumm0}-(ii)  does not seem to be restrictive in practice. 
	\end{remark}
	The case with $\mathcal H^S =\{0\}$ for each $\varphi$ is uninteresting for our study of cointegrated FTS, and thus such a case is not considered throughout this paper. Moreover, if $\varphi = 0$, $\{X_t\}_{t\geq 1}$ is stationary and thus $\mathcal H^N = \{0\}$; this is also an  uninteresting case except when we examine the null hypothesis of $\varphi =0$ in Section \ref{sec:test}. Thus, our asymptotic theory, to be developed in this section, concerns the case when $\varphi \geq 1$.
	
	Related to a sequence $\{X_t\}_{t \geq 1}$ satisfying Assumption \ref{assumm0}, it will be convenient to introduce additional notation. Note that the identity $I=P^N + P^S$ implies the following operator decomposition: for $A \in \mathcal L_{\mathcal H}$, 
	\begin{equation}
		A = A^{NN} + A^{NS} + A^{SN} + A^{SS}, \quad A^{ij} = P^i A P^j, \quad i \in \{N,S\}, \quad j \in \{N,S\}. \label{eqdecom2}
	\end{equation}
	We also define
	\begin{equation}
		\mathcal E_t^N =  P^N\zeta_t, \quad\quad  \mathcal E_t^S = P^S\eta_t, \quad\quad \mathcal E_t = \mathcal E_t^N + \mathcal E_t^S;
	\end{equation}
	note that each sequence of $\mathcal E_t^N$, $\mathcal E_t^S$ and $\mathcal E_t$ is stationary in $\mathcal H$.
	We then let $\Omega$ (resp.\ $\Gamma$) denote the long run covariance (resp.\ one-sided long run covariance) operator of $\{\mathcal E_t\}_{t \in \mathbb{Z}}$, which are defined by
	\begin{align}
		\Omega =  \Gamma +  \sum_{j\geq 1} \mathbb{E}[\mathcal E_{t-j} \otimes \mathcal E_{t} ] \label{omegat} \quad \text{and}\quad \Gamma = \sum_{j\geq 0}   \mathbb{E}[\mathcal E_{t} \otimes \mathcal E_{t-j}].
	\end{align}	
	$\Omega$ and $\Gamma$ are well defined under Assumption \ref{assumm0}. 
	As in  \eqref{eqdecom2}, we may also decompose those operators as $\Omega= \Omega^{NN}+\Omega^{NS}+ \Omega^{SN}+ \Omega^{SS}$ and $\Gamma = \Gamma^{NN}+\Gamma^{NS}+\Gamma^{SN}+\Gamma^{SS}$. We let  $W = \{W(r)\}_{r \in [0,1]}$ denote Brownian motion in $\mathcal H$ whose covariance operator is given by $\Omega$ (see e.g., 
	\citealp{chen1998}).  We then let  $W^N = P^N W$ and $W^S = P^S W$. Note that $W^N$ and $W^S$ are correlated unless $\Omega^{NS} = \Omega^{SN} = 0$. 
	
	For convenience in our asymptotic analysis, we employ the following assumption: in the assumption below, we let the long-run covariance $\Omega$ of $\mathcal E_t$ be positive definite on $\mathcal H$ and, for $x_1,\ldots,x_k \in \mathcal H$,   $\mathcal E_{k,t} = (\langle \mathcal E_{t}, x_1 \rangle, \ldots, \langle \mathcal E_{t}, x_k \rangle)'$ and $W_k(s) =  (\langle W(s), x_1 \rangle, \ldots, \langle W(s), x_k \rangle)'$.
	\begin{assumpWW} \label{assumW}
		For any $k \geq 1$ and  $x_1,\ldots,x_k \in \mathcal H$, 	$T^{-1}  \sum_{t=1}^T \left(\sum_{s=1}^t \mathcal E_{k,s}\right) \mathcal E_{k,t}'$ converges in distribution to $\int_{0}^1 W_k(s) d W_k(s)' + \sum_{j\geq 0} \mathbb{E}[\mathcal E_{k,t-j}\mathcal E_{k,t}']$.
	\end{assumpWW}
	Appropriate lower-level conditions for the above can be found in \cite{hansen1992convergence}; specifically, 
	Assumption \ref{assumW} is satisfied if (i) $\{T^{-1/2} \sum_{t=1}^{\lfloor Ts \rfloor} \mathcal E_{k,t}\}_{s\in [0,1]}$  converges weakly in the Skorohod topology to $W_k$, (ii) $\sup_{t\geq 1} \|\mathcal E_t\|^a < \infty$, and (iii) $\{\mathcal E_{t}\}_{t\geq 1}$ is  strongly mixing of size $-ab(a-b)$ for some $a>b>2$. Among these conditions, the first is implied by Assumption \ref{assumm0} (see Lemma \ref{lem1} in the Supplementary Material).
	

	\section{FPCA of cointegrated FTS and asymptotic results} \label{sec:fpca}

	\subsection{Ordinary FPCA of cointegrated FTS}	\label{sec:ofpca}
	We now consider estimation of $P^N$ and $P^S$ from observations $\{X_t\}_{t=1}^T$. Throughout this section, we assume that $\varphi=\dim(\mathcal H^N)$ is known. Of course, this is not a realistic assumption in most applications. However, we already have a few available methods to determine $\varphi$ such as those in \cite{Chang2016152},  \cite{SS2019}, and \cite{LRS}. In addition to those, it will be shown in Section \ref{sec:test} that the asymptotic results to be given in this section  lead to a novel FPCA-based testing procedure to determine $\varphi$.  

	The sample covariance operator $\widehat{C}$ of $\{X_t\}_{t=1}^T$  is the random  operator given by 
	$\widehat{C} = T^{-1}\sum_{t=1}^T X_t \otimes X_t$. 
	Broadly speaking, FPCA reduces to solving the following eigenvalue problem associated with $\widehat{C}$, 
	\begin{equation} \label{eigenp}
		\widehat{C} \hat{v}_{j} = \hat{\lambda}_{j} \hat{v}_{j}, \quad j=1,2,\ldots,
	\end{equation}
	where $\hat{\lambda}_{1} \geq \ldots \geq  \hat{\lambda}_{T} \geq 0 = \hat{\lambda}_{T+1}= \hat{\lambda}_{T+2}=\ldots$ and $\{\hat{v}_{j}\}_{j \geq 1}$ constitutes an orthonormal basis of $\mathcal H$. From the estimated eigenvectors $\{\hat{v}_j\}_{j \geq 1}$, we construct our preliminary estimators of $P^N$ and $P^S$ as follows, 
	\begin{equation}
		\widehat{P}^N_\varphi = \sum_{j=1}^\varphi \hat{v}_{j} \otimes \hat{v}_{j}, \quad \quad \widehat{P}^S_\varphi = I - \widehat{P}^N_\varphi.   \label{eqprelim} 
	\end{equation} 
	Note  that $\varphi$ as a subscript in \eqref{eqprelim} denotes the the number of eigenvectors used to construct  $\widehat{P}^N_{\varphi}$; even if keeping such a subscript may increase notational complexity, it will be eventually helpful to avoid potential confusion when we need to consider  $\widehat{P}^N_{\varphi_0}$, the projection onto $\spn\{\hat{v}_j\}_{j=1}^{\varphi_0}$, for a hypothesized value $\varphi_0$ in Section \ref{sec:test}. To investigate the asymptotic properties of the preliminary estimators, we may establish the limiting behavior of either of $\widehat{P}^N_\varphi-P^N$ or $\widehat{P}^S_\varphi-P^S$ (one is simply given by the negative of the other). The following theorem deals with the former case: we hereafter write $\int A(r)$ to denote $\int_0^1 A(r)dr$ for any operator- or vector-valued function $A$ defined on $[0,1]$. 
	
	\begin{theorem} \label{prop1} Suppose that Assumptions \ref{assumm0} and \ref{assumW} hold with $\varphi \geq 1$.  Then   
		\begin{align*}
			T(\widehat{P}^N_\varphi - P^N) \quad  &\to_{\mathcal L_{\mathcal H}} \quad  \mathrm{F} + \mathrm{F}^\ast, 
		\end{align*}
		where $\mathrm{F}=_{\fdd} \left(\int W^N(r) \otimes W^N(r)\right)^\dag\left(\int d W^S(r)\otimes W^N(r)   +  \Gamma^{NS}\right)$.
	\end{theorem}
	\citet[Proposition 2.1]{Chang2016152} earlier established that 	$\|\widehat{P}^N_{\varphi}-  P^N\|_{\mathcal L_{\mathcal H}} = \mathrm{O}_p(T^{-1})$. Theorem \ref{prop1} extends their result by providing a more detailed limiting behavior of $\widehat{P}^N_\varphi$. 
	Even if Theorem \ref{prop1} implies consistency of  $\widehat{P}^N_\varphi$,
	it does not facilitate a direct asymptotic inference since $\mathrm{F}$ depends on  nuisance parameters.  We particularly focus on $\Gamma^{NS}$ and $\Omega^{NS}$; as shown in Theorem \ref{prop1}, $\Gamma^{NS} \neq 0$ makes the limiting random operator $\mathrm{F}+\mathrm{F}^\ast$ be not centered at zero, and  $\Omega^{NS} \neq 0$ makes  $W^N$ and $W^S$ be correlated. Therefore, the asymptotic limit of any statistic based on $\widehat{P}^N_\varphi$ or $\widehat{P}^S_\varphi$ is, if it exists, also dependent on $\Gamma^{NS}$ and $\Omega^{NS}$  in general. As will be shown throughout this paper, elimination of such dependence is the most important step to our modified FPCA and statistical tests based on it.



	\subsection{Modified FPCA for cointegrated FTS}  \label{sec:mfpca}
	We now propose a modification of the ordinary FPCA, which not only provides an asymptotically more efficient estimator of $P^N$ (or $P^S$) but establishes the foundation for  our  statistical tests that will be developed in Section \ref{sec:test}. Our modified FPCA may be viewed as an adjustment of the ordinary FPCA for cointegrated FTS rather than an alternative methodology; we actually take full advantage of the asymptotic properties of the preliminary estimators given in Section \ref{sec:ofpca}. 
	To introduce our methodology, we first define 
	\begin{equation}
		Z_{\varphi,t} = \widehat{P}^N_\varphi \Delta X_t + \widehat{P}^S_\varphi  X_t, \quad \quad t = 1,\ldots,T,  \label{defitem1}
	\end{equation}
	where $\widehat{P}^N_\varphi$ and $\widehat{P}^S_\varphi$ are the preliminary estimators obtained from the ordinary FPCA in Section \ref{sec:fpca}. Then the sample counterparts of $\Omega$ and $\Gamma$ given in \eqref{omegat}  are defined by
	\begin{align}
		\hspace{-0.2em}\widehat{\Omega}_\varphi\hspace{-0.1em} =  \hspace{-0.2em}\widehat{\Gamma}_{\varphi}  + \hspace{-0.1em}  \frac{1}{T}\sum_{s=1}^{T-1}  \hspace{-0.1em} \mathrm{k}\left(\frac{s}{h}\right) \hspace{-0.2em} \sum_{t=s+1}^T \hspace{-0.2em}\hspace{-0.1em}Z_{\varphi,t-s}\hspace{-0.1em} \otimes \hspace{-0.1em}Z_{\varphi,t}  \hspace{-0.1em} \quad \text{and} \quad  \hspace{-0.2em}\widehat{\Gamma}_{\varphi}  \hspace{-0.1em} =   \hspace{-0.1em}  \frac{1}{T}\sum_{s=0}^{T-1}  \hspace{-0.1em} \mathrm{k}\left(\frac{s}{h}\right) \hspace{-0.2em} \sum_{t=s+1}^T \hspace{-0.2em} Z_{\varphi,t}\hspace{-0.1em} \otimes \hspace{-0.1em} Z_{\varphi,t-s},	 \label{omegae}
	\end{align}
	where $\mathrm{k}(\cdot)$ (resp.\ $h$) is the kernel (resp.\ the bandwidth) satisfying the following  conditions: \begin{assumpK}\label{assum2} $\mathrm{(i)}$ $\mathrm{k}(0) = 1$, $\mathrm{k}(u) = 0$ if $u > \kappa$ with some $\kappa>0$,  $\mathrm{k}$ is continuous on $[0,\kappa]$, and $\mathrm{(ii)}$  $h \to \infty$ and $h/T \to 0$ as $T \to \infty$.
	\end{assumpK}
	\noindent  
	\bt{Note that the one-sided kernel $\mathrm{k}(\cdot)$ in Assumption \ref{assum2} satisfies the identical conditions given by  \cite{horvath2013estimation} for estimation of the long run covariance of a stationary FTS.} 
	From the identity $I = \widehat{P}^N_\varphi + \widehat{P}^S_\varphi$, we may decompose $\widehat{\Omega}_\varphi$ and $\widehat{\Gamma}_\varphi$ as in \eqref{eqdecom2}, i.e., 
	$\widehat{\Omega}_\varphi = \widehat{\Omega}^{NN}_\varphi +  \widehat{\Omega}^{NS}_\varphi + \widehat{\Omega}^{SN}_\varphi + \widehat{\Omega}^{SS}_\varphi$ and $\widehat{\Gamma}_\varphi = \widehat{\Gamma}^{NN}_\varphi + \widehat{\Gamma}^{NS}_\varphi + \widehat{\Gamma}^{SN}_\varphi + \widehat{\Gamma}^{SS}_\varphi$.
	We then define a modified variable ${X}_{\varphi,t}$ as follows, 
	\begin{equation} \label{modvar}
		{X}_{\varphi,t} = X_t - \widehat{\Omega}^{SN}_\varphi\,\, \left(\widehat{\Omega}^{NN}_{\varphi}{\mid_{\varphi}^\dagger}\right) \,\, \widehat{P}^N_\varphi \Delta X_t, \quad \quad t = 1,\ldots,T.  
	\end{equation}
	The $\varphi$-regularized inverse $\widehat{\Omega}^{NN}_{\varphi}{\mid_{\varphi}^\dagger}$ in \eqref{modvar} can  easily be  computed from eigendecomposition of $\widehat{\Omega}^{NN}$ once we know $\varphi$. 
	Let $\widehat{C}_\varphi$ be the sample covariance operator  of $\{{X}_{\varphi,t}\}_{t=1}^T$, i.e.,  $\widehat{C}_\varphi = T^{-1}\sum_{t=1}^T {X}_{\varphi,t} \otimes {X}_{\varphi,t}$.
	Roughly speaking, replacing $\widehat{C}$ with $\widehat{C}_{\varphi}$ in  \eqref{eigenp}   has the effect of transforming $W^S$, appearing in the expression of $\mathrm{F}$ in Theorem \ref{prop1}, into another Brownian motion that is independent of $W^N$; however, this does not resolve the issue that the center of the limiting operator depends on nuisance parameters. We therefore need a further adjustment, which is achieved by considering the following modified eigenvalue problem:
	\begin{equation} \label{eigenp2a}
		\left(\widehat{C}_{\varphi} - \widehat{\Upsilon}_{\varphi} - \widehat{\Upsilon}_{\varphi}^\ast\right)  \hat{w}_{j} =  \hat{\mu}_{j} \hat{w}_{j}, \quad j=1,2,\ldots,
	\end{equation}
	where 
	\begin{equation} \label{eigenp21}
		\widehat{\Upsilon}_\varphi = \widehat{\Gamma}^{NS}_\varphi - \widehat{\Gamma}^{NN}_\varphi\,\, \left(\widehat{\Omega}^{NN}_{\varphi}{\mid_{\varphi}^\dagger}\right)\,\,\widehat{\Omega}^{NS}_\varphi.   
	\end{equation}
	Note that the operator $\widehat{C}_{\varphi} - \widehat{\Upsilon}_{\varphi} - \widehat{\Upsilon}_{\varphi}^\ast$ is self-adjoint, so \eqref{eigenp2a} does not produce any complex eigenvalues. From the estimated eigenvectors $\{\hat{w}_j\}_{j \geq 1}$, we construct our proposed estimators of $P^N$ and $P^S$ as follows: 
	\begin{equation}
		\widehat{\Pi}^N_\varphi = \sum_{j=1}^\varphi  \hat{w}_{j} \otimes  \hat{w}_{j},  \quad \quad  \widehat{\Pi}^S_\varphi = I -   \widehat{\Pi}^N_\varphi. \label{eqmod01}
	\end{equation} 
	Note that $\widehat{\Pi}^N_\varphi$ and $\widehat{\Pi}^S_\varphi$ are simply obtained by replacing $\hat{v}_j$ with $\hat{w}_j$ in \eqref{eqprelim}. The asymptotic properties of $\widehat{\Pi}^N_\varphi$ or $\widehat{\Pi}^S_\varphi$ are described by the following theorem.

	\begin{theorem} \label{prop2} Suppose that Assumptions \ref{assumm0}, \ref{assumW} and \ref{assum2} hold  with $\varphi \geq 1$. Then 	\begin{align*}
			T(\widehat{\Pi}^N_\varphi-P^N) \quad  \to_{\mathcal L_{\mathcal H}} \quad  \mathrm{G} + \mathrm{G}^\ast,
		\end{align*}
		\sloppy{where $\mathrm{G} \hspace{-0.1em}=_{\fdd} \hspace{-0.1em}\left(\int W^N(r) \hspace{-0.1em}\otimes\hspace{-0.1em} W^N(r)\right)^\dag \left(\int  d {W}^{S.N}(r)\hspace{-0.1em}\otimes\hspace{-0.1em}  W^N(r)\right)$,  ${W}^{S.N}(r)\hspace{-0.1em} = \hspace{-0.1em}W^S(r)\hspace{-0.1em} -\hspace{-0.1em} \Omega^{SN}(\Omega^{NN})^\dag W^N(r)$, and ${W}^{S.N}$ is independent of $W^N$}.
	\end{theorem}
	
	The asymptotic limit of the proposed estimator $\widehat{\Pi}^N_\varphi$ is of a more convenient form than that of the preliminary estimator based on the ordinary FPCA. First note that the limiting operator $\mathrm{G}+\mathrm{G}^\ast$ is now centered at zero, whereas $\mathrm{F}+\mathrm{F}^\ast$ in Theorem \ref{prop1} is not so in general due to $\Gamma^{NS}$. In addition, $\mathrm{G}$ is characterized by two independent Brownian motions while $\mathrm{F}$ is not so in general due to $\Omega^{NS}$. These properties of  $\widehat{\Pi}^N_\varphi$ not only help us develop statistical tests for examining various hypotheses about cointegration in Section \ref{sec:test}, but also make   $\widehat{\Pi}^N_\varphi$ asymptotically more efficient than  $\widehat{P}^N_\varphi$ in a certain sense, see Remark \ref{rem2}.

	\begin{remark}\normalfont \label{rem2}
		\sloppy{For any $k$, $x = (x_1,\ldots,x_k) \in \mathcal H^k$ and $y=(y_1,\ldots,y_k)\in \mathcal H^k$, let $\widehat{P}(k,x,y) = (\langle  TP^N\widehat{P}^S_\varphi x_1,y_1 \rangle,  \ldots, \langle TP^N\widehat{P}^S_\varphi x_k,y_k \rangle)'$ and $\widehat{\Pi}(k,x,y) = (\langle TP^N\widehat{\Pi}^S_\varphi x_1,y_1 \rangle,  \ldots, \langle  TP^N\widehat{\Pi}^S_\varphi x_k,y_k \rangle)'$. From similar arguments used in \citet[Theorem 3.1]{saikkonen1991asymptotically} and \citet[Section 2]{harris1997principal}, the following can be shown:  for any choice of $k$, $x$, $y$, and $\Theta \subset \mathbb{R}^{k}$ that is  convex and symmetric around the origin,}	  
		$	\lim_{T\to \infty} \text{Prob.}\{ \widehat{\Pi}(k,x,y) \in \Theta\}  \geq  \lim_{T\to \infty} \text{Prob.}\{ \widehat{P}(k,x,y) \in \Theta \}$ 
		and the equality does not hold in general unless $\Omega^{NS} = \Gamma^{NS} = 0$. A detailed discussion is given at the end of Section \ref{addremark} in the Supplementary Material. 
		This implies that the asymptotic  distribution of $\widehat{\Pi}(k,x,y) $ is more concentrated at zero than  that of $ \widehat{P}(k,x,y)$.	A similar result can be shown for $P^S\widehat{\Pi}_{\varphi}^N$ and $P^S\widehat{P}_{\varphi}^N$. In this sense, we say that $\widehat{\Pi}_{\varphi}^N$ is more asymptotically efficient than $\widehat{P}_{\varphi}^N$. 
	\end{remark}


	\begin{remark}\normalfont \label{remfin}
		Our methodology can be applied in the case where $\dim(\mathcal H) < \infty$ without further adjustment, which is differentiated from that the existing methods (\citealp{SS2019}, \citealp{Chang2016152}, and \citealp{LRS}) are only discussed in a specific infinite dimensional Hilbert space setting. In this finite dimensional case, the modified FPCA is related to the PCA methodology of \cite{harris1997principal}, but the latter cannot directly be applied to an infinite dimensional setting. For more details, see Section \ref{sec:harris} of the Supplementary Material.   
	\end{remark}
	
	\begin{remark}\normalfont \label{remaddnew}
		\bt{We assumed that $\varphi=\dim(\mathcal H^N)$ is known in this section, but this is not generally possible for practitioners. Of course, $\varphi$ can be replaced by its consistent estimator  and this does not affect the asymptotic result given by Theorem \ref{prop2}. Nevertheless, as long as this replacement is necessary, our estimator is not guaranteed to be better (in the sense of Remark \ref{rem2}) than the ordinary FPCA-based estimator in finite samples; more detailed discussion on this issue is beyond the scope of the present article and hence we leave it as a future work. Despite this limitation,  Theorem \ref{prop2} by itself can be a basis for statistical inference on cointegrated FTS, and one such example is our test to be presented in the next section.}
	\end{remark}
	
	
	\section{Applications: statistical tests based on the modified FPCA}	 \label{sec:test} 
	We develop statistical tests to examine various hypotheses about $\mathcal H^N$ or $\mathcal H^S$ based on the asymptotic results given in Section \ref{sec:mfpca}. As in the previous sections, we  focus on the case without deterministic terms; the discussion is extended to allow a deterministic component in Section \ref{sec:deterministic} of the Supplementary Material. 
	


	\subsection{FPCA-based test for the dimension of $\mathcal H^N$} \label{infer:coranktest}
	In the previous sections, we need prior knowledge of $\varphi = \dim(\mathcal H^N)$. We here provide a novel FPCA-based test that can be applied sequentially to determine $\varphi$. 
	Consider the following null and alternative hypotheses, 
	\begin{align} \label{hypo0}
		H_0 : \dim(\mathcal H^N)  = \varphi_0 \quad \text{against} \quad 	H_1 : \dim(\mathcal H^N)  > \varphi_0,     
	\end{align}
	for $\varphi_0 \geq 0$.	
	It is worth mentioning that stationarity of $\{X_t\}_{t\geq 1}$, i.e., $\varphi_0 = 0$, can be examined by the test to be developed. This means that our test can be used as an alternative to the existing tests of stationarity of FTS proposed by e.g.\ \cite{horvath2014test}, \cite{kokoszka2016kpss}, and \cite{aue2017testing}. 
	It should also be noted that our selection of hypotheses in \eqref{hypo0} is the opposite of that used for the existing tests proposed by \cite{Chang2016152} and \cite{SS2019} in the sense that the alternative hypothesis in those tests is set to $H_1 : \dim(\mathcal H^N) < \varphi_0$. Due to this difference, our test has its own advantage especially when it is sequentially applied to estimate $\varphi$; this will be more detailed in Remark \ref{rem:seq}.  

	For each $\varphi_0$, we  define $\widehat{P}^N_{\varphi_0}$ and $\widehat{P}^S_{\varphi_0}$ as in \eqref{eqprelim} from the eigenproblem \eqref{eigenp}; if $\varphi_0 = 0$, $\widehat{P}^N_{\varphi_0}$ is set to zero.  Given $\widehat{P}^N_{\varphi_0}$ and $\widehat{P}^S_{\varphi_0}$, define $Z_{\varphi_0,t}$, $\widehat{\Omega}_{\varphi_0}$,  $\widehat{\Gamma}_{\varphi_0}$, $\widehat{\Upsilon}_{\varphi_0}$, ${X}_{\varphi_0,t}$, and  $\widehat{C}_{\varphi_0}$ as in Section \ref{sec:fpca}; see \eqref{defitem1}-\eqref{eigenp21}. 
	
	We expect from Theorem \ref{prop1} that $\{\widehat{P}^S_{\varphi_0}X_t\}_{t\geq 1}$ will behave as a stationary process if $\varphi_0 = \varphi$ (since $\|\widehat{P}^S_{\varphi_0} - P^S\|_{\mathcal L_{\mathcal H}} = \mathrm{O}_p(T^{-1})$),  and  the eigenvectors $\{\hat{v}_{j}\}_{j=\varphi_0+1}^{\varphi}$ are asymptotically included in $\mathcal H^N$ if $\varphi_0 < \varphi$. In the latter case,  $(\langle X_t, \hat{v}_{\varphi_0+1} \rangle,\ldots, \langle X_t, \hat{v}_{\varphi} \rangle)'$ is expected to behave as a unit root process, hence $\{\widehat{P}^S_{\varphi_0}X_t\}_{t\geq 1}$ will not be stationary. Based on this idea, it  may be reasonable to ask if we can distinguish the correct hypothesis from incorrect ones specifying $\varphi_0 < \varphi$ by examining stationarity of $\{\widehat{P}^S_{\varphi_0}X_t\}_{t\geq 1}$. 
	
	To simplify the discussion, 	we focus on the time series $\{\langle X_t, \hat{v}_{\varphi_0+1} \rangle \}_{t\geq 1}$, which is expected to be stationary under $H_0$. 
	To examine stationarity of $\{\langle X_t, \hat{v}_{\varphi_0+1} \rangle \}_{t\geq 1}$, we consider the following test statistic, 
	\begin{align}
		\frac{1}{T^2} \sum_{t=1}^T \sum_{s=1}^t \langle X_s, \hat{v}_{\varphi_0+1} \rangle^2 / \LRV(\langle X_t, \hat{v}_{\varphi_0+1}\rangle),  \label{eqkpssorigin}  
	\end{align}  
	where $\LRV(\langle X_t, \hat{v}_{\varphi_0+1}\rangle)=\frac{1}{T}\sum_{s=-T+1}^{T-1}\mathrm{k}({|s|}/{h})\sum_{t=|s|+1}^T  \langle X_t, \hat{v}_{\varphi_0+1} \rangle \langle X_{t-|s|}, \hat{v}_{\varphi_0+1} \rangle$,
		and  $\mathrm{k}(\cdot)$ and $h$ satisfy Assumption \ref{assum2}. Note that  the test statistic is given by the ratio of the sum of squared partial sums of the univariate time series ($\{\langle X_t, \hat{v}_{\varphi_0+1} \rangle\}_{t=1}^T$) to its sample long-run variance ($\LRV(\langle X_t, \hat{v}_{\varphi_0+1}\rangle)$). This is similar to the well-known KPSS test for examining the null hypothesis of stationarity; however, it should be noted that \eqref{eqkpssorigin} is not identical to the original LM statistic proposed by \cite{kwiatkowski1992testing} unless $\{ \langle X_t, \hat{v}_{\varphi_0+1} \rangle\}_{t=1}^T$ is demeaned (see Appendix of their paper).
		If $\Omega^{NS} = \Gamma^{NS} = 0$, the asymptotic null distribution of \eqref{eqkpssorigin} does not depend on any nuisance parameters and is  given by a functional of two independent standard Brownian motions. We thus may assess the plausibility of the null hypothesis with no significant difficulty. 
		In general cases where $\Omega^{NS}=\Gamma^{NS}=0$ is not satisfied, the limiting distribution of \eqref{eqkpssorigin} depends on both of  $\Omega^{NS}$ and $\Gamma^{NS}$, which are unknown in practice (see Theorem \ref{prop3a} and its proof given in Section \ref{appproof1} of the Supplementary Material).   As expected from  Theorem \ref{prop1}, this issue is linked closely to the fact that the asymptotic limit of $\widehat{P}^S_{\varphi_0}$ depends on $\Omega^{NS}$ and $\Gamma^{NS}$.

		The assumption that $\Omega^{NS} = \Gamma^{NS} = 0$ is too restrictive in modeling cointegrated FTS, hence a naive use of \eqref{eqkpssorigin} in practice to examine \eqref{hypo0} should be limited to very special circumstances. However, using the asymptotic results developed for our modified FPCA, the test statistic \eqref{eqkpssorigin} can be  modified to have an asymptotic null distribution that is free of nuisance parameters in general cases.  Consider the following eigenvalue problem:
		\begin{equation} \label{eigenptest}
			\left(\widehat{C}_{\varphi_0} - \widehat{\Upsilon}_{\varphi_0} - \widehat{\Upsilon}_{\varphi_0}^\ast\right) \hat{w}_{j} =\hat{\mu}_{j}  \hat{w}_{j} , \quad j=1,2,\ldots.
		\end{equation}
		We then  replace   $X_t$ and $\hat{v}_{\varphi_0+1}$  in \eqref{eqkpssorigin} with $X_{\varphi_0,t}$ and $\hat{w}_{\varphi_0+1}$, respectively, and obtain the following statistic: 	
		\begin{equation}
			\widehat{Q} =	\frac{1}{T^2} \sum_{t=1}^T \sum_{s=1}^t \langle {X}_{\varphi_0,s}, \hat{w}_{\varphi_0+1} \rangle^2 / \LRV(\langle {X}_{\varphi_0,t}, \hat{w}_{\varphi_0+1}\rangle).  \label{eqkpssorigin2}  
		\end{equation} 
		There is a big gain from this simple replacement: different from \eqref{eqkpssorigin}, the asymptotic null distribution of  \eqref{eqkpssorigin2} does not depend on any nuisance parameters and is given by a functional of two independent standard Brownian motions without requiring $\Omega^{NS}=\Gamma^{NS}=0$; this is, of course, closely related to the asymptotic result given by Theorem \ref{prop2}.  We thus may asymptotically assess the plausibility of the null hypothesis relative to the alternative as follows: \bt{below, we let   $\mathsf{B}$   and $\mathsf{W}$ denote \textit{independent} $\varphi_0$-dimensional and one-dimensional standard Brownian motions, respectively.}
		\begin{proposition}\label{propkpss}
			Suppose that Assumptions  \ref{assumm0}, \ref{assumW} and \ref{assum2} hold. 
			\bt{Then $\widehat{Q}\to_d \int \mathsf{V}^2(s)ds$ under $H_0$ of \eqref{hypo0} and $\widehat{Q}\to_p \infty$ under $H_1$ of \eqref{hypo0}, where $V(r)=\mathsf{W}(r) - \int d\mathsf{W}(s) \mathsf{B}(s)' \left(\int \mathsf{B}(s)\mathsf{B}(s)' \right)^{-1}\int_{0}^{r} \mathsf{B}(s)$. The second term in the expression of $V(r)$ is regarded as zero if $\varphi_0 = 0$.}
		\end{proposition} 
		
		
		\bt{If $\varphi_0 = 0$ (and thus $X_{\varphi,t} = X_t$,  $\widehat{C}_{\varphi_0}=\widehat{C}$, and  $\widehat{\Gamma}_{\varphi_0} =\widehat{\Gamma}_{\varphi_0}^\ast = 0$), the above test is similar but not identical to the test of \cite{horvath2014test}; their test is based on computation of the eigenelements of the sample long-run covariance operator of $X_t$ unlike ours is based on those of $\widehat{C}$.} The asymptotic test given in Proposition \ref{propkpss} 
		in fact corresponds to a special case of the general version of our FPCA-based test, which is discussed in the Supplementary Material (Section \ref{sec_general_test}); see also Section \ref{sec:deterministic2} for our extension of the test for the case where there is a nonzero intercept and/or a linear trend.

		\begin{remark} \normalfont \label{remseq}
			To determine $\varphi$, we may apply our test for $\varphi_0 = 0,1,\ldots$ sequentially. Let $\hat{\varphi}$ denote the value under $H_0$ that is not rejected for the first time at a fixed significance level  $\alpha$. We deduce from Theorem \ref{prop3a} that $\text{Prob.\ }\{ \hat{\varphi} = \varphi\} \to  1- \alpha$ and $\text{Prob.\ }\{ \hat{\varphi} < \varphi\}  \to\  0$. Note that,  even if the sequential procedure requires multiple applications of the proposed test, the correct asymptotic size, $\alpha$, is guaranteed without any adjustments; this property is shared by the existing sequential procedures developed in a Euclidean/Hilbert space setting (see e.g.,  \citeauthor{Johansen1996}, \citeyear{Johansen1996}; \citeauthor{nyblom2000tests}, \citeyear{nyblom2000tests}; \citeauthor{Chang2016152}, \citeyear{Chang2016152}; 
			\citeauthor{SS2019}, \citeyear{SS2019}). If the significance level is chosen such that $\alpha \to 0$ as $T\to \infty$ then, $\text{Prob.\ }\{ \hat{\varphi} = \varphi\} \to  1$.
			
		\end{remark}
		
		\begin{remark} \label{rem:seq} \normalfont
			In \cite{Chang2016152} and \cite{SS2019}, $\varphi$ is determined by testing the following hypotheses: for a pre-specified positive integer $\varphi_{\max}$ and $\varphi_0 = \varphi_{\max},\varphi_{\max}-1,\ldots,1$, 
			\begin{align*} 
				H_0 : \dim(\mathcal H^N)  = \varphi_0 \quad \text{against} \quad 	H_1 : \dim(\mathcal H^N)  < \varphi_0,     
			\end{align*}
			sequentially until $H_0$ is not rejected for the first time. Note that we need a prior information on an upper bound of $\varphi$, i.e., $\varphi_{\max}$, for this type of procedure. However, in an infinite dimensional setting, there is no natural upper bound of $\dim(\mathcal H^N)$. On the other hand, our sequential procedure described in Remark \ref{remseq} does not require such a prior information; 
			the procedure first examines the null hypothesis $H_0 : \dim(\mathcal H^N) = 0$, and $0$ is the minimal possible value of $\dim(\mathcal H^N)$.  \bt{Even with this difference resulting from a different formulation of the alternative hypothesis, our method can be compared with the existing ones as a way to estimate $\varphi$. From our simulation study we find that our procedure seems to perform comparably with  the recent method proposed by  \cite{SS2019} and tends to work better if the stationary component is persistent; see Section \ref{sec:monte} and Table \ref{tabcorank2determination}.}	If $\varphi  \geq 1$ is known in advance,  $\varphi$ can also be estimated by the eigenvalue ratio criterion proposed by  \cite{LRS} under their assumptions. Our procedure given in Remark \ref{remseq} is significantly differentiated from any of these existing procedures, and thus ours can complement those in practice. 
		\end{remark}
		


		\subsection{Tests of hypotheses about cointegration}\label{sec:inferenceatt}
		Practitioners may be interested in testing various hypotheses about $\mathcal H^N$ or $\mathcal H^S$. For example, we may want to test if a specific element $x_0$ is included in $\mathcal H^N$ or the span of a specified set of vectors contains $\mathcal H^N$. More generally, we here consider testing the following hypotheses: for a specified subspace $\mathcal M$,
		\begin{align} 
			&	H_0 : \mathcal M \subset \mathcal H^N, \,\, \text{or equivalently,}\,\,  \mathcal M^\perp \supset \mathcal H^S \quad \text{against} \quad  	H_1 :  H_0 \text{ is not true}, \label{hptest1}\\
			&H_0 : \mathcal M \supset \mathcal H^N, \,\, \text{or equivalently,}\,\,  \mathcal M^\perp \subset \mathcal H^S \quad \text{against} \quad  	H_1 :  H_0 \text{ is not true}. \label{hptest2} 
		\end{align}
		
		
		Let $P^{{\mathcal M}}$ denote the projection onto ${{\mathcal M}}$ and assume that $\varphi$ is known; in practice, 
		we may apply any statistical method discussed in Remarks \ref{remseq} and \ref{rem:seq} to determine  $\varphi$. 
		In this case, \eqref{hptest1} and \eqref{hptest2} can be tested by examining the dimension of the attractor space associated with the residuals $\{(I-P^{{\mathcal M}})X_t\}_{t\geq 1}$. 
		
		Specifically, let $\widehat{Q}_{\mathcal M}$ be the test statistic computed as in \eqref{eqkpssorigin2} from $\{(I-P^{{\mathcal M}})X_t\}_{t=1}^T$ for $\varphi_0 = \varphi-\dim(\mathcal M)$ (resp.\ $\varphi_0 = 0$)  if we examine \eqref{hptest1} (resp.\  \eqref{hptest2}). Then we may deduce from Proposition \ref{propkpss} that, under $H_0$ of \eqref{hptest1}, $\widehat{Q}_{\mathcal M}  \to_d  \int \mathsf{W}^2(s)ds$ while it diverges to infinity under $H_1$ of \eqref{hptest1}.

		\begin{remark} \label{remadd}\normalfont
			We assess the plausibility of $H_0$ in \eqref{hptest1} or \eqref{hptest2}  by checking if the dimension of the attractor space associated with $\{(I-P^{{\mathcal M}})X_t\}_{t\geq1}$ is higher than that is implied by $H_0$. This is possible since our test is consistent against the alternative hypotheses with higher dimensional attractor spaces. 	Unlike our proposed test, the top-down tests proposed by \cite{Chang2016152} and \cite{SS2019} cannot be used in this way.  
		\end{remark}

		\subsection{Numerical studies on the proposed tests: a brief summary}\label{sec_numeri}  
		\bt{In  Sections \ref{sec:monte}-\ref{app:addsim}  of the Supplementary Material, we examine the finite sample properties of our proposed tests for \eqref{hypo0}, \eqref{hptest1} and \eqref{hptest2}. The tests seem to perform reasonably but slightly over-reject the null hypothesis when the stationary component is persistent. We also compare the finite sample performance of our test as a way to estimate $\varphi=\dim(\mathcal H^N)$ (Remark \ref{remseq}) with the recent testing procedure of \cite{SS2019}. We found that, overall, ours  performs comparably with theirs in general; moreover, in our simulation setting with a persistent stationary component, ours tends to outperform theirs. The performance of our testing procedure seems to be dependent on the choice of $h$  while the method proposed by \cite{SS2019} does not require such a bandwidth choice (since theirs does not require any long-run covariance estimation). Despite this disadvantage, we may conclude from the simulation evidence that our testing procedure can be an attractive alternative or complement to theirs. Of course, as mentioned in Remark \ref{remadd}, the proposed test in the present paper can also be used to examine various hypotheses on cointegration from a relevant residual FTS, which may be another advantage of ours over the test of \cite{SS2019}. In addition to the simulation studies, we in Sections \ref{sec:emp} and \ref{sec:emp2} illustrate our methodology with two empirical examples: U.S.\ age-specific employment rates and monthly earning densities.} 
	

	\section{Concluding remarks : some cautions and future direction}\label{sec_conclusion}
	This note adds some novel asymptotic results for the FPCA-based estimator of $P_N$ or $P_S$ of cointegrated FTS. We then propose a modification of FPCA for statistical inference on the cointegrating behavior. Some cautions need to be made for practitioners. First, in order for the modified FPCA to work well, we need to accurately estimate $P_N$ first; however, as remarked by \cite{SS2019}, it is not easy particularly if $\dim(\mathcal H^N)$ is large but the sample size is small. Moreover, this note does not address in detail on the choice of the bandwidth parameter $h$ used to estimate long-run covariance operators of some residual time series considered in the note. Given that all such time series are expected to be stationary, the method  proposed by \cite{rice2017plug} may be used, but this issue obviously requires a further investigation.  Despite these limitations, our proposed methodology gives us some useful results, such as statistical tests on $\dim(\mathcal H^N)$ and cointegrating properties. \bt{To put it in more detail, our test as a way to estimate $\dim(\mathcal H^N)$ performs comparably with the existing methods (it is found to perform better for some simulation setting considered in Section \ref{sec:monte}), and the test by itself can be used to examine various hypotheses on cointegration (see Remark \ref{remadd}); see also Remarks \ref{rem2} and \ref{remfin}.} 
	The theoretical results in this paper may be used for more sophisticated statistical models requiring dimension-reduction of a FTS exhibiting unit-root-like behavior; a potential example may be the cointegrating regression model involving functional regressors. Provided that practitioners want to employ FPCA-based dimension-reduction, which has been widely used in various contexts, our results will be helpful to derive detailed asymptotic properties of estimators for such models. 


	\appendix 
	\begin{center}
		\Large Supplementary Material on ``Functional principal component analysis for cointegrated functional time series''
	\end{center}
	\begin{abstract}
	This supplementary material consists of two parts. In Part I, we provide some relevant extensions of the theoretical results given in the main article. In Part II, we provide mathematical justifications for the results given in the main article and Part I.
	
\end{abstract}
	\subsection*{Outline}
	Part I of this supplementary material consists of four sections (Sections \ref{sec:harris}-\ref{sec:numerical}).  Section \ref{sec:harris} shows how our modified FPCA is related to the PCA methodology proposed by \cite{harris1997principal}. Section \ref{sec_general_test} provides a more general version of the statistical test given in Section \ref{infer:coranktest} and  Section \ref{sec:deterministic} extends the theoretical results given in the main article (and also those in  Section \ref{sec_general_test}).  In Section \ref{sec:numerical}, we provide simulation results (to see the finite sample performances of the proposed tests) and empirical applications (to illustrate our methodology for practitioners). In Part II, we provide mathematical justifications for the theoretical results given in the main article and Part I. Section \ref{appintro0} summarizes mathematical preliminaries and Section \ref{appproof} provides proofs.


	\bigskip
	\noindent \textbf{\Large{Part I : Supplementary results and numerical studies}}
	\section{Modified FPCA in a finite dimensional setting} \label{sec:harris}
	We consider a special case when $\dim(\mathcal H)<\infty$ and the minimum eigenvalue of $\mathbb{E} [\mathcal E_t \otimes \mathcal E_t]$ is strictly positive. Even if our modified FPCA can be applied without any adjustment in this case, we here provide a way to obtain an estimator whose asymptotic properties are equal to those of $\widehat{\Pi}^N_\varphi$ given in Theorem \ref{prop2}. Define 
	\begin{equation*}
		\ddot{X}_{\varphi,t} = X_t - \widehat{\Omega}^{SN}_\varphi \,\left(\widehat{\Omega}^{NN}_{\varphi}{\mid_{\varphi}^\dagger}\right)\,\, \widehat{P}^N_\varphi \Delta X_t - \widehat{\Gamma}^{SN}_\varphi \widehat{C}_{\varphi,Z}^{-1} Z_{\varphi.t},  \quad \quad t = 1,\ldots,T,
	\end{equation*} 
	where  $\widehat{C}_{\varphi,Z} = T^{-1} \sum_{t=1}^T Z_{\varphi,t} \otimes Z_{\varphi,t}$. Let $\ddot{C}_\varphi$ be the sample covariance of $\{\ddot{X}_{\varphi,t}\}_{t=1}^T$ and consider the following eigenvalue problem, 
	\begin{equation} \label{eigenp2aa}
		\ddot{C}_{\varphi}\ddot{w}_{j} = \ddot{\mu}_{j}\ddot{w}_{j}, \quad j=1,2,\ldots,\dim(\mathcal H).
	\end{equation}
	We then define $\ddot{\Pi}^N_\varphi = \sum_{j=1}^\varphi \ddot{w}_j \otimes \ddot{w}_j$ and  $\ddot{\Pi}^S_\varphi = I - \ddot{\Pi}^N_\varphi$.  In fact, \eqref{eigenp2aa} is a simple adaptation of the eigenvalue problem proposed in \cite{harris1997principal} for multivariate cointegrated systems. Based on the asymptotic results given in  \cite{harris1997principal}, the following can be shown.
	\begin{proposition} \label{prop2a} Suppose that Assumptions \ref{assumm0}, \ref{assumW} and \ref{assum2} hold  with $\varphi \geq 1$, $\dim(\mathcal H) < \infty$, and the minimum eigenvalue of $\mathbb{E} [\mathcal E_t \otimes \mathcal E_t]$ is strictly positive. Then 	\begin{align*}
			T(\ddot{\Pi}^N_\varphi-P^N) \quad  \to_{\mathcal L_{\mathcal H}} \quad  \mathrm{G} + \mathrm{G}^\ast,
		\end{align*}
		where $\mathrm{G}$ is given in Theorem \ref{prop2}.
	\end{proposition}
	However, the asymptotic result given in Proposition \ref{prop2a} essentially requires  $\widehat{C}_{\varphi,Z}^{-1}$ to converge in probability to the inverse of $\mathbb{E} [\mathcal E_t \otimes \mathcal E_t]$ (see the proof of Theorem 2 in \citeauthor{harris1997principal}, \citeyear{harris1997principal}). In an infinite dimensional setting, $\mathbb{E} [\mathcal E_t \otimes \mathcal E_t]$ does not allow its inverse as an element of $\mathcal L_{\mathcal H}$ and, moreover, $\widehat{C}_{\varphi,Z}^{-1}$ does not converge to any bounded linear operator since the inverse of the minimum eigenvalue of $\mathbb{E} [\mathcal E_t \otimes \mathcal E_t]$ diverges to infinity. This is the reason why Proposition \ref{prop2a}, unlike Theorem \ref{prop2}, is not applicable in an infinite dimensional setting. 
	
	\begin{remark} \label{remnewadded}	\normalfont 
		\bt{As in Section 3 of \cite{harris1997principal}, it may also be possible to develop Wald-type tests of various linear restrictions on the cointegrating vectors  $\mathcal H^S$ and they will have standard $\chi^2$ limiting distributions (due to the mixed Gaussianity induced by the modified estimator, see the proof of Theorem 4 of \citealp{harris1997principal}) with degrees of freedom increasing linearly in $\dim(\mathcal H)$. However, this approach explicitly requires a finite dimensional setting and thus not directly applicable to a more general potentially infinite dimensional setting. On the other hand, our tests developed in Section \ref{sec:inferenceatt} can be used regardless of if $\dim(\mathcal H) < \infty$ or not.}
	\end{remark}
	\section{General version of the test given in Proposition \ref{propkpss} } \label{sec_general_test}
	Let $\widehat{\Pi}^K_{\varphi_0}$ be the projection onto the span of $\{\hat{w}_j\}_{j=1}^K$ for any arbitrary finite integer $K$ in $(\varphi_0,\dim(\mathcal H)]$  (note that $K = \dim(\mathcal H)$ is possible only when $\dim(\mathcal H)<\infty$); that is, $\widehat{\Pi}^K_{\varphi_0} =  \sum_{j=1}^K \hat{w}_{j}\otimes \hat{w}_{j}$. Define
	\begin{equation} \label{defz}
		z_{\varphi_0,t} = (\langle X_{\varphi_0,t},\hat{w}_{\varphi_0+1}\rangle, \ldots, \langle X_{\varphi_0,t}, \widehat{w}_{K}\rangle)'.
	\end{equation}
	Once $\{\hat{w}_{j}\}_{j=1}^K$ are given, $z_{\varphi_0,t}$ can easily be  computed.  It should be noted that $z_{\varphi_0,t}$ may be understood as $\widehat{\Pi}^K_{\varphi_0}\widehat{\Pi}^S_{\varphi_0}{X}_{\varphi_0,t}$, i.e., the projected image of possibly infinite dimensional element $\widehat{\Pi}^S_{\varphi_0}{X}_{\varphi_0,t}$ on $\spn(\{\hat{w}_j\}_{j=1}^K)$ of dimension $K$; this $K$-dimensional time series of  $z_{\varphi_0,t}$  is all we need to construct our test statistic. 
	It may be helpful to outline how the sequence $\{z_{\varphi_0,t}\}_{t\geq 1}$ behaves under $H_0$ and $H_1$ prior to defining the test statistic. From the asymptotic results derived in Section \ref{sec:fpca},  the following  can be shown (Lemma \ref{lem3}): for any $\varphi_0 \leq \varphi$,  
	\begin{align*}
		&\hat{w}_j \to_p w_j \in \mathcal H^N,   \quad\quad j=1,\dots,\varphi, \\ 
		&\hat{w}_j \to_p w_j \in \mathcal H^S,  \,\quad\quad j=\varphi+1,\dots, K.
	\end{align*}
	From the above results, we expect that the subvector  $(\langle {X}_{\varphi_0,t},\hat{w}_{\varphi_0+1}\rangle, \ldots, \langle {X}_{\varphi_0,t}, \hat{w}_{\varphi}\rangle)'$   of $z_{\varphi_0,t}$ will behave as a unit root process while the remaining subvector will behave as a stationary process. Only when $\varphi_0 =\varphi$, $\{z_{\varphi_0,t}\}_{t\geq 1}$ will behave as a stationary process. Based on this idea, our test statistic is constructed to examine stationarity of $\{z_{\varphi_0,t}\}_{t\geq 1}$ as follows:   
	\begin{equation}
		\widehat{\mathcal Q}(K,\varphi_0) = \frac{1}{T^2}\sum_{t=1}^T \left(\sum_{s=1}^t z_{\varphi_0,s}\right)'\LRV(z_{\varphi_0,t})^{-1} \left(\sum_{s=1}^t z_{\varphi_0,s}\right),  \label{gkpsstest}
	\end{equation}
	where \begin{equation}
		\LRV(z_{\varphi_0,t}) \hspace{-0.2em}=\hspace{-0.2em}  \frac{1}{T}\sum_{t=1}^{T} z_{\varphi_0,t} z_{\varphi_0,t}'   \hspace{-0.1em}+\hspace{-0.1em} \frac{1}{T}\sum_{s=1}^{T-1}  \mathrm{k}\left(\frac{s}{h}\right) \sum_{t=s+1}^T \left\{z_{\varphi_0,t} z_{\varphi_0,t-s}' \hspace{-0.em}+\hspace{-0.1em} z_{\varphi_0,t-s} z_{\varphi_0,t}' \right\}. 
	\end{equation}
	The test statistic is similarly constructed as other KPSS-type statistics developed for multivariate cointegrated systems; see e.g.\ \cite{shin1994residual}, \cite{choi1995testing}, and \cite{harris1997principal}. Computation of our test statistic is easy  once the projected time series $\{z_{\varphi_0,t}\}_{t=1}^T$ is obtained from the modified eigenvalue problem \eqref{eigenptest}. If $K = \varphi_0+1$, the test statistic \eqref{gkpsstest} becomes identical to \eqref{eqkpssorigin2}.  	
	
	We hereafter let $\mathsf{B}$   and $\mathsf{W}$ denote \textit{independent} $\varphi_0$-dimensional and $(K-\varphi_0)$-dimensional standard Brownian motions, respectively.  The asymptotic properties of our test statistic are described by these independent Brownian motions  as follows.
	
	\begin{theorem} \label{prop3a} Suppose that Assumptions  \ref{assumm0}, \ref{assumW} and \ref{assum2} hold, and $K$ is a finite integer in $(\varphi_0,\dim(\mathcal H)]$. Under $H_0$ of \eqref{hypo0},
		\begin{equation*}
			\widehat{\mathcal Q}(K,\varphi_0)  \quad  \to_d \quad  \int V(r)'V(r),
		\end{equation*}
		where $V(r)=\mathsf{W}(r) - \int d\mathsf{W}(s) \mathsf{B}(s)' \left(\int \mathsf{B}(s)\mathsf{B}(s)' \right)^{-1}\int_{0}^{r} \mathsf{B}(s)$, and the second term is regarded as zero if $\varphi_0 = 0$.	Under $H_1$ of \eqref{hypo0}, $	\widehat{\mathcal Q}(K,\varphi_0) \to_p \infty$. 
	\end{theorem}
	
	\section{Inclusion of deterministic terms} \label{sec:deterministic}
	We now extend the main results given in Sections \ref{sec:fpca}, \ref{sec:test} and \ref{sec_general_test} to allow deterministic terms that may be included in the time series of interest. The proofs of the results in this section will be given later in Section \ref{appproof}. 
	
	In particular, we in this section consider the following unobserved component models:
	\begin{align}
		&\text{Model D1}\,:\, X_{c,t} = \mu_1 +X_t, \label{eqdeter1}\\
		&\text{Model D2}\,:\, X_{\ell,t} = \mu_1 + \mu_2t +X_t,\label{eqdeter2}
	\end{align}
	where $X_t$ is  a cointegrated time series considered in Section \ref{sec:fpca}.
	We define the functional residual $\overline{U}_{t}$  as in \cite{kokoszka2016kpss}: for $t=1,\ldots,T$,
	\begin{align*}
		&\overline{U}_t = \begin{cases}
			X_{c,t} - \frac{1}{T}\sum_{t=1}^T{X}_{c,t} \quad &\text{if Model D1 is true} \\ 
			X_{\ell,t} - \frac{1}{T}\sum_{t=1}^T{X}_{\ell,t} - \left( t-\frac{T+1}{2} \right) \frac{\sum_{t=1}^T \left( t-\frac{T+1}{2} \right) X_{\ell,t}}{\sum_{t=1}^T \left( t-\frac{T+1}{2} \right)^2} \quad &\text{if Model D2 is true}. 
		\end{cases} 
	\end{align*}
	The sample covariance operator of $\{\overline{U}_t\}_{t=1}^T$ is given by $\overline{C} = T^{-1} \sum_{t=1}^T \overline{U}_{t} \otimes \overline{U}_{t}$. 
	\subsection{Extension of the results given in Section \ref{sec:fpca}} \label{sec:deterministic1}
	Consider the eigenvalue problems given by 
	\begin{align} 
		&\overline{C} \,\overline{v}_{j} = \overline{\lambda}_j\, \overline{v}_{j} , \quad j=1,2\ldots. \label{eigenpadd}
	\end{align}
	We similarly obtain our preliminary estimator $\overline{P}^{N}_\varphi$  from \eqref{eigenpadd} as in \eqref{eqprelim}, and let $\overline{P}^{S}_\varphi = I - \overline{P}^{N}_\varphi$. As will be shown in Theorem \ref{prop3}, the asymptotic limits of these preliminary estimators depend on $\Omega^{NS}$ and $\Gamma^{NS}$.   Define for $t = 1,\ldots,T$,
	\begin{align*}
		&\overline{Z}_{\varphi,t} = \overline{P}^{N}_\varphi \Delta \overline{U}_{t} + \overline{P}^{S}_\varphi   \overline{U}_{t}. 
	\end{align*}
	The sample long run covariance and one-sided long run covariance of $\{\overline{Z}_{\varphi,t}\}_{t=1}^T$ are similarly defined as in \eqref{omegae}, and denoted by $\overline{\Omega}_\varphi$ (resp.\ $\overline{\Gamma}_\varphi$).  As in \eqref{eqdecom2}, we  consider  the following operator decompositions: for $i \in \{N,S\}$ and  $j \in \{N,S\}$,
	\begin{align*}
		&\overline{\Omega}_\varphi^{ij} =   \overline{P}^{i}_\varphi	\overline{\Omega}_\varphi \overline{P}_\varphi^{j},  \quad\quad 	\overline{\Gamma}_\varphi^{ij} =   \overline{P}^{i}_\varphi	\overline{\Gamma}_\varphi \overline{P}^{j}_\varphi.
	\end{align*}
	For $t = 1,\ldots,T$, define  
	\begin{align*}
		&\overline{U}_{\varphi,t} =  	\overline{U}_{t} - \overline{\Omega}_\varphi^{SN}\,\,\left(\overline{\Omega}^{NN}_{\varphi}{\mid_{\varphi}^\dagger}\right) \overline{P}_\varphi^{N}\Delta 	\overline{U}_{t}.
	\end{align*}
	As in Section \ref{sec:mfpca}, we let 
	\begin{align*}
		&\overline{C}_\varphi = T^{-1}\sum_{t=1}^T \overline{U}_{\varphi,t}\otimes\overline{U}_{\varphi,t}, \quad \quad \overline{\Upsilon}_\varphi = \overline{\Gamma}^{NS}_\varphi - \overline{\Gamma}^{NN}_\varphi\,\, (\overline{\Omega}^{NN}_{\varphi}{\mid_{\varphi}^\dagger}) \,\,\overline{\Omega}^{NS}_\varphi,
	\end{align*}
	and consider the following modified eigenvalue problems that are parallel to \eqref{eigenp2a},
	\begin{align*} 
		&\left(\overline{C}_\varphi - \overline{\Upsilon}_\varphi  - \overline{\Upsilon}_\varphi^\ast\right) \overline{w}_{j}  =  \overline{\mu}_{j}\, \overline{w}_{j}, \quad j=1,2,\ldots.
	\end{align*}
	We then construct $\overline{\Pi}^{N}_\varphi$ and  $\overline{\Pi}^{S}_\varphi$ as in \eqref{eqmod01}.
	
	{To describe the asymptotic properties of $\overline{P}^{N}_\varphi$ and $\overline{\Pi}^{N}_\varphi$, we hereafter let
		\begin{align}
			\overline{W}^N(r) = \begin{cases}
				{W}^N(r)-\int {W}^N(s)	 \quad  &\text{if Model D1 is true},  \\
				W^N(r) + (6r-4) \int W^N(s) + (6-12r)  \int s W^N(s)  \quad  &\text{if Model D2 is true}.
			\end{cases}
		\end{align}	
		Moreover, we let $\overline{\mathrm{F}}$ and $\overline{\mathrm{G}}$ be random bounded linear operators satisfying that
		\begin{align*}
			&\overline{\mathrm{F}}	=_{\fdd} \left(\int \overline{W}^N(r)\otimes \overline{W}^N(r)\right)^\dag\left(\int  d {W}^{S}(r)\otimes  \overline{W}^N(r) + \Gamma^{NS}\right),\\  
			&\overline{\mathrm{G}}	=_{\fdd} \left(\int \overline{W}^N(r)\otimes \overline{W}^N(r)\right)^\dag\left(\int  d {W}^{S.N}(r)\otimes  \overline{W}^N(r)\right). 
		\end{align*} where, as in Theorem \ref{prop2}, ${W}^{S.N}= W^S- \Omega^{SN}(\Omega^{NN})^\dag W^N$ and  ${W}^{S.N}$ is independent of $W^N$. The asymptotic properties of the estimators are given as follows.	
		
		\begin{theorem} \label{prop3} Suppose that Assumptions  \ref{assumm0}, \ref{assumW} and \ref{assum2} hold with $\varphi \geq 1$.
			\begin{align*}
				&T(\overline{P}^{N}_\varphi - P^N)\to_{\mathcal L_{\mathcal H}} \overline{\mathrm{F}} + \overline{\mathrm{F}}^\ast \,\,\, \text{and} \,\,\, T(\overline{\Pi}^{N}_\varphi - P^N) \to_{\mathcal L_{\mathcal H}}  \overline{\mathrm{G}} + \overline{\mathrm{G}}^\ast. 
			\end{align*}
		\end{theorem}  
		Moreover, as in Remark \ref{rem2}, it can be  shown without difficulty that $\overline{\Pi}^{N}_\varphi$ is more asymptotically efficient than $\overline{P}^{N}_\varphi$; see Remark \ref{rem2} and our detailed discussion to be given in Section \ref{appsecproof2}. 

		\subsection{Extension of the results given in Sections \ref{sec:test}  and \ref{sec_general_test}} \label{sec:deterministic2}
		For any hypothesized value $\varphi_0$,	we similarly construct $\overline{P}^N_{\varphi_0}$ and $\overline{P}^S_{\varphi_0}$ from  \eqref{eigenpadd}. Define $\overline{Z}_{\varphi_0,t}$, $\overline{\Omega}_{\varphi_0}$, $\overline{\Gamma}_{\varphi_0}$,  $\overline{\Upsilon}_{\varphi_0}$, $\overline{U}_{\varphi,t}$, and $\overline{C}_{\varphi_0}$ for each of Model D1 and Model D2 as in Section \ref{sec:deterministic}. We then consider the following modified eigenvalue problems, 
		\begin{align} 
			&\left(\overline{C}_{\varphi_0} - \overline{\Upsilon}_{\varphi_0} - \overline{\Upsilon}_{\varphi_0}^\ast\right) \overline{w}_{j} =\overline{\mu}_{j}  \overline{w}_{j} , \quad j=1,2,\ldots. \label{eigenptest2} 
		\end{align}
		As in \eqref{defz}, we define the following: for $K > \varphi_0$,
		\begin{align*}
			&\overline{z}_{\varphi,t} = (\langle \overline{U}_{\varphi,t}, \overline{w}_{\varphi_0+1}\rangle, \ldots, \langle \overline{U}_{\varphi,t}, \overline{w}_{K}\rangle)'. 
		\end{align*} 
		To examine the null and alternative hypotheses given in \eqref{hypo0}, we consider the following statistics for Model D1 and Model D2 respectively; 
		\begin{align*}
			&\overline{\mathcal Q}(K,\varphi_0) = \frac{1}{T^2}\sum_{t=1}^T \left(\sum_{s=1}^t\overline{z}_{\varphi,t}\right)'\LRV(\overline{z}_{\varphi,t})^{-1}  \left(\sum_{s=1}^t\overline{z}_{\varphi,t}\right).
		\end{align*}	
		To describe the asymptotic properties of $\overline{\mathcal Q}({K,\varphi_0})$, we  let  
		\begin{eqnarray}
			&\overline{\mathsf{B}}(r)=\begin{cases}
				{\mathsf{B}}(r)-\int \mathsf{B}(s) \quad &\text{if Model D1 is true},\\
				\mathsf{B}(r) + (6r-4)\int \mathsf{B}(s) + (6-12r) \int s\mathsf{B}(s)\quad &\text{if Model D2 is true},
			\end{cases}\\
			&\overline{\mathsf{W}}(r)=\begin{cases}
				\mathsf{W}(r)-r\mathsf{W}(1) \quad &\text{if Model D1 is true},\\
				\mathsf{W}(r) + (2r-3r^2)\mathsf{W}(1) + (6r^2-6r) \int \mathsf{W}(s)\quad &\text{if Model D2 is true},
			\end{cases}
		\end{eqnarray}	
		where  $\mathsf{B}$   and $\mathsf{W}$ are the independent Brownian motions that are defined for the case with no deterministic terms. 	
		The asymptotic properties of the statistics are given follows:
		\begin{theorem} \label{prop3aa}   Suppose that Assumptions  \ref{assumm0}, \ref{assumW} and \ref{assum2} hold, and $K$ is a finite integer in $(\varphi_0,\dim(\mathcal H)]$. Under $H_0$,  
			\begin{align*}
				&\overline{\mathcal Q}({K,\varphi_0})  \to_d \int \overline{V}(r)'\overline{V}(r),
			\end{align*}
			where $$\overline{V}(r) =  \overline{\mathsf{W}}(r) - \int d\mathsf{W}(s)\overline{\mathsf{B}}(s)'\left(\int {\overline{\mathsf{B}}(s)}{\overline{\mathsf{B}}(s)}'  \right)^{-1}\int_{0}^{r} \overline{\mathsf{B}}(s),$$
			and the second term in the expression of $\overline{V}$ is regarded as zero if $\varphi_0 = 0$.		Under $H_1$ and each of the models, $\overline{\mathcal Q}({K,\varphi_0}) \to_p \infty$.
		\end{theorem}
		Tables 2 and 3 of \cite{harris1997principal} report critical values for some choices of $K$ and $\varphi_0$. 
		At least to some extent, our test  may be viewed as an extension of the PCA-based test proposed by \cite{harris1997principal}, which extends the KPSS test for univariate/multivariate cointegrated time series. In such a finite dimensional setting, the test statistic \eqref{gkpsstest} can be reconstructed by the outputs from the PCA described in Section \ref{sec:harris} and $K$ can be set to $\dim(\mathcal H)$; this makes our test become identical to that given by \cite{harris1997principal}.

		It should be noted that  the asymptotic null distribution given in Theorem \ref{prop3a} depends on  $K$.  This is a new aspect that arises from the fact that the proposed test examines stationarity of potentially infinite dimensional time series  $\{\widehat{\Pi}^S_{\varphi_0} {X}_{\varphi_0,t}\}_{t \geq 1}$ by projecting it onto a $K$-dimensional subspace as in \cite{horvath2014test}.  As a consequence, critical values for the test statistic depend on both $K$ and $\varphi_0$; however,  for any fixed $K$ and $\varphi_0$, those can  easily be simulated by standard methods  since the limiting distribution of the test statistic is simply given by a functional of two independent standard Brownian motions. One can refer to Table 1 of \cite{harris1997principal} reporting critical values for a few different choices of  $K$ and $\varphi_0$,  with the caution that the table reports critical values depending on the dimension of time series and the cointegration rank, which correspond to $K$ and $K-\varphi_0$ in this paper, respectively. Even if Theorem \ref{prop3a} holds for any arbitrary finite integer $K$ in $(\varphi_0,\dim(\mathcal H)]$, it may be better, in practice, to set $K$ to be only slightly greater than $\varphi_0$, such as $K = \varphi_0+1$ or $\varphi_0+2$; our simulation results support that a large value of $K$ tends to yield worse finite sample properties of the test; see Section \ref{sec:monte} and also the discussion given by \citet[Section 3.5]{SS2019} in a similar context. 

		\section{Numerical studies} \label{sec:numerical}
		\subsection{Monte Carlo Simulations} \label{sec:monte}
		We first investigate the finite sample performances of our tests given in Sections \ref{sec:test}  and \ref{sec_general_test} by Monte Carlo study. 
		For all simulation experiments, the number of replications is 2000, and  the  nominal size is 5\%. 
		\subsubsection*{Simulation setups}
		Let $\{f_j\}_{j \geq 1}$ be the Fourier basis of $L^2[0,1]$, the Hilbert space of square integrable functions on $[0,1]$ equipped with inner product $\langle f,g\rangle = \int f(u) g(u)du$ for $f,g \in L^2[0,1]$. For each $\varphi$ $(\leq 5$, in our simulation expements), we let $\mathcal  E^N_t$ and $\mathcal E^S_t$ be generated by the following stationary functional AR(1) models: 
		\begin{equation} 
			\mathcal E^N_t =  \sum_{j=1}^{\varphi} \alpha_{j} {\langle g^N_{j}, \mathcal E_{t-1} \rangle }g^N_{j}  + P^N\varepsilon_t,  \quad \quad  
			\mathcal E^S_t =  \sum_{j=1}^{10} \beta_{j} {\langle g^S_{j}, \mathcal E_{t-1} \rangle }g^S_{j}  + P^S\varepsilon_t, \label{eqdgp11}
		\end{equation}
		where $\{g_j^N\}_{j=1}^\varphi$  (resp.\ $\{g_j^S\}_{j=1}^{10}$) are randomly drawn from   $\{f_1,\ldots,f_6\}$ (resp.\ $\{f_7,\ldots,f_{18}\}$) without replacement, and $\varepsilon_t$ is given as follows: for standard normal random variables  $\{\theta_{j,t}\}_{j \geq 1}$ that are independent across $j$ and $t$, $\varepsilon_t = \sum_{j=1}^{80} \theta_{j,t} (0.95)^{j-1} f_j$. 
		We then construct $X_t$ from the relationships $P^N \Delta X_t=\mathcal E^N_t$ and $P^S X_t= \mathcal E^S_t$ for $t\geq 1$. Note that $\mathcal H^N$ is given by the span of  $\{g^N_j\}_{j=1}^\varphi$, which is not fixed across different realizations of the DGP; moreover, an orthonormal basis $\{g^N_j\}_{j=1}^\varphi$ of $\mathcal H^N$ is selected only from a set of  smoother Fourier basis functions $\{f_j\}_{j=1}^{6}$. The former is to avoid potential effects caused by the particular shapes of $\mathcal H^N$ as in \cite{SS2019},  and the latter is not to make estimation of $\mathcal H^N$ with small samples too difficult: as expected from the results given by  \cite{SS2019}, the finite sample performances of tests for the dimension of $\mathcal H^N$  may become poorer as $\mathcal H^N$ includes less smooth functions. We also let  $\{\alpha_j\}_{j=1}^\varphi$ and $\{\beta_j\}_{j=1}^{10}$ be randomly chosen as follows:  for some $\beta_{\min}$ and $\beta_{\max}$,  
		\begin{align*} 
			\alpha_j  \sim  U[-0.5,0.5], \quad \beta_j \sim U[\beta_{\min},\beta_{\max}].
		\end{align*}
		It is known that persistence of the stationary component ($\{\mathcal E^S_t\}_{t \in \mathbb{Z}}$ in this paper) has a significant effect on  finite sample properties of KPSS-type tests (see e.g.,\ \citeauthor{kwiatkowski1992testing}, \citeyear{kwiatkowski1992testing}; \citeauthor{nyblom2000tests}, \citeyear{nyblom2000tests}). We thus will investigate  finite sample properties of our tests under a lower persistence scheme ($\beta_{\min}=0$, $\beta_{\max}= 0.5$) and  a higher persistence scheme ($\beta_{\min}=0.5$, $\beta_{\max}= 0.7$). It may be more common in practice to have a nonzero intercept or a linear time trend in \eqref{eqcointeg2}, and  Section \ref{sec:deterministic2}  provides a relevant extension of Theorem \ref{prop3a} to accommodate such a case. We here consider the former case and add an intercept $\mu_1$, which is also randomly chosen, to each realization of the DGP (some simulation results for the case with a linear trend are reported in Table \ref{tabcorank2trend}); specifically, $\mu_1 = \sum_{j=1}^4 {\tilde{\theta}_j p_{j} }/{\sqrt{ \sum_{j=1}^4 \tilde{\theta}_j^2}}$,  $\{\tilde{\theta}_j\}_{j=1}^4$ are independent standard normal random variables, and $p_j$ is $(j-1)$-th order Legendre polynomial defined on $[0,1]$.  Finally, functional observations used to compute the test statistic are constructed by smoothing $X_t$ observed on 200 regularly spaced points of $[0,1]$ using 20 quadratic B-spline basis functions (the choice of basis functions has minimal effect in our simulation setting).

		We need to specify 	the kernel function $\mathrm{k}(\cdot)$, the bandwidth $h$, and a positive integer $K$ satisfying $K>\varphi_0$ to implement our test. In our simulation experiments, $\mathrm{k}(\cdot)$ is set to the Parzen kernel, and two different values of $h$, $T^{1/3}$ and $T^{2/5}$, are employed to see the effect of the bandwidth parameter on the finite sample performances of our tests.  Moreover, our simulation results reported in this section are obtained by setting $K = \varphi_0 + 1$, which is the smallest possible choice of $K$.  We found some simulation evidence supporting that   this choice of $K$ tends to result in better finite sample properties; to see this, one can compare the results reported in Table \ref{tabcorank} and those in Tables \ref{tabcorank2} and \ref{tabcorank2a}.

		
		\subsubsection*{Simulation results} \label{sec:sim1}
		We investigate the finite sample performance of the proposed test for the dimension of $\mathcal H^N$.  In our simulation experiments, finite sample powers are computed when $\varphi = \varphi_0 + 1$; as $\varphi$ gets larger away from $\varphi_0$, our test tends to exhibit a better finite sample power as is expected (see Table  \ref{tabcorank2add} with Table  \ref{tabcorank}). Table  \ref{tabcorank} summarizes the simulation results under the two different persistence schemes.
		Under the lower persistence scheme ($\beta_{\min}=0$, $\beta_{\max}= 0.5$), for all considered values of $\varphi$ and $h$ the test has excellent size control  with a reasonably good finite sample power. On the other hand, it displays over-rejection under the higher persistence scheme ($\beta_{\min}=0.5$, $\beta_{\max}= 0.7$). Our simulation results evidently show that (i) such an over-rejection tends to disappear as $T$ gets larger and/or $h$ gets bigger, and (ii)  choosing a bigger bandwidth with fixed $T$ tends to lower finite sample power.  To summarize, employing a bigger bandwidth helps us avoid potential over-rejection, which can happen when $\{\mathcal E^S_t\}_{t\in \mathbb{Z}}$ is persistent, at the expense of power. This trade-off between correct size and power in finite samples seems to be commonly observed for other KPSS-type tests; see e.g.\  \cite{kwiatkowski1992testing} and \cite{nyblom2000tests}. 	
		
		
		We also investigate  finite sample properties of the tests for examining hypotheses about $\mathcal H^N$ or $\mathcal H^S$. Among many potentially interesting hypotheses, we consider the following: for a specific vector $x_0 \in \mathcal H$ and an orthonormal set $\{x_j\}_{j=1}^\varphi \subset \mathcal H$, 
		\begin{align}
			H_0 : x_0 \in \mathcal H^N \quad &\text{against} \quad  H_1 : x_0 \notin \mathcal H^N,   \label{eqh01}\\
			H_0 : \spn(\{x_j\}_{j=1}^\varphi) = \mathcal H^N \quad &\text{against} \quad H_1 :  \spn(\{x_j\}_{j=1}^\varphi) \neq \mathcal H^N. \label{eqh02}
		\end{align}
		Note that $\spn(\{g^N_{j}\}_{j=1}^\varphi) = \mathcal H^N$ and $g^S_1 \in \mathcal H^S$ in each realization of the DGP (see \eqref{eqdgp11}).  In our simulation experiments for \eqref{eqh01}, finite sample sizes and powers are computed by setting 
		\begin{align*}
			x_0 =  g_1^N + \frac{\gamma}{T}g_1^S, \quad \gamma = 0,20,40,60.
		\end{align*}
		Clearly, $x_0 = g_1^N \in \mathcal H^N$ if $\gamma=0$. On the other hand, $x_0$ deviates slightly  from $\mathcal H^N$ in the direction of $g_1^S$ when $\gamma>0$, but  such a deviation gets smaller as $T$ increases. 	In our experiments for \eqref{eqh02}, finite sample sizes and powers are computed by setting  
		\begin{align}
			x_1 = g^N_1 + \frac{\gamma}{T} g^S_1, \quad x_j = g^N_j  \text{ \,for }  j=2,\ldots,\varphi, \quad \gamma = 0, 20,40,60. \label{eqhp001}
		\end{align}
		Obviously, $\gamma = 0$ implies that $\spn(\{x_j\}_{j=1}^\varphi) = \mathcal H^N$ in \eqref{eqhp001} while $\gamma>0$ makes $\spn(\{x_j\}_{j=1}^\varphi)$  deviates slightly from $\mathcal H^N $ in the direction of $g_1^S$. The simulation results for our tests of the hypotheses \eqref{eqh01} and \eqref{eqh02} for a few different values of $\varphi$ are reported in 	Tables \ref{tabbottom01a} and \ref{tabbottom01};   since the proposed tests are essentially simple modifications of our  test for the dimension of $\mathcal H^N$, they seem to have similar finite sample properties to those we observed in Table \ref{tabcorank}.  We can also observe a trade-off between correct size and power when $\{\mathcal E^S_t\}_{t\in \mathbb{Z}}$ is persistent  in each of Tables \ref{tabbottom01a} and \ref{tabbottom01}. \bt{As mentioned in Remark \ref{remseq}, our test can  sequentially be applied to estimate $\varphi=\dim(\mathcal H^N)$ as an alternative method to the existing ones. In Table \ref{tabcorank2determination}, we compare ours with the recent testing procedure proposed by \cite{SS2019}. According to our simulation results, our proposed method in general tends to perform comparably with that of \cite{SS2019}. In panel (a) where the stationary component is not persistent, even if the test of \cite{SS2019} tends to be better when $\varphi$ is large, our test appears to exhibit some advantages when $T$ is small and $\varphi \leq 3$. Thus the results do not indicate any absolute advantage. On the other hand, in panel (b) where the stationary component is persistent, the performance of our testing procedure tends to be significantly better than that of \cite{SS2019}. This simulation evidence (with our additional simulation results which are reported in panel (c) of Table  \ref{tabcorank2determination}) suggests that ours can be an attractive alternative or complement to the test of \cite{SS2019}  for practitioners.}
		
		\bt{Among the considered tests in this section, our test for \eqref{eqh01} requires a prior knowledge on the true dimension $\varphi = \dim(\mathcal H^N)$ and the rejection frequencies given in Table \ref{tabbottom01a} are computed assuming that $\varphi$ is known. In Table \ref{tabadd2}, we replace $\varphi$ with $\hat{\varphi}$ obtained from our testing procedure (Remark \ref{remseq}). In this case, the magnitude of over-rejection is found to be bigger overall compared to that in the previous case with known $\varphi$. Such an over-rejection seems to be particularly severe when $\varphi$ is large and $T$ is small. This may be due to that inaccuracy of $\hat{\varphi}$ tends to be bigger when $\varphi$ is larger and/or $T$ is smaller (see Table \ref{tabcorank2determination}).}
		

		\begin{table}[b!]
			\caption{Rejection frequencies (\%) of the test for \eqref{hypo0},  $K = \varphi_0+1$}  \label{tabcorank}
			\renewcommand*{\arraystretch}{1} \vspace{-0.75em}
			\rule[2pt]{1\textwidth}{1.15pt}
			\begin{subtable}{.48\linewidth}
				\vspace{-0.0em}\caption{$\beta_j \sim U[0,0.5]$, \,\,\,$h=T^{1/3}$}\vspace{-0.4em}
				\begin{tabular*}{1\textwidth}{@{\extracolsep{\fill}}c@{\hspace{3\tabcolsep}}|c@{\hspace{1.25\tabcolsep}}cccc}
					\hline
					$T$ & ${\varphi_0=0}$ & {$1$} &{$2$} &{$3$} & {$4$} \\
					\hline 
					& 	\multicolumn{5}{c}{{size ($\varphi = \varphi_0$)}}\\
					250 &     4.6&  4.4 & 4.1 & 4.3  &4.5	 \\
					500 &    4.6&  5.2&  5.1 & 4.8&  4.5 \\	
					&&&&&\\[-11pt]
					&	\multicolumn{5}{c}{power ($\varphi = \varphi_0+1$)}\\	
					250 &  97.1& 89.0& 79.9& 73.9& 69.5	\\
					500 & 99.3 &96.5& 94.6 &92.5& 90.9
					\\ \hline
				\end{tabular*}
			\end{subtable} \quad
			\begin{subtable}{.48\linewidth}
				\vspace{-0.0em}\caption{$\beta_j \sim U[0,0.5]$, \,\,\,$h=T^{2/5}$}\vspace{-0.4em}
				\begin{tabular*}{1\textwidth}{@{\extracolsep{\fill}}c@{\hspace{3\tabcolsep}}|c@{\hspace{1.25\tabcolsep}}cccc}
					\hline
					$T$ & ${\varphi_0=0}$ & {$1$} &{$2$} &{$3$} & {$4$} \\
					\hline 
					& 	\multicolumn{5}{c}{{size ($\varphi = \varphi_0$)}}\\
					250 &  4.3 & 4.4&  4.4 & 4.0 & 4.4 \\
					500 &  4.4 & 5.2 & 5.2&  5.2&  4.4	\\	
					&&&&&\\[-11pt]
					&	\multicolumn{5}{c}{power ($\varphi = \varphi_0+1$)}\\	
					250 &  91.7& 76.6& 63.1& 54.9& 48.4				\\
					500 &   97.0& 90.4 &82.3& 78.9& 74.1				 \\ \hline
				\end{tabular*}
			\end{subtable}
			
			\vspace{0.5em}
			\begin{subtable}{.48\linewidth}
				\vspace{-0.2em}\caption{$\beta_j \sim U[0.5,0.7]$, \,\,\,$h=T^{1/3}$}\vspace{-0.4em}
				\begin{tabular*}{1\textwidth}{@{\extracolsep{\fill}}c@{\hspace{3\tabcolsep}}|c@{\hspace{1.25\tabcolsep}}cccc}
					\hline
					$T$ & ${\varphi_0=0}$ & {$1$} &{$2$} &{$3$} & {$4$}   \\
					\hline 
					& 	\multicolumn{5}{c}{{size ($\varphi = \varphi_0$)}}\\
					250 &  10.9& 10.6& 12.0& 11.2 &14.2	\\
					500 &   8.0 & 8.9  &9.4 & 9.2& 11.6	\\	
					&&&&&\\[-11pt]
					&	\multicolumn{5}{c}{power ($\varphi = \varphi_0+1$)}\\	
					250 &97.1& 89.2& 80.1& 74.0& 69.7		\\
					500 & 99.3 &96.5 &94.6& 92.6& 91.0\\ \hline
				\end{tabular*}
			\end{subtable} \quad
			\begin{subtable}{.48\linewidth}
				\vspace{-0.2em}\caption{$\beta_j \sim U[0.5,0.7]$, \,\,\,$h=T^{2/5}$}\vspace{-0.4em}
				\begin{tabular*}{1\textwidth}{@{\extracolsep{\fill}}c@{\hspace{3\tabcolsep}}|c@{\hspace{1.25\tabcolsep}}cccc}
					\hline
					$T$ & ${\varphi_0=0}$ & {$1$} &{$2$} &{$3$} & {$4$}  \\
					\hline 
					& 	\multicolumn{5}{c}{{size ($\varphi = \varphi_0$)}}\\
					250 & 8.4&  8.0&  9.6&  8.8& 10.8	\\
					500 &  6.0&  7.0&  6.8 & 7.6&  7.0		 \\	
					&&&&&\\[-11pt]
					&	\multicolumn{5}{c}{power ($\varphi = \varphi_0+1$)}\\	
					250 & 91.8 &76.5 &63.7 &55.0 &48.8				 \\
					500 &97.0& 90.3& 82.4& 78.9 &74.1	\\ \hline
				\end{tabular*}
			\end{subtable}
			\rule{1\textwidth}{1.25pt} 
		\end{table}

		\singlespacing
		\subsection{Supplementary simulation results} \label{app:addsim}
		\begin{table}[H]
			\caption{Supplementary results to Table \ref{tabcorank}}  \label{tabcorank2add}
			\renewcommand*{\arraystretch}{1}
			\rule[6pt]{1\textwidth}{1.25pt}
			\begin{subtable}{.48\linewidth}
				\vspace{-0.3em}\caption{$\beta_j \sim U[0,0.5]$,  \,\,\,$h=T^{1/3}$}\vspace{-0.4em}
				\begin{tabular*}{1\textwidth}{@{\extracolsep{\fill}}c|c@{\hspace{1.25\tabcolsep}}cccc}
					\hline
					$T$ & ${\varphi_0\hspace{-0.1cm}=\hspace{-0.1cm}0}$ & {$1$} &{$2$} &{$3$} & {$4$} \\
					\hline 
					&	\multicolumn{5}{c}{power ($\varphi = \varphi_0+2$)}\\	
					250 & 99.1 &95.4& 87.7 &81.9& 77.7	\\
					500 & 100.0& 98.8 &97.5 &96.1 &94.7	\\ \hline
				\end{tabular*}
			\end{subtable} \quad
			\begin{subtable}{.48\linewidth}
				\vspace{-0.3em}\caption{$\beta_j \sim U[0,0.5]$, \,\,\,$h=T^{2/5}$}\vspace{-0.4em}
				\begin{tabular*}{1\textwidth}{@{\extracolsep{\fill}}c|c@{\hspace{1.25\tabcolsep}}cccc}
					\hline
					$T$ & ${\varphi_0\hspace{-0.1cm}=\hspace{-0.1cm}0}$ & {$1$} &{$2$} &{$3$} & {$4$} \\
					\hline 
					&	\multicolumn{5}{c}{power ($\varphi = \varphi_0+2$)}\\	
					250 &97.8& 85.5& 71.8 &62.4& 55.9		\\
					500 &  99.5& 95.2 &89.7 &85.9 &82.1			\\ \hline
				\end{tabular*}
			\end{subtable}
			\vspace{0.3em}
			\begin{subtable}{.48\linewidth}
				\vspace{0.3em}\caption{$\beta_j \sim U[0.5,0.7]$, \,\,\, $h=T^{1/3}$}\vspace{-0.4em}
				\begin{tabular*}{1\textwidth}{@{\extracolsep{\fill}}c|c@{\hspace{1.25\tabcolsep}}cccc}
					\hline
					$T$ & ${\varphi_0\hspace{-0.1cm}=\hspace{-0.1cm}0}$ & {$1$} &{$2$} &{$3$} & {$4$}   \\
					\hline 
					&	\multicolumn{5}{c}{power ($\varphi = \varphi_0+2$)}\\	
					250 & 99.2& 95.4& 87.5 &82.1& 77.9	\\
					500 &100.0 &98.8 &97.5& 96.2& 94.8		\\ \hline
				\end{tabular*}
			\end{subtable} \quad
			\begin{subtable}{.48\linewidth}
				\vspace{0.3em}\caption{$\beta_j \sim U[0.5,0.7]$, \,\,\,$h=T^{2/5}$}\vspace{-0.4em}
				\begin{tabular*}{1\textwidth}{@{\extracolsep{\fill}}c|c@{\hspace{1.25\tabcolsep}}cccc}
					\hline
					$T$ & ${\varphi_0\hspace{-0.1cm}=\hspace{-0.1cm}0}$ & {$1$} &{$2$} &{$3$} & {$4$}  \\
					\hline 
					&	\multicolumn{5}{c}{power ($\varphi = \varphi_0+2$)}\\	
					250 &97.9& 85.5 &72.0& 62.5& 56.2\\
					500 &99.5 &95.2 &89.7 &85.8 &82.3		\\ \hline
				\end{tabular*}
			\end{subtable}
			\rule{1\textwidth}{1.25pt} 
		\end{table} 	
		
		\begin{table}[H]
			\caption{Simulation results for \eqref{hypo0}, $K=\varphi_0+2$}  \label{tabcorank2}
			\renewcommand*{\arraystretch}{1}
			\rule[6pt]{1\textwidth}{1.25pt}
			\begin{subtable}{.48\linewidth}
				\vspace{-0.3em}\caption{$\beta_j \sim U[0,0.5]$,  \,\,\,$h=T^{1/3}$}\vspace{-0.4em}
				\begin{tabular*}{1\textwidth}{@{\extracolsep{\fill}}c|c@{\hspace{1.25\tabcolsep}}cccc}
					\hline
					$T$ & ${\varphi_0\hspace{-0.1cm}=\hspace{-0.1cm}0}$ & {$1$} &{$2$} &{$3$} & {$4$} \\
					\hline 
					& 	\multicolumn{5}{c}{{size ($\varphi = \varphi_0$)}}\\
					250 &  4.6 & 4.8&  5.6  &4.9 & 4.6 \\
					500 &  5.6 & 4.8 & 5.2&  4.6 & 4.6	\\	
					&&&&&\\[-11pt]
					&	\multicolumn{5}{c}{power ($\varphi = \varphi_0+1$)}\\	
					250 &94.2& 77.9 &64.0& 54.7& 48.0			\\
					500 &97.9& 91.3& 85.0 &82.1 &74.9		\\ \hline
				\end{tabular*}
			\end{subtable} \quad
			\begin{subtable}{.48\linewidth}
				\vspace{-0.3em}\caption{$\beta_j \sim U[0,0.5]$, \,\,\,$h=T^{2/5}$}\vspace{-0.4em}
				\begin{tabular*}{1\textwidth}{@{\extracolsep{\fill}}c|c@{\hspace{1.25\tabcolsep}}cccc}
					\hline
					$T$ & ${\varphi_0\hspace{-0.1cm}=\hspace{-0.1cm}0}$ & {$1$} &{$2$} &{$3$} & {$4$} \\
					\hline 
					& 	\multicolumn{5}{c}{{size ($\varphi = \varphi_0$)}}\\
					250 &4.3&  4.0&  4.8 & 4.3 & 4.2	 \\
					500 &5.2 & 4.8 & 5.1  &4.9 & 4.0\\	
					&&&&&\\[-11pt]
					&	\multicolumn{5}{c}{power ($\varphi = \varphi_0+1$)}\\	
					250 & 88.4& 62.7& 44.6& 38.0 &32.0\\
					500 & 94.1& 79.8& 66.3& 61.4& 54.9		\\ \hline
				\end{tabular*}
			\end{subtable}
			
			\vspace{0.5em}
			\begin{subtable}{.48\linewidth}
				\vspace{0.3em}\caption{$\beta_j \sim U[0.5,0.7]$, \,\,\, $h=T^{1/3}$}\vspace{-0.4em}
				\begin{tabular*}{1\textwidth}{@{\extracolsep{\fill}}c|c@{\hspace{1.25\tabcolsep}}cccc}
					\hline
					$T$ & ${\varphi_0\hspace{-0.1cm}=\hspace{-0.1cm}0}$ & {$1$} &{$2$} &{$3$} & {$4$}   \\
					\hline 
					& 	\multicolumn{5}{c}{{size ($\varphi = \varphi_0$)}}\\
					250 & 14.2& 12.0& 13.4 &13.2& 14.6	\\
					500 &  9.6& 10.0&  9.6& 10.3& 12.0	\\	
					&&&&&\\[-11pt]
					&	\multicolumn{5}{c}{power ($\varphi = \varphi_0+1$)}\\	
					250 & 95.1& 80.6& 69.0& 61.0& 55.7			\\			
					500 &98.4& 92.4& 87.1 &85.0& 80.4	\\ \hline
				\end{tabular*}
			\end{subtable} \quad
			\begin{subtable}{.48\linewidth}
				\vspace{0.3em}\caption{$\beta_j \sim U[0.5,0.7]$, \,\,\,$h=T^{2/5}$}\vspace{-0.4em}
				\begin{tabular*}{1\textwidth}{@{\extracolsep{\fill}}c|c@{\hspace{1.25\tabcolsep}}cccc}
					\hline
					$T$ & ${\varphi_0\hspace{-0.1cm}=\hspace{-0.1cm}0}$ & {$1$} &{$2$} &{$3$} & {$4$}  \\
					\hline 
					& 	\multicolumn{5}{c}{{size ($\varphi = \varphi_0$)}}\\
					250 &9.6 & 8.6 &10.0 & 8.4& 10.5\\
					500 &  7.3&  7.6  &7.6 & 7.6  &7.0				 \\	
					&&&&&\\[-11pt]
					&	\multicolumn{5}{c}{power ($\varphi = \varphi_0+1$)}\\	
					250 &  90.1& 65.6& 49.2& 42.8& 36.6				 \\
					500 &  94.5 &80.9& 69.0 &64.3 &59.4	\\ \hline
				\end{tabular*}
			\end{subtable}
			\rule{1\textwidth}{1.25pt} 
		\end{table} 
		\begin{table}[H]
			\caption{Simulation results for \eqref{hypo0}, $K=\varphi_0+3$}  \label{tabcorank2a}
			\renewcommand*{\arraystretch}{1}
			\rule[6pt]{1\textwidth}{1.25pt}
			\begin{subtable}{.48\linewidth}
				\vspace{-0.3em}\caption{$\beta_j \sim U[0,0.5]$,  \,\,\,$h=T^{1/3}$}\vspace{-0.4em}
				\begin{tabular*}{1\textwidth}{@{\extracolsep{\fill}}c|c@{\hspace{1.25\tabcolsep}}cccc}
					\hline
					$T$ & ${\varphi_0\hspace{-0.1cm}=\hspace{-0.1cm}0}$ & {$1$} &{$2$} &{$3$} & {$4$} \\
					\hline 
					& 	\multicolumn{5}{c}{{size ($\varphi = \varphi_0$)}}\\
					250 &  4.0  &5.2 & 6.0 & 4.9  &4.3 \\
					500 &   4.0&  5.4 & 5.0&  4.2&  4.6	\\	
					&&&&&\\[-11pt]
					&	\multicolumn{5}{c}{power ($\varphi = \varphi_0+1$)}\\	
					250 & 91.4 &69.0& 50.6& 41.6& 36.4		\\
					500 &96.8 &86.6& 73.7 &69.6& 64.9		\\ \hline
				\end{tabular*}
			\end{subtable} \quad
			\begin{subtable}{.48\linewidth}
				\vspace{-0.3em}\caption{$\beta_j \sim U[0,0.5]$, \,\,\,$h=T^{2/5}$}\vspace{-0.4em}
				\begin{tabular*}{1\textwidth}{@{\extracolsep{\fill}}c|c@{\hspace{1.25\tabcolsep}}cccc}
					\hline
					$T$ & ${\varphi_0\hspace{-0.1cm}=\hspace{-0.1cm}0}$ & {$1$} &{$2$} &{$3$} & {$4$} \\
					\hline 
					& 	\multicolumn{5}{c}{{size ($\varphi = \varphi_0$)}}\\
					250 & 3.5 & 4.8&  5.2&  4.3&  3.9	 \\
					500 & 4.1 & 5.0&  4.4&  3.4  &4.0\\	
					&&&&&\\[-11pt]
					&	\multicolumn{5}{c}{power ($\varphi = \varphi_0+1$)}\\	
					250 & 85.0 &51.6& 32.4 &26.7& 23.0\\
					500 &  92.2 &70.2 &52.4& 46.9& 42.8		\\ \hline
				\end{tabular*}
			\end{subtable}
			
			\vspace{0.5em}
			\begin{subtable}{.48\linewidth}
				\vspace{0.3em}\caption{$\beta_j \sim U[0.5,0.7]$, \,\,\, $h=T^{1/3}$}\vspace{-0.4em}
				\begin{tabular*}{1\textwidth}{@{\extracolsep{\fill}}c|c@{\hspace{1.25\tabcolsep}}cccc}
					\hline
					$T$ & ${\varphi_0\hspace{-0.1cm}=\hspace{-0.1cm}0}$ & {$1$} &{$2$} &{$3$} & {$4$}   \\
					\hline 
					& 	\multicolumn{5}{c}{{size ($\varphi = \varphi_0$)}}\\
					250 & 13.6& 13.6& 14.2 &13.0& 15.0	\\
					500 & 10.7& 10.3& 10.0& 10.1& 12.0	\\	
					&&&&&\\[-11pt]
					&	\multicolumn{5}{c}{power ($\varphi = \varphi_0+1$)}\\	
					250 &  93.4& 73.3& 58.3& 50.4& 47.1			\\			
					500 & 97.5& 88.4& 78.3& 76.4& 72.4	\\ \hline
				\end{tabular*}
			\end{subtable} \quad
			\begin{subtable}{.48\linewidth}
				\vspace{0.3em}\caption{$\beta_j \sim U[0.5,0.7]$, \,\,\,$h=T^{2/5}$}\vspace{-0.4em}
				\begin{tabular*}{1\textwidth}{@{\extracolsep{\fill}}c|c@{\hspace{1.25\tabcolsep}}cccc}
					\hline
					$T$ & ${\varphi_0\hspace{-0.1cm}=\hspace{-0.1cm}0}$ & {$1$} &{$2$} &{$3$} & {$4$}  \\
					\hline 
					& 	\multicolumn{5}{c}{{size ($\varphi = \varphi_0$)}}\\
					250 & 8.7 & 8.4  &9.0&  8.1&  9.6\\
					500 & 7.0 & 7.1 & 6.9 & 6.8 & 6.6	 \\	
					&&&&&\\[-11pt]
					&	\multicolumn{5}{c}{power ($\varphi = \varphi_0+1$)}\\	
					250 &87.6 &55.0 &37.6&31.1& 29.4			 \\
					500 & 93.2& 72.3& 56.0& 51.2 &48.2	\\ \hline
				\end{tabular*}
			\end{subtable}
			\rule{1\textwidth}{1.25pt} 
		\end{table}

		\begin{table}[H]
			\caption{Simulation results for \eqref{hypo0}, $K=\varphi_0+1$ (trend-adjusted)}  \label{tabcorank2trend}
			\renewcommand*{\arraystretch}{1}
			\rule[6pt]{1\textwidth}{1.25pt}
			\begin{subtable}{.48\linewidth}
				\vspace{-0.3em}\caption{$\beta_j \sim U[0,0.5]$,  \,\,\,$h=T^{1/3}$}\vspace{-0.4em}
				\begin{tabular*}{1\textwidth}{@{\extracolsep{\fill}}c|c@{\hspace{1.25\tabcolsep}}cccc}
					\hline
					$T$ & ${\varphi_0\hspace{-0.1cm}=\hspace{-0.1cm}0}$ & {$1$} &{$2$} &{$3$} & {$4$} \\
					\hline 
					& 	\multicolumn{5}{c}{{size ($\varphi = \varphi_0$)}}\\
					250 &  4.3 & 5.1&  4.8 & 5.8&  5.6 \\
					500 &  4.4 & 4.4 & 5.1 & 5.2 & 4.7	\\	
					&&&&&\\[-11pt]
					&	\multicolumn{5}{c}{power ($\varphi = \varphi_0+1$)}\\	
					250 & 98.6& 95.5& 89.3& 82.4& 76.8	\\
					500 & 99.9 &99.4& 98.5& 97.2& 95.2	\\ \hline
				\end{tabular*}
			\end{subtable} \quad
			\begin{subtable}{.48\linewidth}
				\vspace{-0.3em}\caption{$\beta_j \sim U[0,0.5]$, \,\,\,$h=T^{2/5}$}\vspace{-0.4em}
				\begin{tabular*}{1\textwidth}{@{\extracolsep{\fill}}c|c@{\hspace{1.25\tabcolsep}}cccc}
					\hline
					$T$ & ${\varphi_0\hspace{-0.1cm}=\hspace{-0.1cm}0}$ & {$1$} &{$2$} &{$3$} & {$4$} \\
					\hline 
					& 	\multicolumn{5}{c}{{size ($\varphi = \varphi_0$)}}\\
					250 &  4.0 & 5.0 & 4.9&  5.3& 5.1	 \\
					500 & 4.0 & 4.4&  5.0&  4.8&  4.0\\	
					&&&&&\\[-11pt]
					&	\multicolumn{5}{c}{power ($\varphi = \varphi_0+1$)}\\	
					250 &94.7& 85.3 &74.4& 62.1& 54.1					\\
					500 &98.7& 95.7 &91.6& 85.9& 82.8	\\ \hline
				\end{tabular*}
			\end{subtable}
			
			\vspace{0.5em}
			\begin{subtable}{.48\linewidth}
				\vspace{0.3em}\caption{$\beta_j \sim U[0.5,0.7]$, \,\,\, $h=T^{1/3}$}\vspace{-0.4em}
				\begin{tabular*}{1\textwidth}{@{\extracolsep{\fill}}c|c@{\hspace{1.25\tabcolsep}}cccc}
					\hline
					$T$ & ${\varphi_0\hspace{-0.1cm}=\hspace{-0.1cm}0}$ & {$1$} &{$2$} &{$3$} & {$4$}   \\
					\hline 
					& 	\multicolumn{5}{c}{{size ($\varphi = \varphi_0$)}}\\
					250 &  12.8& 14.5& 16.0& 16.8 &17.6	\\
					500 &    8.9 &12.0 &10.4& 12.9& 13.2\\	
					&&&&&\\[-11pt]
					&	\multicolumn{5}{c}{power ($\varphi = \varphi_0+1$)}\\	
					250 & 98.5& 95.7& 89.1& 82.1& 76.0	\\
					500 & 99.9& 99.4& 98.6& 97.2& 95.0	\\ \hline
				\end{tabular*}
			\end{subtable} \quad
			\begin{subtable}{.48\linewidth}
				\vspace{0.3em}\caption{$\beta_j \sim U[0.5,0.7]$, \,\,\,$h=T^{2/5}$}\vspace{-0.4em}
				\begin{tabular*}{1\textwidth}{@{\extracolsep{\fill}}c|c@{\hspace{1.25\tabcolsep}}cccc}
					\hline
					$T$ & ${\varphi_0\hspace{-0.1cm}=\hspace{-0.1cm}0}$ & {$1$} &{$2$} &{$3$} & {$4$}  \\
					\hline 
					& 	\multicolumn{5}{c}{{size ($\varphi = \varphi_0$)}}\\
					250 &  9.6& 10.6& 12.0& 12.9 &14.9\\
					500 &6.2 & 8.8 & 7.0 & 9.3 & 8.5 \\	
					&&&&&\\[-11pt]
					&	\multicolumn{5}{c}{power ($\varphi = \varphi_0+1$)}\\	
					250 &94.7& 85.5 &74.6 &62.9 &54.6	 \\
					500 &98.7& 95.7& 91.8 &86.0& 82.8	\\ \hline
				\end{tabular*}
			\end{subtable}
			\rule{1\textwidth}{1.25pt} 
			{\footnotesize Notes: For each realization of the DGP of Model D2, $\mu_1 = \sum_{j=1}^4 {\tilde{\theta}_j p_{j} }/{\sqrt{ \sum_{j=1}^4 \tilde{\theta}_j^2}}$, $\{\tilde{\theta}_j\}_{j=1}^4$ are independent copies of $N(0,1)$, and $p_j$ is $(j-1)$-th order Legendre polynomial defined on $[0,1]$. $\mu_2$ is generated in a similar manner.} 
		\end{table}

		\begin{table}[H]
			\caption{Rejection frequencies (\%) of the test for \eqref{eqh01}} 
			\renewcommand*{\arraystretch}{1}
			\label{tabbottom01a} \vspace{-0.75em}
			\rule[2pt]{1\textwidth}{1.25pt}
			\begin{subtable}{.48\linewidth}
				\vspace{-0.0em}\caption{$\beta_j \sim U[0,0.5]$, \,\,\,$h=T^{1/3}$}\vspace{-0.4em}
				\begin{tabular*}{1\textwidth}{@{\extracolsep{\fill}}c@{\extracolsep{\fill}}c@{\hspace{3\tabcolsep}}|c@{\hspace{1.25\tabcolsep}}ccc}
					\hline
					$T$ & $\gamma/T$ & ${\varphi=1}$ & {$2$} &{$3$} &{$4$}   \\
					\hline 
					250 & 0&   4.6  &4.8&  5.1  &5.2 \\
					500 & 0&   4.6&  4.3 &4.9 & 4.4\\	
					&&&&&\\[-12pt]
					250 & 0.08& 26.2& 19.9& 15.8& 14.3\\
					500 & 0.04& 14.1&  9.9 &10.6 & 8.8	\\	
					&&&&&\\[-12pt]	
					250 & 0.16&61.1& 50.7& 41.2&38.6\\
					500 & 0.08& 42.9 &33.0 &28.2& 27.6	 \\	
					&&&&&\\[-12pt]	
					250 & 0.24&81.1& 72.1 &62.& 59.6	\\
					500 & 0.12&67.7& 56.3 &49.8& 47.8			\\ 
					\hline
				\end{tabular*}
			\end{subtable} \quad
			\begin{subtable}{.48\linewidth}
				\vspace{-0.0em}\caption{$\beta_j \sim U[0,0.5]$, \,\,\,$h=T^{2/5}$}\vspace{-0.4em}
				\begin{tabular*}{1\textwidth}{@{\extracolsep{\fill}}c@{\extracolsep{\fill}}c@{\hspace{3\tabcolsep}}|c@{\hspace{1.25\tabcolsep}}ccc}
					\hline
					$T$ & $\gamma/T$ & ${\varphi=1}$ & {$2$} &{$3$} &{$4$}  \\
					\hline 
					250 & 0& 4.3 & 4.4 & 4.2 & 4.0 \\
					500 & 0& 4.4&  4.1 & 4.8 & 4.2	\\	
					&&&&&\\[-12pt]	
					250 & 0.08& 25.8& 20.6& 16.8 &14.6 \\
					500 & 0.04& 13.6& 10.7& 11.4& 10.0\\	
					&&&&&\\[-12pt]	
					250 & 0.16&60.0 &50.0& 40.4 &38.4			 \\
					500 & 0.08& 42.1& 34.6 &32.0& 30.6		 \\	
					&&&&&\\[-12pt]	
					250 & 0.24& 79.0& 68.0 &57.5& 54.3	\\
					500 & 0.12&66.7& 57.6& 53.4& 51.1			\\ 
					\hline
				\end{tabular*}
			\end{subtable}
			\quad
			\begin{subtable}{.48\linewidth}
				\vspace{0.3em}\caption{$\beta_j \sim U[0.5,0.7]$, \,\,\,$h=T^{1/3}$}\vspace{-0.4em}
				\begin{tabular*}{1\textwidth}{@{\extracolsep{\fill}}c@{\extracolsep{\fill}}c@{\hspace{3\tabcolsep}}|c@{\hspace{1.25\tabcolsep}}ccc}
					\hline
					$T$ & $\gamma/T$ & ${\varphi=1}$ & {$2$} &{$3$} &{$4$}  \\
					\hline 
					250 & 0&   10.7& 11.5& 10.8 &13.0	 \\
					500 & 0&   8.8 & 8.2&  9.5& 10.8		\\		
					&&&&&\\[-12pt]	
					250 & 0.08& 37.4& 31.4& 25.2& 23.2		 \\
					500 & 0.04&26.4& 20.3& 17.8& 18.4		\\	
					&&&&&\\[-12pt]	
					250 & 0.16&71.3& 57.5& 46.6 &44.0			 \\
					500 & 0.08&56.5 &44.0& 38.4& 35.4			 \\	
					&&&&&\\[-12pt]	
					250 & 0.24&  84.4 &74.0 &63.0 &59.7		\\
					500 & 0.12&   76.0& 64.7& 58.0& 54.6		\\ 
					\hline
				\end{tabular*}
			\end{subtable}
			\quad
			\begin{subtable}{.48\linewidth}
				\vspace{0.3em}\caption{$\beta_j \sim U[0.5,0.7]$, \,\,\,$h=T^{2/5}$}\vspace{-0.4em}
				\begin{tabular*}{1\textwidth}{@{\extracolsep{\fill}}c@{\extracolsep{\fill}}c@{\hspace{3\tabcolsep}}|c@{\hspace{1.25\tabcolsep}}ccc}
					\hline
					$T$ & $\gamma/T$ & ${\varphi=1}$ & {$2$} &{$3$} &{$4$}  \\
					\hline 
					250 & 0&  7.5 & 8.7&  8.4&  9.0\\
					500 & 0&  7.0 & 6.0  &7.4 & 7.6			\\	
					&&&&&\\[-12pt]	
					250 & 0.04& 33.8 &26.4 &21.6& 18.4 \\
					500 & 0.08& 24.0 &18.9& 16.2 &15.2	\\	
					&&&&&\\[-12pt]	
					250 & 0.16& 67.0 &51.7& 41.0& 35.4		 \\
					500 & 0.08& 53.9 &41.5& 36.0& 31.6		 \\	
					&&&&&\\[-12pt]	
					250 & 0.24&80.6& 65.9& 53.9 &47.4	\\
					500 & 0.12& 73.5& 61.7& 55.7& 49.9	\\ 
					\hline
				\end{tabular*}
			\end{subtable}
			\rule{1\textwidth}{1.25pt} 
		\end{table}  
		
		\begin{table}[H]
			\caption{Rejection frequencies (\%) of the test for \eqref{eqh02}} 
			\renewcommand*{\arraystretch}{1}
			\label{tabbottom01}  \vspace{-0.75em}
			
			\rule[2pt]{1\textwidth}{1.25pt}
			\begin{subtable}{.48\linewidth}
				\vspace{-0.0em}\caption{$\beta_j \sim U[0,0.5]$, \,\,\,$h=T^{1/3}$}\vspace{-0.4em}
				\begin{tabular*}{1\textwidth}{@{\extracolsep{\fill}}c@{\extracolsep{\fill}}c@{\hspace{3\tabcolsep}}|c@{\hspace{1.25\tabcolsep}}ccc}
					\hline
					$T$ & $\gamma/T$ & ${\varphi=1}$ & {$2$} &{$3$} &{$4$}   \\
					\hline 
					250 & 0&   4.6&  4.4 & 4.3  &4.8 \\
					500 & 0& 4.8&  4.8&  4.8  &4.6		 \\	
					&&&&&\\[-12pt]
					250 & 0.08&  24.2& 25.4& 27.6& 30.9	\\
					500 & 0.04&  14.7& 14.4 &16.0 &18.2\\	
					&&&&&\\[-12pt]	
					250 & 0.16&  60.7& 61.3& 64.2 &67.0			\\
					500 & 0.08& 42.0& 45.2& 47.4& 48.0 \\	
					&&&&&\\[-12pt]	
					250 & 0.24& 81.1& 80.0& 82.7 &83.8	\\
					500 & 0.12& 66.2 &68.8& 71.5& 72.3	\\ 
					\hline
				\end{tabular*}
			\end{subtable} \quad
			\begin{subtable}{.48\linewidth}
				\vspace{-0.0em}\caption{$\beta_j \sim U[0,0.5]$, \,\,\,$h=T^{2/5}$}\vspace{-0.4em}
				\begin{tabular*}{1\textwidth}{@{\extracolsep{\fill}}c@{\extracolsep{\fill}}c@{\hspace{3\tabcolsep}}|c@{\hspace{1.25\tabcolsep}}ccc}
					\hline
					$T$ & $\gamma/T$ & ${\varphi=1}$ & {$2$} &{$3$} &{$4$}  \\
					\hline 
					250 & 0& 4.4&  4.4 & 4.2&  4.3 \\
					500 & 0&   4.8&  4.6 & 4.6  &4.8	\\	
					&&&&&\\[-12pt]	
					250 & 0.08& 24.0& 24.9& 26.6& 29.5		 \\
					500 & 0.04&   14.6& 14.2& 15.6& 17.8		\\	
					&&&&&\\[-12pt]	
					250 & 0.16& 58.8& 59.8 &61.9& 64.9	 \\
					500 & 0.08&41.6 &44.7& 46.9& 47.4		 \\	
					&&&&&\\[-12pt]	
					250 & 0.24&79.0& 77.9& 79.8 &80.7	\\
					500 & 0.12&64.9& 68.0& 70.5 &71.1\\ 
					\hline
				\end{tabular*}
			\end{subtable}
			\quad
			\begin{subtable}{.48\linewidth}
				\vspace{0.3em}\caption{$\beta_j \sim U[0.5,0.7]$, \,\,\,$h=T^{1/3}$}\vspace{-0.4em}
				\begin{tabular*}{1\textwidth}{@{\extracolsep{\fill}}c@{\extracolsep{\fill}}c@{\hspace{3\tabcolsep}}|c@{\hspace{1.25\tabcolsep}}ccc}
					\hline
					$T$ & $\gamma/T$ & ${\varphi=1}$ & {$2$} &{$3$} &{$4$}  \\
					\hline 
					250 & 0&12.8& 12.9& 13.4& 14.0 \\
					500 & 0&  8.7 & 9.8 & 9.4 &11.0	\\	
					&&&&&\\[-12pt]	
					250 & 0.08&38.8 &41.0& 40.6& 43.0			 \\
					500 & 0.04& 26.2& 29.2 &28.6& 31.1		\\	
					&&&&&\\[-12pt]	
					250 & 0.16&68.7& 69.3& 70.1& 72.5			 \\
					500 &  0.08& 53.8 &59.1& 57.9 &59.2				 \\	
					&&&&&\\[-12pt]	
					250 & 0.24&  85.0 &83.6& 85.2 &86.3	\\
					500 & 0.12&74.3 &76.3& 78.9& 78.4\\ 
					\hline
				\end{tabular*}
			\end{subtable}
			\quad
			\begin{subtable}{.48\linewidth}
				\vspace{0.3em}\caption{$\beta_j \sim U[0.5,0.7]$, \,\,\,$h=T^{2/5}$}\vspace{-0.4em}
				\begin{tabular*}{1\textwidth}{@{\extracolsep{\fill}}c@{\extracolsep{\fill}}c@{\hspace{3\tabcolsep}}|c@{\hspace{1.25\tabcolsep}}ccc}
					\hline
					$T$ & $\gamma/T$ & ${\varphi=1}$ & {$2$} &{$3$} &{$4$}  \\
					\hline 
					250 & 0& 9.4 & 9.0 &10.2&  9.6\\
					500 & 0&  6.7 & 7.4 & 7.2 & 7.6		\\	
					&&&&&\\[-12pt]	
					250 & 0.04& 35.0& 36.2& 36.0 &37.6 \\
					500 & 0.08& 23.9 &25.4& 25.2 &27.0		\\	
					&&&&&\\[-12pt]	
					250 & 0.16& 64.8 &65.2& 65.5& 67.4			 \\
					500 & 0.08& 51.3 &55.4 &54.6& 55.1	 \\	
					&&&&&\\[-12pt]	
					250 & 0.24& 81.0 &79.7& 80.6 &81.6	\\
					500 & 0.12& 71.9 &74.1 &76.3 &75.0\\ 
					\hline
				\end{tabular*}
			\end{subtable}
			\rule{1\textwidth}{1.25pt} 
		\end{table} 	
		\begin{table}[H]
			\caption{Relative frequencies of correct determination of $\varphi$ } \centering \label{tabcorank2determination}
			\renewcommand*{\arraystretch}{1}
			\rule[6pt]{1\textwidth}{1.25pt}
			\begin{subtable}{.9\linewidth} \vspace{1.25em}
				\vspace{-0.0em}\caption{$\beta_j \sim U[0,0.5]$}\label{tabcorank2determination1}\vspace{-0.0em}
				\bt{\begin{tabular*}{1\textwidth}{c@{\extracolsep{\fill}}c|c@{\hspace{1.25\tabcolsep}}ccccc}
						\hline
						Method	&$T$ & ${\varphi\hspace{-0.1cm}=\hspace{-0.1cm}0}$ & {$1$} &{$2$} &{$3$} & {$4$} & $5$ \\			\hline 
						$h=T^{1/3}$ &	250 &  0.945 &0.931 &0.847 &0.742& 0.616& 0.488\\
						&	500 &  0.951& 0.941 &0.920& 0.889 &0.852& 0.825
						\\\hline 	
						$h=T^{2/5}$ &	250 & 0.946& 0.899& 0.713& 0.531& 0.324& 0.200	\\
						&	500 &0.950 &0.924& 0.845& 0.745&0.642& 0.568	\\ \hline 
						NSS$_1$ &	250 & 0.626& 0.529& 0.480& 0.498& 0.472& 0.454	\\
						&	500 &  0.994 &0.861 &0.891& 0.901& 0.913 &0.930			\\ \hline 
						NSS$_2$ &	250 & 0.618& 0.517& 0.471& 0.481& 0.442 &0.436	\\
						&	500 &  0.993& 0.859& 0.893& 0.900& 0.910& 0.929		\\ \hline 	
				\end{tabular*}}
			\end{subtable} \quad
			\begin{subtable}{.9\linewidth} \vspace{1.75em}
				\caption{$\beta_j \sim U[0.5,0.7]$}\label{tabcorank2determination2}\vspace{-0.0em}
				\bt{\begin{tabular*}{1\textwidth}{c@{\extracolsep{\fill}}c|c@{\hspace{1.25\tabcolsep}}ccccc}
						\hline
						Method	&$T$ & ${\varphi\hspace{-0.1cm}=\hspace{-0.1cm}0}$ & {$1$} &{$2$} &{$3$} & {$4$} & $5$ \\			\hline 
						$h=T^{1/3}$ &	250 & 0.874& 0.874& 0.784& 0.687& 0.554 &0.450 \\
						&	500 & 0.916& 0.904& 0.878& 0.847& 0.791& 0.775
						\\\hline 	
						$h=T^{2/5}$ &	250 & 0.906 &0.869 &0.682& 0.509& 0.310& 0.190	\\
						&	500 & 0.933& 0.905& 0.831& 0.726& 0.625& 0.548
						\\ \hline 
						NSS$_1$ &	250 & 0.076 &0.072 &0.040 &0.032& 0.015 &0.011	\\
						&	500 & 0.706& 0.546& 0.531& 0.480& 0.450& 0.399
						\\ \hline 
						NSS$_2$ &	250 & 0.066 &0.062 &0.032& 0.028 &0.010& 0.008	\\
						&	500 & 0.698& 0.535& 0.515& 0.466& 0.432& 0.368	\\ \hline 	
				\end{tabular*}}
			\end{subtable}
			\begin{subtable}{.9\linewidth} \vspace{1.75em}
				\caption{$\beta_j \sim U[0,0.7]$}\label{tabcorank2determination3}\vspace{-0.0em}
				\bt{\begin{tabular*}{1\textwidth}{c@{\extracolsep{\fill}}c|c@{\hspace{1.25\tabcolsep}}ccccc}
						\hline
						Method	&$T$ & ${\varphi\hspace{-0.1cm}=\hspace{-0.1cm}0}$ & {$1$} &{$2$} &{$3$} & {$4$} & $5$ \\			\hline 
						$h=T^{1/3}$ &	250 &   0.922& 0.905& 0.823& 0.721& 0.594 &0.469\\
						&	500 &  0.941& 0.929 &0.908& 0.872& 0.831& 0.802
						\\\hline 	
						$h=T^{2/5}$ &	250 & 0.936& 0.888&0.701& 0.525& 0.318 &0.198	\\
						&	500 & 0.945& 0.917&0.843& 0.732 &0.631& 0.553
						\\ \hline 
						NSS$_1$ &	250 & 0.398 &0.340& 0.279& 0.263& 0.218& 0.192	\\
						&	500 & 0.955& 0.777& 0.801& 0.807& 0.805& 0.801				\\ \hline 
						NSS$_2$ &	250 & 0.376& 0.326& 0.263& 0.245& 0.202& 0.168	\\
						&	500 & 0.951 &0.773& 0.798& 0.803& 0.801& 0.795
						\\ \hline 	
				\end{tabular*}}
			\end{subtable}
			\begin{flushleft}
				\rule{1\textwidth}{1.25pt} 
				{\footnotesize 	\bt{Notes: $\varphi$ is estimated by our proposed method with $h=T^{1/3}$ and $h=T^{2/5}$. For comparison, $\varphi$ is also estimated by the sequential method proposed by \cite{SS2019} when the dimension of the asymptotic superspace (see Remark 13, \citealp{SS2019}) is set to $\varphi_0+1$  (NSS$_1$) and $\varphi_0+2$  (NSS$_2$). The maximal possible number of stochastic trends, which is also a necessary tuning parameter for their method, is set to $7$.}} 
			\end{flushleft}
		\end{table}

		\begin{table}[H]
			\caption{Rejection frequencies (\%) of the test for \eqref{eqh01} with $\hat{\varphi}$} 
			\renewcommand*{\arraystretch}{1}
			\label{tabadd2} \vspace{-0.75em}
			\rule[2pt]{1\textwidth}{1.25pt}
			\begin{subtable}{.48\linewidth}
				\vspace{-0.0em}\caption{$\beta_j \sim U[0,0.5]$, \,\,\,$h=T^{1/3}$}\vspace{-0.4em}
				\bt{\begin{tabular*}{1\textwidth}{@{\extracolsep{\fill}}c@{\extracolsep{\fill}}c@{\hspace{3\tabcolsep}}|c@{\hspace{1.25\tabcolsep}}ccc}
						\hline
						$T$ & $\gamma/T$ & ${\varphi=1}$ & {$2$} &{$3$} &{$4$}   \\
						\hline 
						250 & 0&   4.2  &13.8&  20.8 &22.9 \\
						500 & 0&   3.8& 7.1 &8.1 & 8.8\\	
						&&&&&\\[-12pt]
						250 & 0.08& 23.8& 26.8& 30.1& 29.0\\
						500 & 0.04& 14.1& 13.9 &12.6 & 13.4	\\	
						&&&&&\\[-12pt]	
						250 & 0.16&60.5& 56.9& 53.9&47.8\\
						500 & 0.08& 42.4 &36.6 &32.5& 29.9	 \\	
						&&&&&\\[-12pt]	
						250 & 0.24&80.1& 77.1 &72.2& 65.5	\\
						500 & 0.12&67.4& 60.0 &54.4& 50.0			\\ 
						\hline
				\end{tabular*}}
			\end{subtable} \quad
			\begin{subtable}{.48\linewidth}
				\vspace{-0.0em}\caption{$\beta_j \sim U[0,0.5]$, \,\,\,$h=T^{2/5}$}\vspace{-0.4em}
				\bt{\begin{tabular*}{1\textwidth}{@{\extracolsep{\fill}}c@{\extracolsep{\fill}}c@{\hspace{3\tabcolsep}}|c@{\hspace{1.25\tabcolsep}}ccc}
						\hline
						$T$ & $\gamma/T$ & ${\varphi=1}$ & {$2$} &{$3$} &{$4$}  \\
						\hline 
						250 & 0& 5.4 & 14.0 & 20.6& 20.8 \\
						500 & 0& 6.8&  10.2 & 9.8 & 9.6	\\	
						&&&&&\\[-12pt]	
						250 & 0.08& 24.0& 28.3& 31.0 &27.9 \\
						500 & 0.04& 16.9& 17.4& 16.0& 15.6\\	
						&&&&&\\[-12pt]	
						250 & 0.16&59.6 &57.3& 53.8 &45.2			 \\
						500 & 0.08& 43.4& 41.2 &37.3&33.4		 \\	
						&&&&&\\[-12pt]	
						250 & 0.24& 77.8& 74.3 &69.5& 59.2	\\
						500 & 0.12&67.1& 62.5& 58.0& 52.8			\\ 
						\hline
				\end{tabular*}}
			\end{subtable}
			\quad
			\begin{subtable}{.48\linewidth}
				\vspace{0.3em}\caption{$\beta_j \sim U[0.5,0.7]$, \,\,\,$h=T^{1/3}$}\vspace{-0.4em}
				\bt{\begin{tabular*}{1\textwidth}{@{\extracolsep{\fill}}c@{\extracolsep{\fill}}c@{\hspace{3\tabcolsep}}|c@{\hspace{1.25\tabcolsep}}ccc}
						\hline
						$T$ & $\gamma/T$ & ${\varphi=1}$ & {$2$} &{$3$} &{$4$}  \\
						\hline 
						250 & 0&   10.8& 20.8&26.1 &27.2	 \\
						500 & 0&  9.1 & 10.9&  12.5 & 13.3	\\		
						&&&&&\\[-12pt]	
						250 & 0.08& 36.8& 37.0 &37.0& 33.7		 \\
						500 & 0.04&25.6 &  23.4&20.8&19.7		\\	
						&&&&&\\[-12pt]	
						250 & 0.16&65.9& 62.4&58.8 & 50.8			 \\
						500 & 0.08&54.3 &47.9& 41.0& 37.0			 \\	
						&&&&&\\[-12pt]	
						250 & 0.24& 79.9 &77.4 &71.5 &63.4		\\
						500 & 0.12&   74.7& 66.9& 60.9& 55.5		\\ 
						\hline
				\end{tabular*}}
			\end{subtable}
			\quad
			\begin{subtable}{.48\linewidth}
				\vspace{0.3em}\caption{$\beta_j \sim U[0.5,0.7]$, \,\,\,$h=T^{2/5}$}\vspace{-0.4em}
				\bt{\begin{tabular*}{1\textwidth}{@{\extracolsep{\fill}}c@{\extracolsep{\fill}}c@{\hspace{3\tabcolsep}}|c@{\hspace{1.25\tabcolsep}}ccc}
						\hline
						$T$ & $\gamma/T$ & ${\varphi=1}$ & {$2$} &{$3$} &{$4$}  \\
						\hline 
						250 & 0&  9.8  &19.9& 24.4&  24.0				\\
						500 & 0&  10.6  &12.6 &13.0 & 12.6		\\	
						&&&&&\\[-12pt]	
						250 & 0.04& 33.8 &34.2 &33.4&29.9 \\
						500 & 0.08& 26.1 &24.3& 21.1  &18.8	\\	
						&&&&&\\[-12pt]	
						250 & 0.16& 63.7  &58.7&52.7& 43.8		 \\
						500 & 0.08& 53.2 &47.6& 40.6& 35.1		 \\	
						&&&&&\\[-12pt]	
						250 & 0.24&76.5& 71.9& 65.8 &54.0	\\
						500 & 0.12& 73.9& 65.8& 58.0& 51.8	\\ 
						\hline
				\end{tabular*}}
			\end{subtable}
			\rule{1\textwidth}{1.25pt} 
		\end{table}  
		
		\onehalfspacing
		\medskip

		\subsection{Empirical illustration 1: age-specific employment rates}\label{sec:emp}
		We revisit two empirical applications considered, respectively, by \citet[Section 5.1]{SS2019} and \citet[Section 5.1]{Chang2016152} with  extended time spans. In this section, we discuss the first empirical example in detail. 
		Specifically, we here consider the time series of U.S.\ age-specific employment rates for the working age (15-64) population, observed monthly from January 1986 to Dec 2019.  The raw survey data is available from the Current Population Survey (CPS) at \url{https://www.ipums.org/}. 
		For age $a \in [15,64]$, the age-specific employment rate at time $t$, denoted by $Z_{a,t}$, is computed as in \citet[Section 5.1]{SS2019}. Such employment rates take values in $[0,1]$ by construction, hence we take the logit transformation $\phi(Z_{a,t})$ as suggested by \cite{SS2019}, and then  obtain functional observations $X_t(u)$ for $u\in [15,64]$ by  smoothing $\phi (Z_{a,t})$ over $a$ using 18 quadratic B-spline functions. Due to this construction, $\{X_t\}_{t=1}^T$ may be understood as a high but finite dimensional time series; however as discussed in Section \ref{sec:harris}, our methodology is compatible with the finite dimensionality of $X_t$, so this does not cause any theoretical issues in applying our methodology. 
		In Figure \ref{figadd} we plot the functional observations and univariate time series $\{\langle X_t, x\rangle\}_{t \geq 1}$ for some choices of $x$, which may help us explore characteristics of the FTS. Specifically we consider $x = x_1,x_2$ and $x_3$ defined as follows:
		\begin{align}
			x_1(u) = \frac{1}{25} \mathbbm{1}\left\{15 \leq u < 40\right\}, \quad	x_2(u) = \frac{1}{24} \mathbbm{1}\left\{ 40 \leq u \leq 64\right\}, \quad	x_3 = x_2 - x_1. \label{eqexam01}
		\end{align}
		Clearly, $\langle X_t, x_1\rangle$ (resp.\ $\langle X_t, x_2\rangle$) computes the average (logit) employment rate of individuals aged less than 40 years (resp.\ no less than 40 years), and $\langle X_t, x_3\rangle$ computes their difference.  Provided that  the two time series given in Figure \ref{figadd}-(b) seem to be nonstationary, we expect that there exists at least one stochastic trend; however, it is not possible to conclude  from the plots that there are multiple stochastic trends since a single stochastic trend, say $w$, can make both time series of $\langle X_t, x_1\rangle$ and $\langle X_t, x_2\rangle$ nonstationary if $\langle w,x_1 \rangle$ and $\langle w,x_2 \rangle$ are  nonzero. It can also be seen from Figure \ref{figadd}-(c) that for some choice of $x$ the time series of $\langle X_t,x\rangle$ may exhibit a lower persistence than those of $x_1$ and $x_2$; this may be suggestive of the possible existence of cointegrating vectors.  
		\begin{figure}[b!]
			\caption{Age group characteristics}
			\centering
			\label{figadd}
			\begin{subfigure}[b]{0.32\textwidth}
				\includegraphics[width=1\textwidth,trim = 0.8cm 0.8cm 0.8cm 1.9cm,clip]{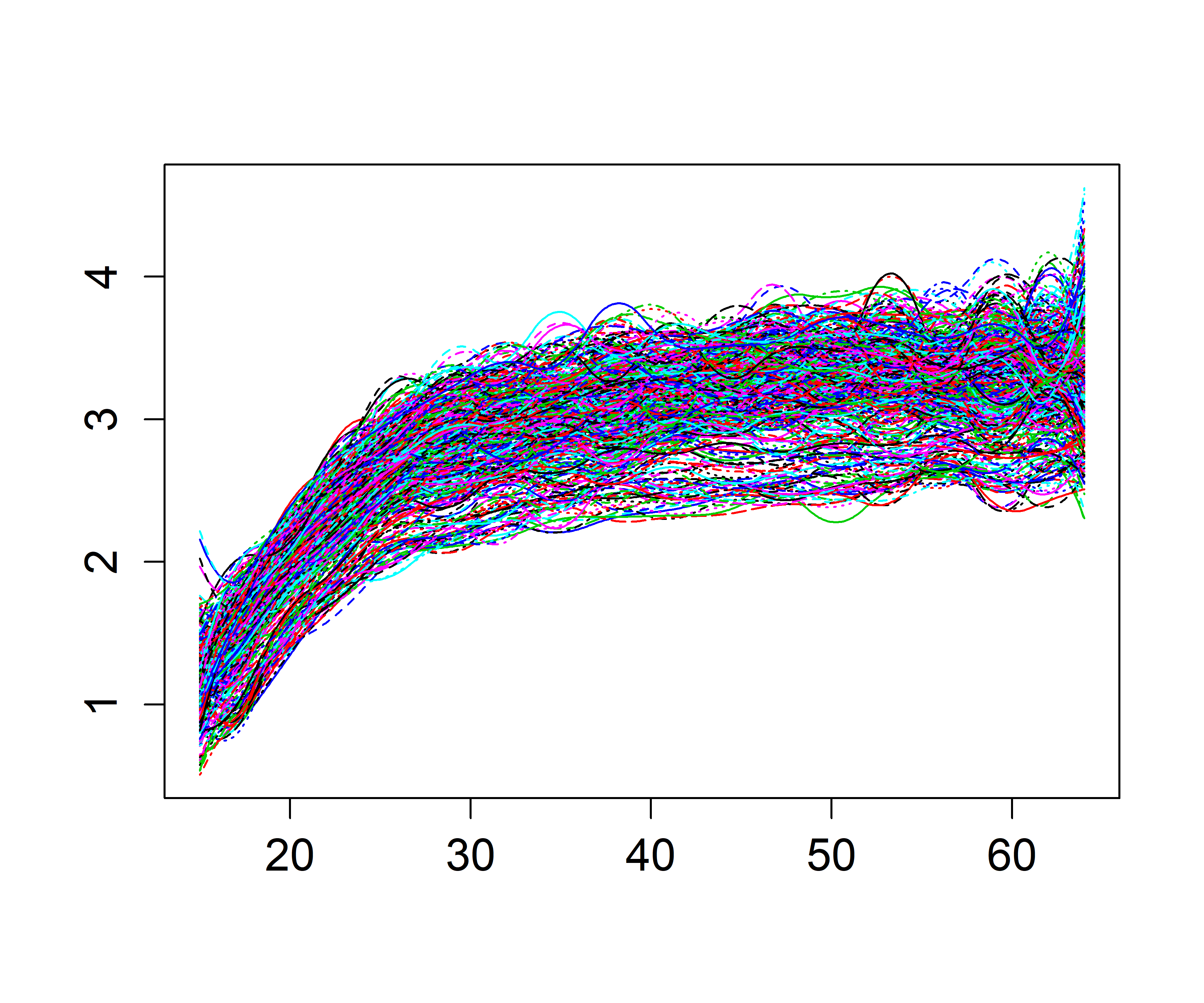}	
				\vspace{-1.5\baselineskip}
				\caption{observations ($X_t$)} 
			\end{subfigure} \quad  
			\begin{subfigure}[b]{0.32\textwidth}
				\includegraphics[width=1\textwidth,trim = 0.8cm 0.8cm 0.8cm 1.9cm,clip]{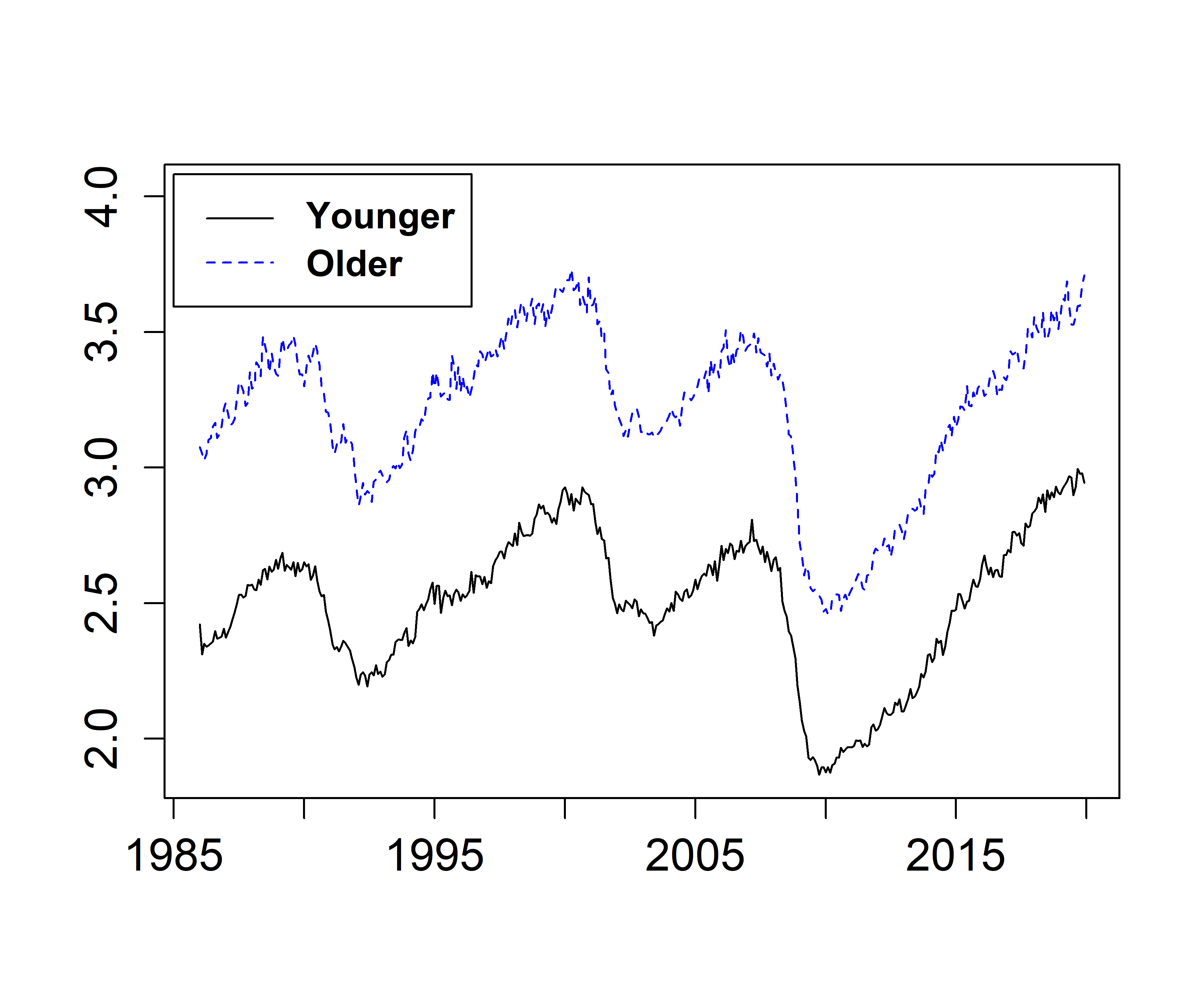}	
				\vspace{-1.5\baselineskip}
				\caption{two groups  ($\langle\hspace{-0.1em} X_t\hspace{-0.075em},\hspace{-0.075em}x_1\hspace{-0.075em}\rangle$ \hspace{-0.13em}\&\hspace{-0.13em} $\langle\hspace{-0.1em} X_t\hspace{-0.075em},\hspace{-0.075em}x_2\hspace{-0.075em}\rangle$)}
			\end{subfigure}
			\begin{subfigure}[b]{0.32\textwidth}
				\includegraphics[width=1\textwidth,trim = 0.8cm 0.8cm 0.8cm 1.9cm,clip]{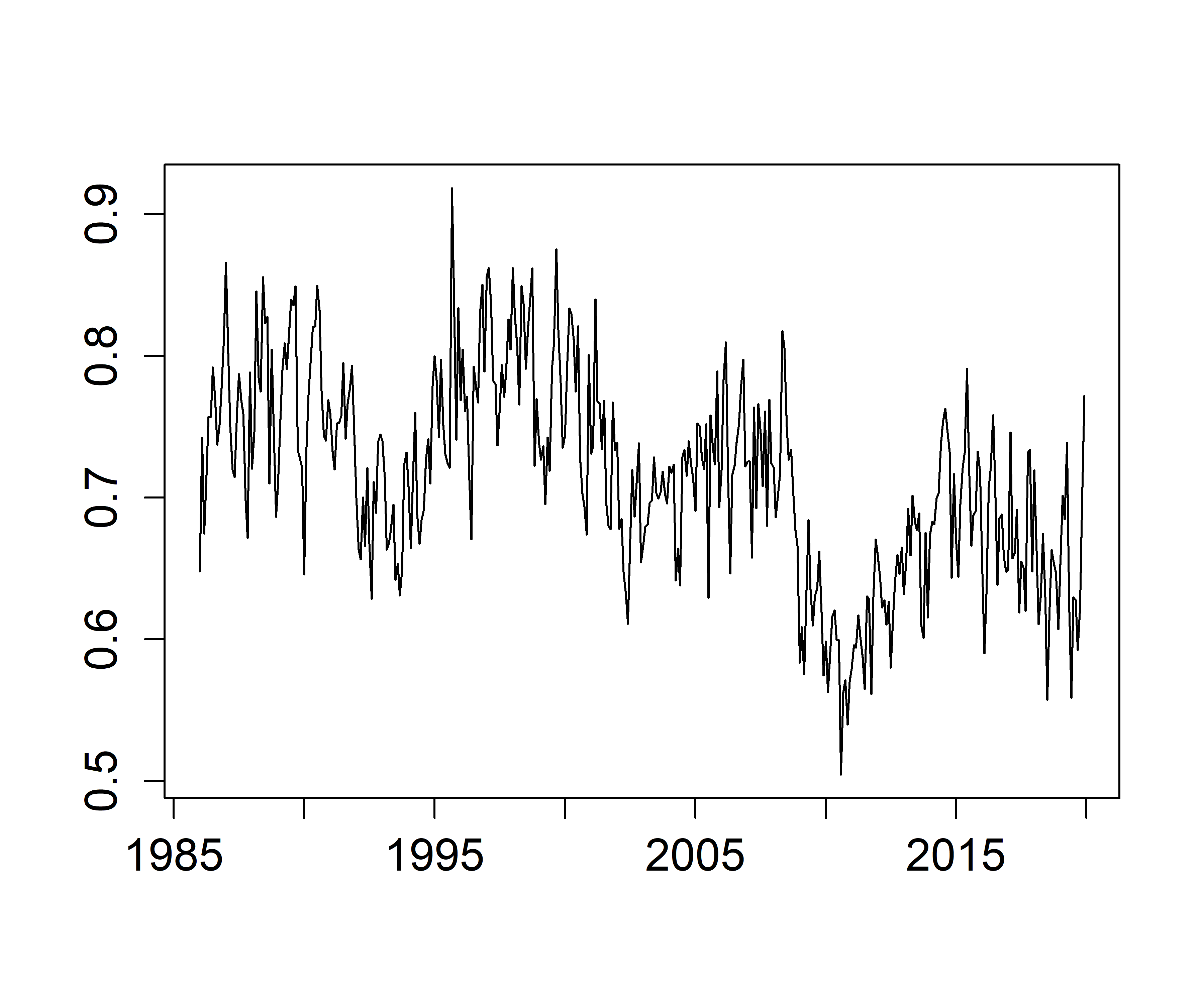}
				\vspace{-1.5\baselineskip}
				\caption{difference   ($\langle X_t,x_3  \rangle$)}
			\end{subfigure} 
		\end{figure}
		
		When a cointegrated FTS is given, it is important to estimate $\mathcal H^N$. We already know  from our earlier results that $\mathcal H^N$ can be estimated by the span of the first $\varphi$ (=$\dim(\mathcal H^N)$) eigenvectors computed from the proposed eigenvalue problem \eqref{eigenp2a}; however in practice $\varphi$ is unknown. We thus apply our FPCA-based sequential procedure to determine $\varphi$. As in \cite{SS2019}, the model with a linear trend (see Section \ref{sec:deterministic}) is considered for this empirical example. Table~\ref{tabl_wage} reports the test results when $\mathrm{k}(\cdot)$, $h$, and  $K$ are set to the Parzen kernel, $T^{2/5}$, and $\varphi_0+1$, respectively. The proposed sequential procedure concludes that the dimension of $\mathcal H^N$ is $1$ both at $5\%$ and $1\%$ significance levels, hence $\mathcal H^N$ can be estimated by the span of the first eigenvector $\overline{w}_{1}$ obtained from our modified FPCA. Figure \ref{figadd2} shows the estimated eigenvector and the time series $\{\langle \overline{U}_t, \overline{w}_{1} \rangle \}_{t=1}^T$. Note that the time series of $\langle X_t, x\rangle$ for any $x \in \mathcal H^N$ is a unit root process with a linear deterministic trend, hence the time series of $\langle \overline{U}_t, \overline{w}_{1} \rangle$ is expected to behave as a unit root process, which is as shown in Figure \ref{figadd2}-(b).  	
		
		\begin{table}[t!] \renewcommand*{\arraystretch}{1}
			\caption{Test results for the dimension of $\mathcal H^N$ - age-specific employment rates} \small
			\label{tabl_wage}	\vspace{-1em}
			\rule[1pt]{1\textwidth}{1.25pt}
			\begin{tabular*}{0.98\textwidth}{@{\extracolsep{\fill}}cccc}
				$\varphi_0$ & $0$ & $1$ & $2$   \\ \hline
				Test statistic & 0.3175$^{\ast\ast}$& 0.1195 & 0.0665 
			\end{tabular*}	
			\rule{1\textwidth}{1.25pt} 
			{\footnotesize Notes: $T=408$.  We use ${}^{\ast}$ and  ${}^{\ast\ast}$ to denote rejection at 5\% and 1\%  significance level, respectively; the approximate critical values for 95\% (resp.\ 99\%) are given by 0.15, 0.12, and 0.10 (resp.\ 0.22, 0.18, and 0.15) sequentially.} 
		\end{table}
		\begin{figure}[t!]
			\caption{Estimated orthonormal basis of $\mathcal H^N$  - age-specific employment rates}
			\centering
			\label{figadd2}
			\begin{subfigure}[b]{0.38\textwidth}
				\includegraphics[width=1\textwidth,trim = 0.2cm 1.5cm 0.7cm 2cm,clip]{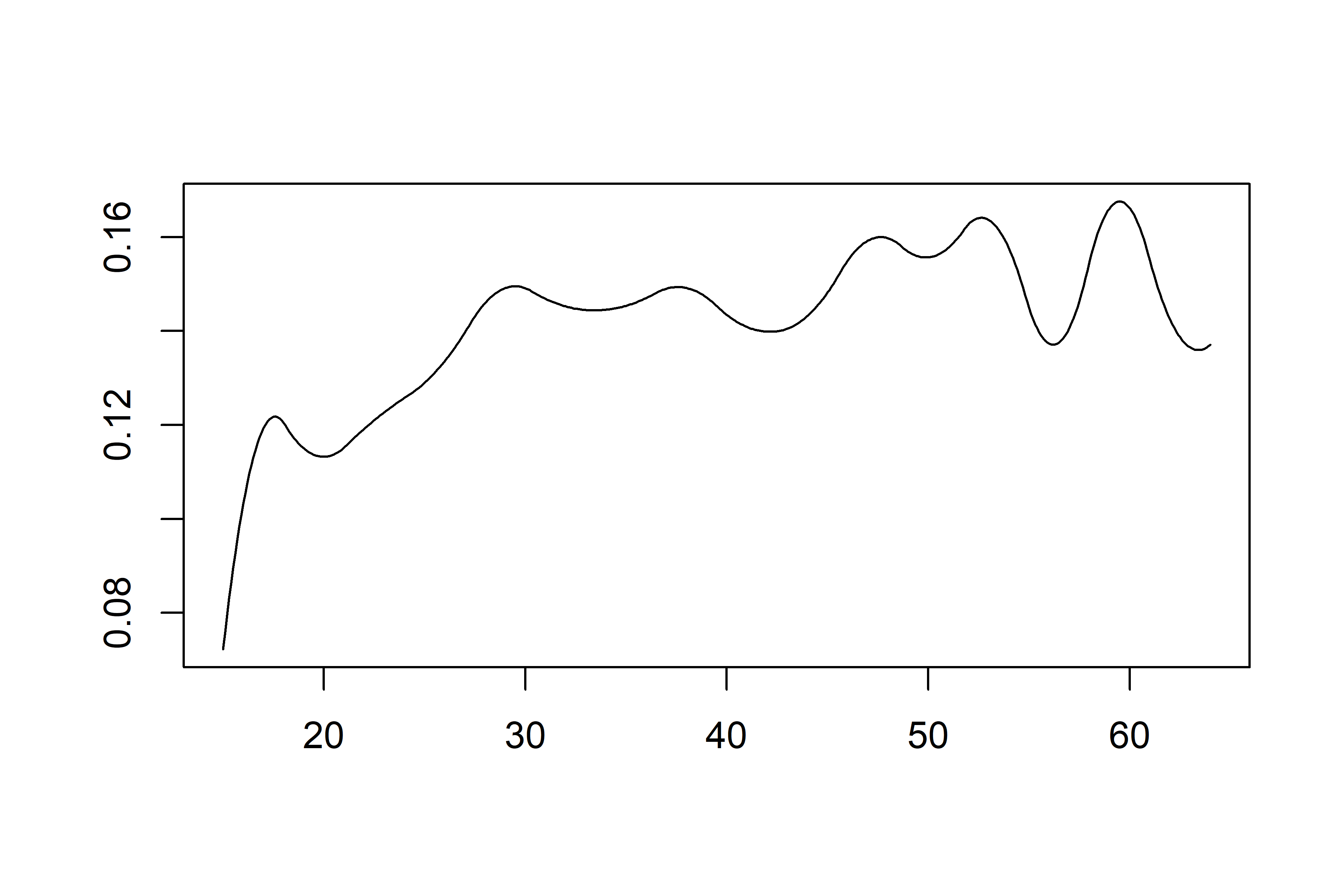}	
				\vspace{-1.5\baselineskip}
				\caption{$\overline{w}_{1}$}
			\end{subfigure}
			\begin{subfigure}[b]{0.38\textwidth}
				\includegraphics[width=1\textwidth,trim = 0.2cm 1.5cm 0.7cm 2cm,clip]{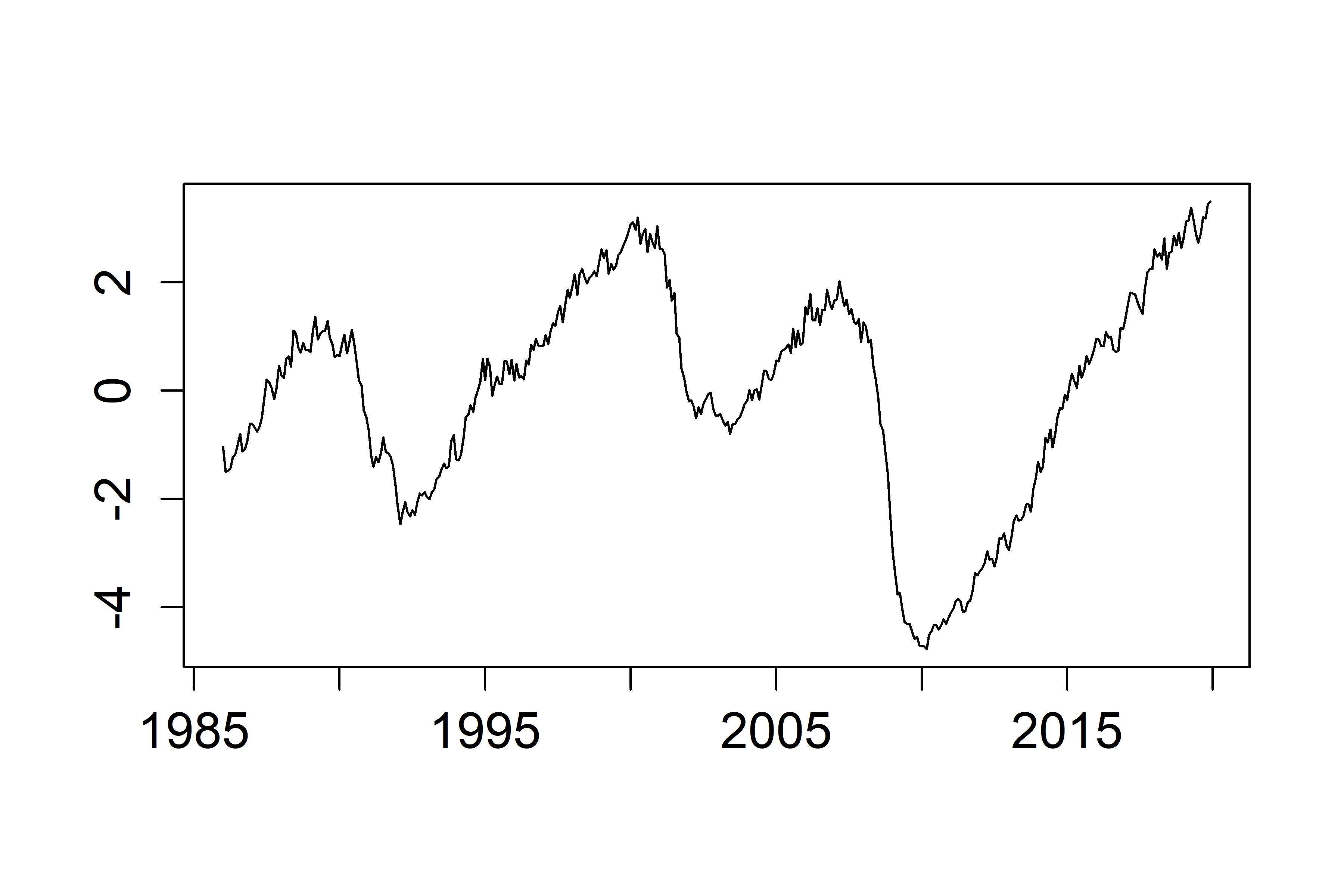}	
				\vspace{-1.5\baselineskip}
				\caption{time series of $\langle \overline{U}_t, \overline{w}_{1} \rangle$}
			\end{subfigure}
		\end{figure}
		
		In practice, one may also be interested in testing various hypotheses about cointegration such as \eqref{hptest1} and \eqref{hptest2}. These hypotheses can be examined by the tests given in Section \ref{sec:inferenceatt} once $\varphi$ (=$\dim(\mathcal H^N)$) is estimated. For illustrative purposes, we consider the following hypotheses, 
		\begin{align*}
			&H_0 : x \in \mathcal H^N \quad \text{against} \quad H_1: x \notin \mathcal H^N, \\ 
			&H_0 :x \in \mathcal H^S \quad \text{against} \quad  H_1: x \notin \mathcal H^S. 
		\end{align*}
		where $x=x_1,x_2$ or $x_3$ given in \eqref{eqexam01}. The test results are  reported in Table \ref{tabl_wage2} under the assumption that ${\varphi}=1$ as we earlier concluded from our sequential procedure. The null hypothesis that $x\in \mathcal H^N$ is rejected in every case at 1\% significance level, which means that each of the considered functions is not entirely included in $\mathcal H^N$ but given by a linear combination of nonzero elements of $\mathcal H^N$ and $\mathcal H^S$; these results are, at least to some extent,  expected from that $\mathcal H^N$ is estimated by $\spn\{\overline{w}_1\}$ and  the considered functions have quite different shapes from that of $\overline{w}_1$; see Figure \ref{figadd2}-(a).  Moreover, the null hypothesis that $x\in \mathcal H^S$  is rejected at $1\%$ significance level for $x=x_1$ or $x_2$, and it is rejected at $5\%$ (but not $1\%$) significance level for $x=x_3$. These test results would lead us to similarly conclude that each of the considered functions is not entirely included in $\mathcal H^S$, but now it is worth noting that the null hypothesis for $x=x_3$ is rejected with  less confidence; that is,  the data supports that $x_3$ is relatively closer to $\mathcal H^S$ than $x_1$ and $x_2$. This result was earlier conjectured from the plots given in Figure \ref{figadd}, and  we here note that such a conjecture has been, at least to some extent, examined by our proposed FPCA-based test.   
		
		
		In this  example, we found a strong evidence that there exists at least one stochastic trend in the time series of age-specific employment rates. It may be of interest to investigate whether this time series is cointegrated with some  economic variables exhibiting a unit root-type nonstationarity. This can be further explored in the future by developing a cointegrating regression model involving functional observations.

			\begin{table}[t!] \renewcommand*{\arraystretch}{1}
				\caption{Test results for $H_0 : x \in \mathcal H^N$ or $H^S$  - age-specific employment rates} \small	\vspace{-1em}
				\label{tabl_wage2}
				\rule[1pt]{1\textwidth}{1.25pt}
				\begin{tabular*}{1\textwidth}{@{\extracolsep{\fill}}cccc}
					$x$  & $x_1$ & $x_2$ & $x_3$   \\ \hline  
					&&& \vspace{-9pt} \\
					Test of $H_0 : x \in \mathcal H^N$ & 0.3128$^{\ast\ast}$ &0.3215$^{\ast\ast}$ &   0.3178$^{\ast\ast}$   \\
					Test of $H_0 : x \in \mathcal H^S$\, &  0.3221$^{\ast\ast}$ &0.3133$^{\ast\ast}$& 0.1960$^{\ast}$
				\end{tabular*}
				\rule{1\textwidth}{1.25pt}  \\
				{\footnotesize Notes: $T=408$. We use ${}^{\ast}$ and  ${}^{\ast\ast}$ to denote rejection at 5\% and 1\%  significance level, respectively; the approximate critical value for 95\% (resp.\ 99\%) is 0.15 (resp.\ 0.22). in each case.} 
			\end{table}
			

				\subsection{Empirical illustration 2: cross-sectional earning densities}\label{sec:emp2}
				We consider a monthly sequence of U.S.\ earning densities running from January 1989 to December 2019; this empirical example is similar to that given by \citet[Section 5.1]{Chang2016152}.  The raw individual weekly earning data can be obtained from the CPS at \url{https://ipums.org}. \nocite{IPUMS2020} Since individual earnings in each month are reported in current dollars, they are all adjusted to January 2000 prices using monthly consumer price index data obtained from the Federal Reserve
				Economic Database at \url{https://fred.stlouisfed.org}. Reported weekly earnings are censored from above, with the threshold for censoring switching partway through the sample: the top-coded nominal earning is 1923 before January 1998, and 2885 afterward. In addition, the raw dataset contains many abnormally small nominal earnings, such as \$0.01  a week. In the subsequent analysis, weekly earnings below the 3th percentile and above the 97th percentile are excluded to avoid potential effects of those abnormal earnings. Each earning density is estimated from the remaining  individual earnings as in \cite{SEO2019} by applying local likelihood density estimation proposed in \cite{loader1996}; this will be more detailed at the end of this section. One may treat such  density-valued observations as random elements of the familiar Hilbert space $L^2$ of square integrable functions to apply our methodology; however, as pointed out by  \cite{petersen2016}, this is not in general advisable since the set of probability densities is not a linear subspace in this case (see also  \citeauthor{delicado2011dimensionality}, \citeyear{delicado2011dimensionality};  \citeauthor{Hron2016330},  \citeyear{Hron2016330}; \citeauthor{kokoszka2019forecasting},  \citeyear{kokoszka2019forecasting}; \citeauthor{zhang2020wasserstein}, \citeyear{zhang2020wasserstein}).  Therefore, we first embed the estimated densities  into a  Hilbert space via the transformation approach proposed in \cite{Egozcue2006}. Let $\mathrm{S} \subset \mathbb{R}$ be the support of the probability density functions and let $|\mathrm{S}|$ be the length of $\mathrm{S}$; in this empirical example $\mathrm{S}$ is set to $[75.36,  1823.94]$ where the left endpoint (resp.\ the right endpoint) corresponds to the minimal (resp.\ the maximal) individual earning over the time span.  We then define
				\begin{align}
					&f_t(u) = \psi(X_t) = \log X_t(u) - \frac{1}{|\mathrm{S}|}\int \log X_t(u)  du, \quad u \in \mathrm{S}. \label{eqtrans1}\end{align}
				Then the transformed series $\{{f}_{t}\}_{t \geq 1}$ are regarded as a time series taking values in $L^2_0[0,1]$, the collection of all $x \in L^2[0,1]$ satisfying $\int x(u)du = 0$, which is a Hilbert space. The transformation $\psi$ given above turns out to be an isomorphism between a Hilbert space of probability density functions (called a Bayes Hilbert space)  and $L_0^2[0,1]$, and its inverse is given by $\psi^{-1}(f) = \exp(f)/\int \exp(f(u))du$ for any $f \in L_0^2[0,1]$; see \cite{Egozcue2006}. For our purpose to illustrate our modified FPCA methodology, we hereafter only concern with the transformed series $\{{f}_{t}\}_{t \geq 1}$; the results obtained in this section can be naturally converted into those for the original density-valued time series $\{X_t\}_{t=1}^T$ using the inverse transformation $\psi^{-1}$ under certain mathematical conditions; see \cite{SEO2019} for a more detailed discussion.  Figure \ref{figadd4} reports the estimated earning densities  and their transformations as our functional observations. 
				
				Given that the functional observations are constructed from weekly earnings which tend to increase over time, it may be reasonable to consider the model with a linear trend as in Section \ref{sec:emp}.  We apply the sequential procedure based on our proposed test to determine $\varphi$.  We use 18 quadratic B-spline functions for the representation of the functional observations, and   $\mathrm{k}(\cdot)$, $h$, and $K$ are set in the same way.  Table \ref{tabl_earning} reports the test results.  Our sequential procedure  concludes that the dimension of $\mathcal H^N$ is $2$ at $1\%$ significant level. Given this result, $\mathcal H^N$ is estimated by the span of the first two eigenvectors $\{\overline{w}_j\}_{j=1}^2$, which are obtained from our modified FPCA. If we let $\overline{f}_t$ denote the residual from detrending $f_t$, it is then  expected that the time series of $\langle \overline{f}_t,\overline{w}_j \rangle$ behaves as a unit root process for $j=1$ and $2$. These are shown in Figure \ref{figadd3}. Once the dimension of $\mathcal H^N$ is estimated, then we may examine various hypotheses about cointegration using the tests developed in Section \ref{sec:inferenceatt}.
				
				\begin{figure}[t!]
					\caption{Estimated earning densities and transformed densities}
					\centering
					\label{figadd4}
					\begin{subfigure}[b]{0.33\textwidth}
						\includegraphics[width=1\textwidth,trim = 0.9cm 0.9cm 0.9cm 2cm,clip]{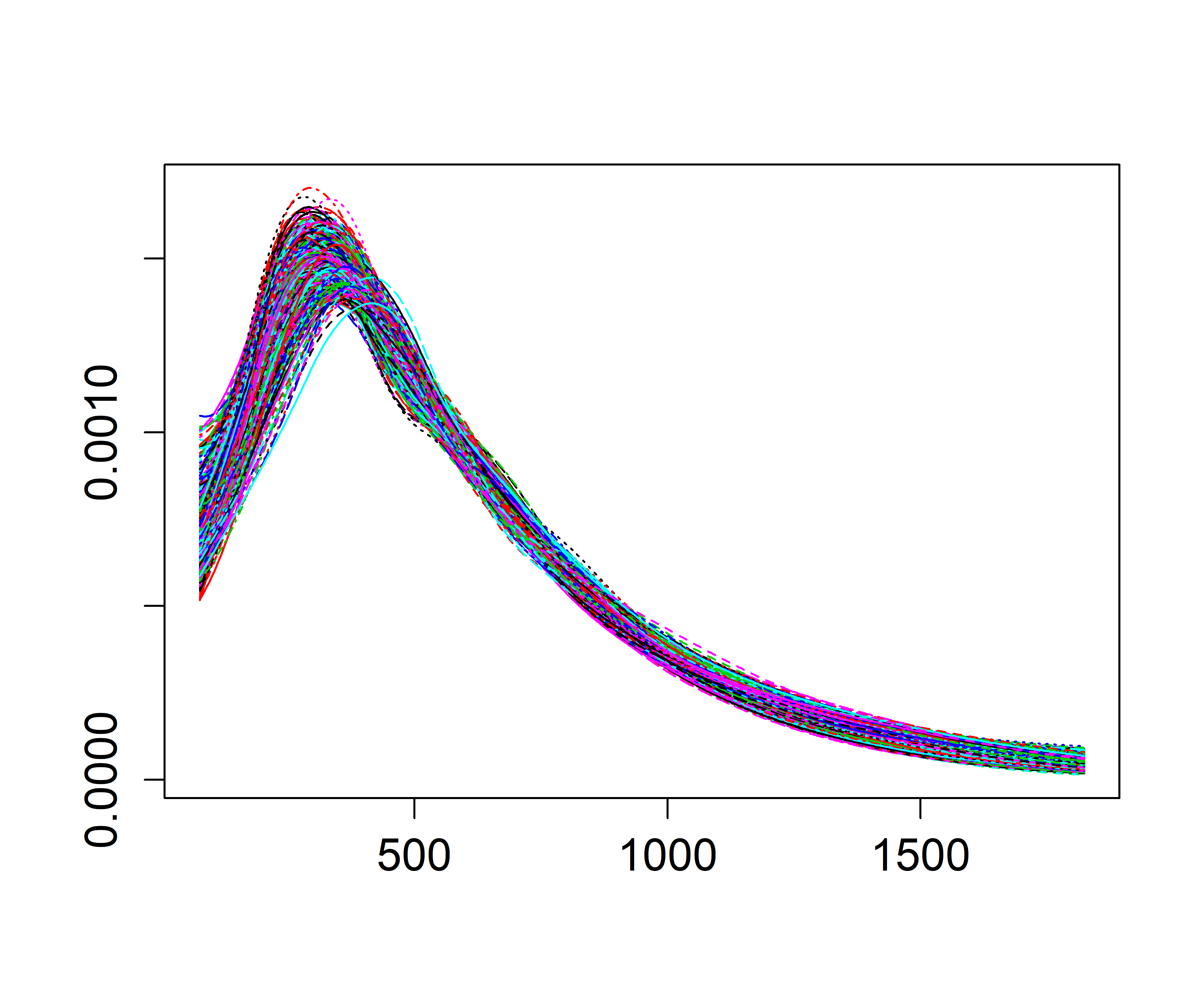}	
						\vspace{-1.5\baselineskip}
						\caption{earning densities}
					\end{subfigure} 
					\begin{subfigure}[b]{0.33\textwidth}
						\includegraphics[width=1\textwidth,trim = 0.9cm 0.9cm 0.9cm 2cm,clip]{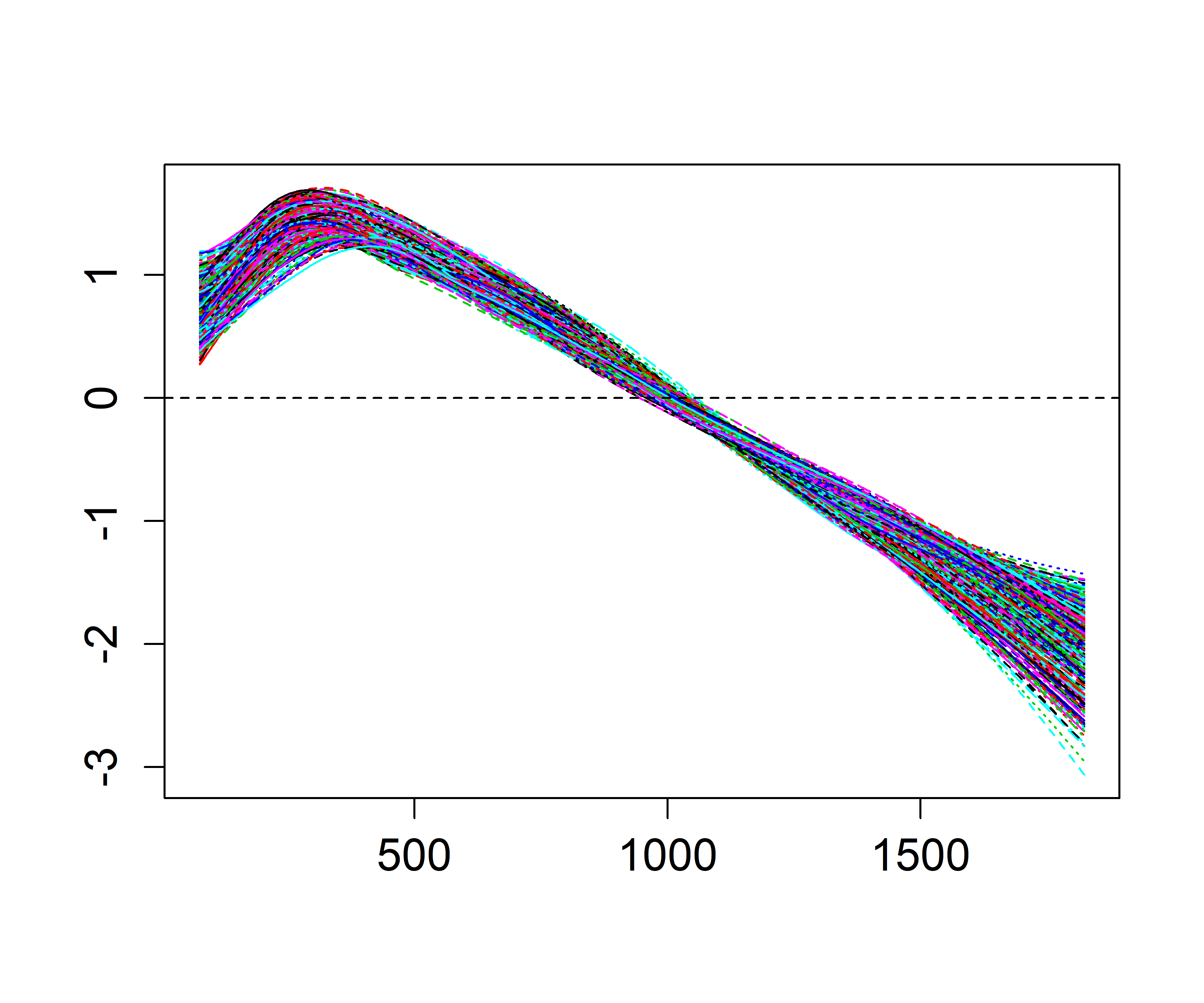}	
						\vspace{-1.5\baselineskip}
						\caption{transformed densities}
					\end{subfigure}
					\end{figure}
					\begin{table}[t!] \renewcommand*{\arraystretch}{1}
						\caption{Test results for the dimension of $\mathcal H^N$ -  earning densities} \small \centering
						\label{tabl_earning}	
						\vspace{-1em}
						\rule[1pt]{0.75\textwidth}{1.25pt}
						\begin{tabular*}{0.725\textwidth}{@{\extracolsep{\fill}}cccc}
							$\varphi_0$ & $0$ & $1$ & $2$  \\ \hline
							Test statistic &  0.4398$^{\ast\ast}$&  {0.2734}$^{\ast\ast}$ & {0.0715}  \\
						\end{tabular*}	
						\rule{0.75\textwidth}{1.25pt} 
						\begin{flushleft}
							{\footnotesize Notes: $T=372$. We use ${}^{\ast}$ and  ${}^{\ast\ast}$ to denote rejection at 5\% and 1\%  significance level, respectively; the approximate critical values for 95\% (resp.\ 99\%) are given by 0.15, 0.12, and 0.10 (resp.\ 0.22, 0.18, and 0.15) sequentially.} 
						\end{flushleft}
					\end{table}
					\begin{figure}[t!]
						\caption{Eigenvectors from the modified FPCA \& their characteristics - earning densities}
						\centering
						\label{figadd3} %
						\begin{subfigure}[b]{0.38\textwidth}
							\includegraphics[width=1\textwidth,trim =  0.2cm 1.5cm 0.7cm 2cm,clip]{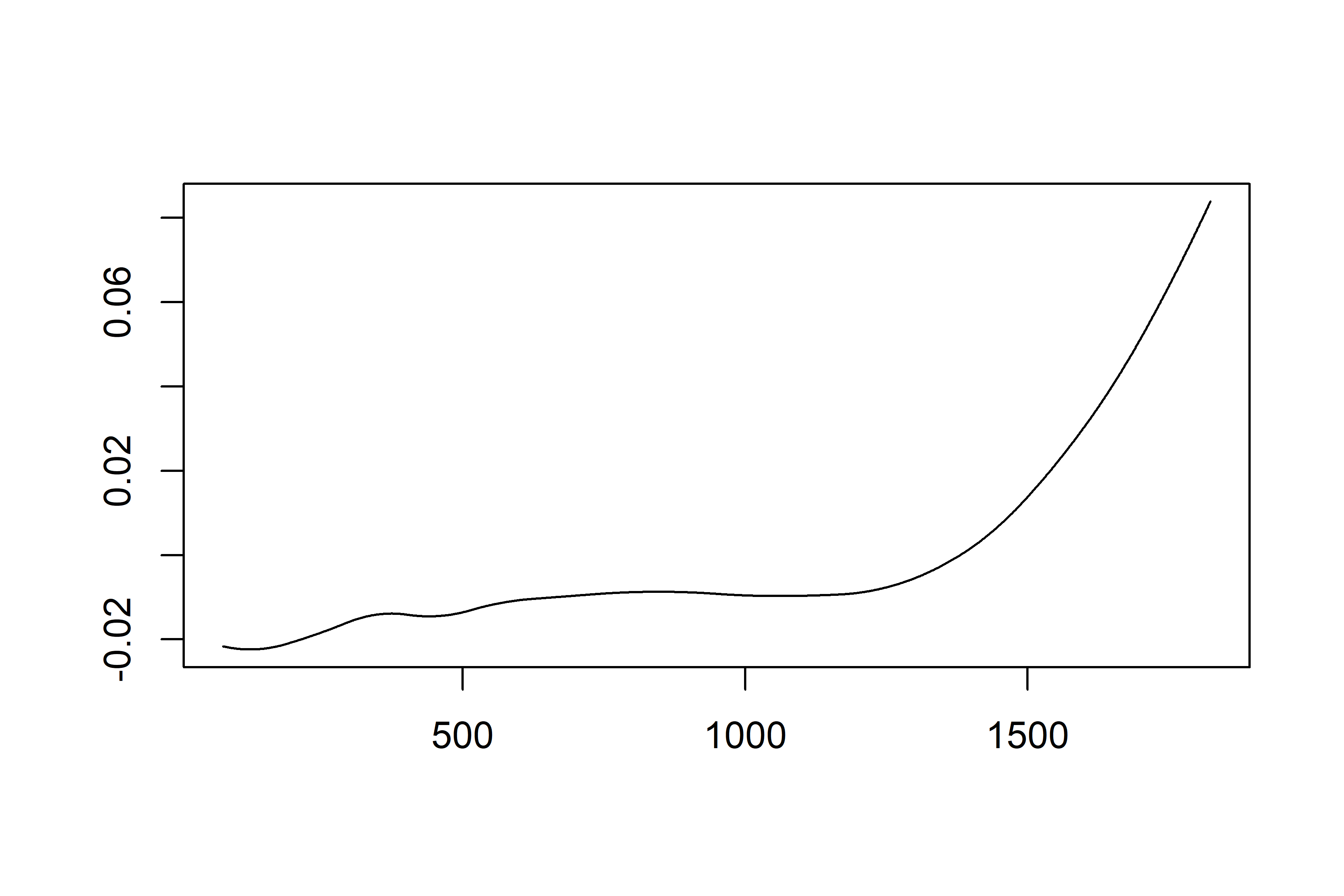}	
							\vspace{-1.\baselineskip}
							\caption{$\overline{w}_{1}$}
						\end{subfigure} 
						\begin{subfigure}[b]{0.38\textwidth}
							\includegraphics[width=1\textwidth,trim  =  0.2cm 1.5cm 0.7cm 2cm,clip]{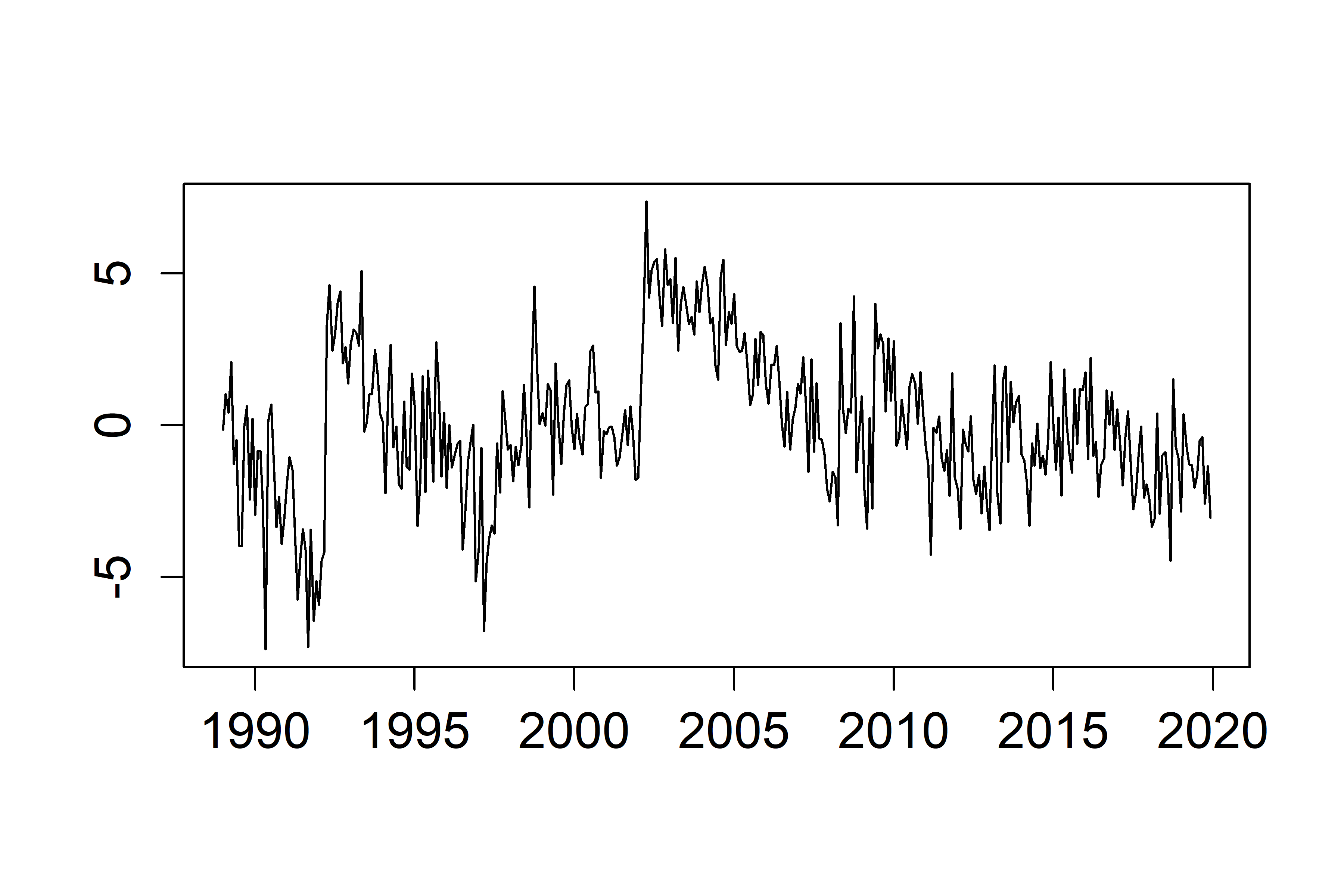}	
							\vspace{-1.\baselineskip}
							\caption{time series of $\langle \overline{f}_{t}, \overline{w}_{1} \rangle$}
						\end{subfigure}
						\begin{subfigure}[b]{0.38\textwidth}
							\includegraphics[width=1\textwidth,trim = 0.2cm 1.5cm 0.7cm 2cm,clip]{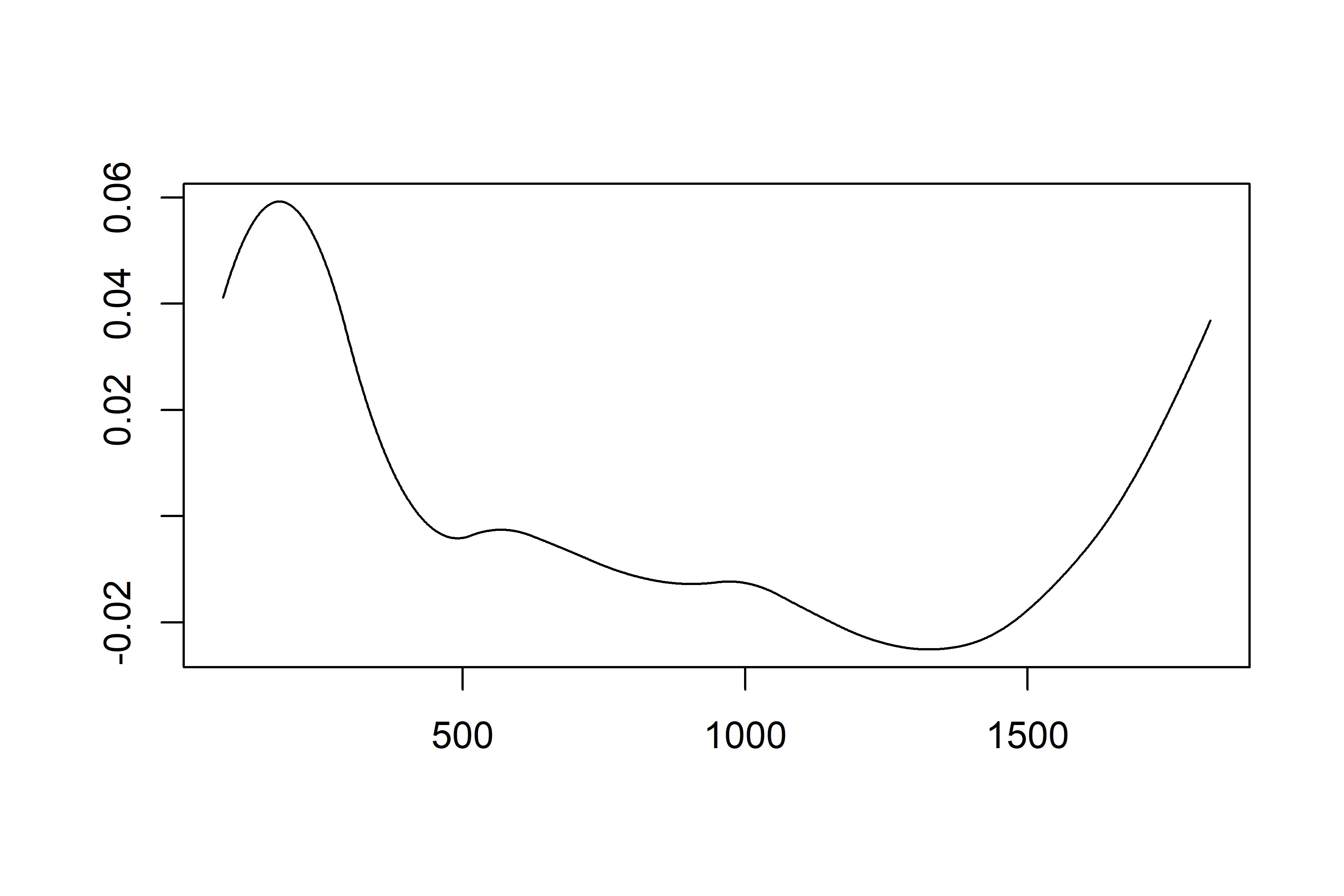}	
							\vspace{-1.\baselineskip}
							\caption{$\overline{w}_{2}$}
						\end{subfigure}
						\begin{subfigure}[b]{0.38\textwidth}
							\includegraphics[width=1\textwidth,trim =  0.2cm 1.5cm 0.7cm 2cm,clip]{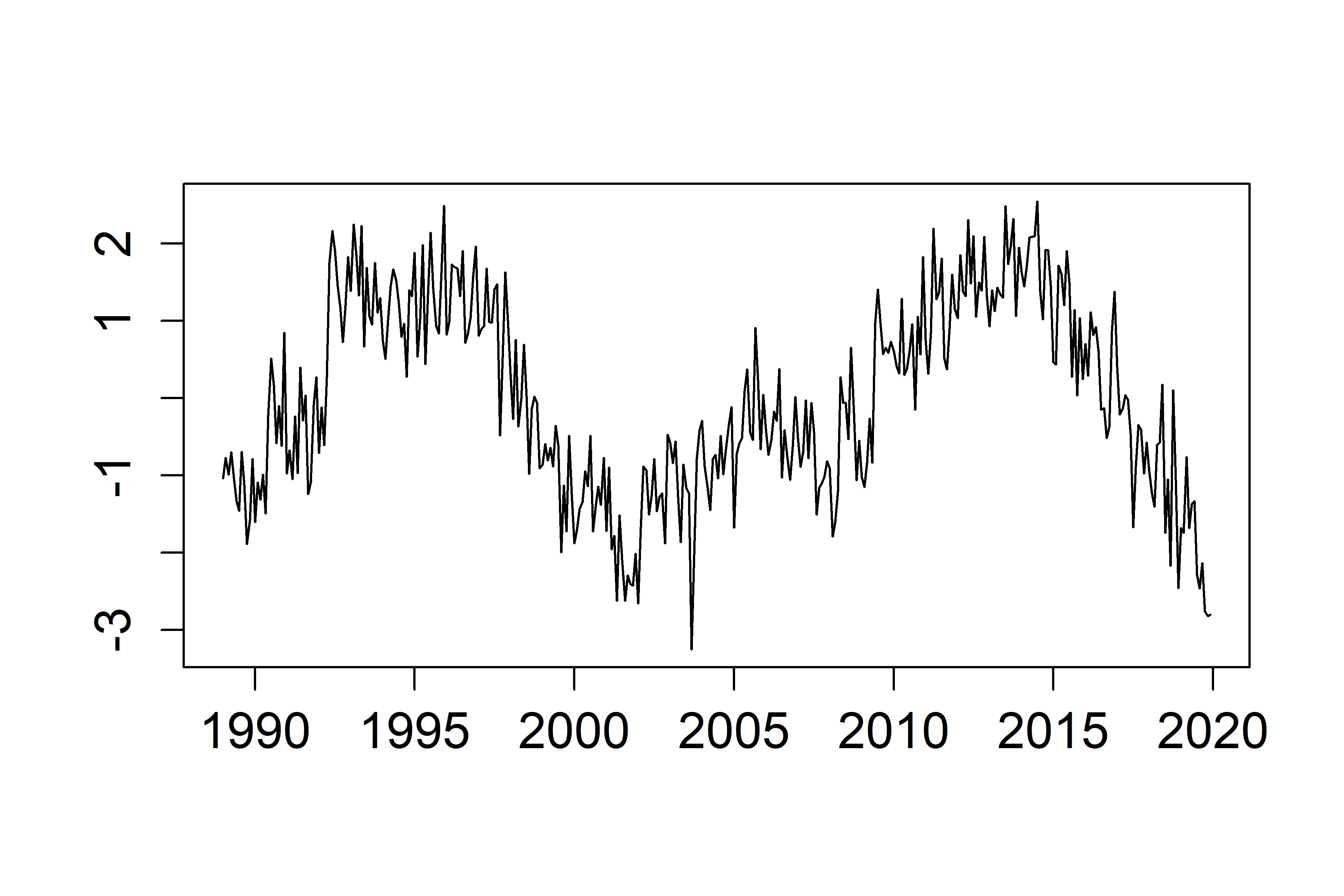}	
							\vspace{-1.\baselineskip}
							\caption{time series of $\langle \overline{f}_{t}, \overline{w}_{2} \rangle$}
						\end{subfigure}
						\vspace{-0.7em}
					\end{figure}
					We in this example found an evidence that the time series of earning densities exhibits stochastic trends. As in the example given in Section \ref{sec:emp},  it may be interesting to investigate if this time series of earning densities is cointegrated with other economic time series, and this can certainly be further explored in the future study. 
					
					\subsubsection*{Estimation of densities of earning} \label{app:logdensity}
					To obtain densities of earning, we employ the local likelihood method  by \cite{loader1996,Loader2006}. Suppose that $X$ is the density of interest which is supported on $\mathrm{S}$, and there are $n$ individual earnings available for estimation of $X$.   Given survey responses $a_1, \ldots, a_{n}$ with design weights $w_1, \ldots, w_{n}$ such that $\sum_{i=1}^n w_i = n$, the weighted log-likelihood is given by $	\mathcal L(X) = \sum_{i=1}^n w_i \log (X(a_i)) - n \left(\int X(u) du - 1 \right)$.
					Under some smoothness assumptions, we can consider a localized version of $\mathcal L(X)$, and $\log X(a)$ can be locally approximated by a polynomial function, as follows.
					\begin{equation} \label{locallf}
						\mathcal L_p(X)(a) = \sum_{i=1}^n w_i  \mathcal W\left(\frac{a_i-a}{b}\right) \mathcal Q(a_i-a ; \varpi)- n \int \mathcal W\left(\frac{u-a}{b}\right) \exp(\mathcal Q(u-a ; \varpi)) du,
					\end{equation} 
					where  $\mathcal W(\cdot)$ is a suitable kernel function,  $b$ is a nearest neighborhood bandwidth ensuring that a fixed percent of the data is included in the local neighborhood of $a$, and $\mathcal Q(u-a ; \varpi)$ is the $q$-th order polynomial in $u-a$ with coefficients $\varpi = (\varpi_{0},\ldots,\varpi_{q})$. Let $(\hat{\varpi}_{0},\ldots,\hat{\varpi}_{q})$ be the maximizer of \eqref{locallf}. The local likelihood log-density estimate  is then given by $\widehat{\log X} (a) = \hat{\varpi}_{0}$	and therefore $\hat{X} (a) =  \exp(\hat{\varpi}_{0})$. The procedure is repeated for a fine grid of points, and then $\hat{X}$ may be obtained from an interpolation method  described in  \citep[Chapter 12]{Loader2006}. Following this  procedure, each earning density is obtained on the common support $\mathrm{S} = [75.35, 1823.94]$ 		in our empirical application in Section \ref{sec:emp2}: to be more specific, we set $\mathcal W(u) = (1-|u|^3)^3 1\{|u|<1\}$,  $\mathcal Q(u-a ; \varpi) = \varpi_{0}  + \varpi_{1} (u-a) + \varpi_{2} (u-a)^2$, and choose $b$ so that $40\%$ of the data is contained in the local neighborhood of $a$. The employed kernel function is called the tricube kernel and used in \cite{loader1996}.

					\newpage 
					
					\noindent \Large{\textbf{Part II : Mathematical Appendix}} \normalsize
					\section{Preliminaries}	\label{appintro0}
					\subsection{Hilbert space and bounded linear operators}\label{appintro01}
					Let $\mathcal H$ denote a real separable Hilbert space with inner product $\langle \cdot, \cdot \rangle$ and norm $\|\cdot\| = \langle \cdot,\cdot \rangle^{1/2}$. We let $\mathcal{L}_{\mathcal H}$ denote the space of bounded  linear operators on $\mathcal H$, equipped with the usual operator norm $\|A\|_{\mathcal L_{\mathcal H}} = \sup_{\|x\|\leq 1} \|Ax\|$. For an operator $A \in \mathcal L_{\mathcal H}$, we let $A^\ast \in \mathcal L_{\mathcal H}$ denote the adjoint of $A$, and let $\ran A$ (resp.\ $\ker A$) denote the range (resp.\ kernel) of $A$, which are, respectively, defined by $\ran A = \{Ax: x \in \mathcal H\}$ and 	$\ker A = \{x\in \mathcal H: Ax = 0\}$.	The rank of $A$, denoted by $\rank A$, is equal to $\dim(\ran A)$.
					If $A=A^*$, then $A$ is said to be self-adjoint. An operator $A \in \mathcal L_{\mathcal H}$ is positive semidefinite if $\langle Ax,x\rangle \geq 0$ for any $x \in \mathcal H$, and positive definite if also $\langle Ax,x\rangle \neq 0$ for any nonzero $x \in \mathcal H$. Throughout this paper, $x\otimes y$ for $x,y \in \mathcal H$ denotes the operator $z \mapsto \langle x,z \rangle y$ of rank one. An operator $A \in \mathcal L_{\mathcal H}$ is said to be compact if there exists two orthonormal bases $\{u_j\}_{j \geq 1}$ and $\{v_j\}_{j \geq 1}$, and a real-valued sequence     $\{a_j\}_{j \geq 1}$ tending to zero, such that  $A = \sum_{j=1}^\infty a_ju_j \otimes v_j$; we may assume that $u_j = v_j$ and  $a_1 \geq a_2 \geq \ldots \geq 0$ if $A$ is also self-adjoint and positive semidefinite \citep[p.\ 35]{Bosq2000}. 
					For any $A \in \mathcal L_{\mathcal H}$ with a closed range, we let $A^\dag\in\mathcal L_{\mathcal H}$ denote the Moore-Penrose inverse; see \cite{Engl1981}.  
					For any self-adjoint, positive semidefinite, and compact  operator $A \in \mathcal L_{\mathcal H}$ satisfying $A = \sum_{j=1}^\infty a_j u_j \otimes u_j$ for some $a_1\geq a_2 \geq \ldots \geq 0$ with $a_{m} > 0$, we hereafter let $A{\mid_{m}^\dagger}= \sum_{j=1}^m a_j^{-1} u_j \otimes u_j$ and be called the $m$-regularized inverse of $A$. $A{\mid_{m}^\dagger}$ is set to $0 \in \mathcal L_{\mathcal H}$ if $m = 0$. The $m$-regularized inverse $A{\mid_{m}^\dagger}$ is understood as the partial inverse of $A$ on the restricted domain $\spn(\{u_j\}_{j=1}^m)$. If $\rank A = m$, then $A{\mid_{m}^\dagger} = A^\dag$. 
					\subsection{Random elements in $\mathcal H$ and $\mathcal L_{\mathcal H}$} \label{appintro}
					Let $(\mathcal S,\mathcal F, \mathcal P)$ denote the underlying probability triple. An $\mathcal H$-valued random variable $X$ is defined as a measurable map from $\mathcal S$ to $\mathcal H$, where $\mathcal H$ is understood to be equipped with the usual Borel $\sigma$-field.  An $\mathcal H$-valued random variable $X$ is said to be integrable (resp.\ square integrable) if $\mathbb{E} \|X\| < \infty$ (resp.\  $\mathbb{E} \|X\|^2 < \infty$). If $X$ is integrable, there exists a unique element $\mathbb{E}X \in \mathcal H$ satisfying $\mathbb{E}\langle X,x\rangle = \langle \mathbb{E}X,x\rangle$ for any $x \in \mathcal H$. The element $\mathbb{E}X$ is called the expectation of $X$. Let $L^2_{\mathcal H}$ denote the space of $\mathcal H$-valued random variables $X$ (identifying random elements that are equal  almost surely) that satisfy $\mathbb{E}X = 0$ and $\mathbb{E}\|X\|^2 < \infty$. For any $X \in L^2_{\mathcal H}$, we may define its covariance operator as $C_X = \mathbb{E} \left[  X \otimes X  \right]$, which is guaranteed to be self-adjoint, positive semidefinite and compact.
					
					Now let $A$ be a map from $\mathcal S$ to $\mathcal L_{\mathcal H}$ such that $\langle A x,y \rangle$ is Borel measurable for all $x,y \in \mathcal H$. Such a map $A$ is called a random bounded linear operator; see  \cite{skorohod2001}. For such operators $A$ and $B$, we write $A =_{\fdd} B$ if $(\langle Ax_1,y_1 \rangle,\ldots,\langle Ax_k,y_k \rangle) =_d (\langle Bx_1,y_1 \rangle,\ldots,\langle Bx_k,y_k \rangle)$
				for every $k\geq 1$ and $x_1,\ldots,x_k$, $y_1,\ldots,y_k \in \mathcal H$.
				For any sequence $\{A_j\}_{j \geq 1}$ in $\mathcal L_{\mathcal H}$, we write $A_j \to_{\mathcal L_{\mathcal H}} A$, if  $\|A_j - A\|_{\mathcal L_{\mathcal H}}  \to_p  0$.  
				It will be convenient to define two other modes of convergence of random bounded linear operators for our  proofs of the main results. First, we write $A_j \to_w A$, if for all $x,y \in \mathcal H$, 
				\begin{equation*} |\langle A_jx,y \rangle  - \langle Ax,y \rangle |  \to_p   0. 
				\end{equation*}  Moreover, we write $A_j \to_{wd} A$ if, for all $k$ and $x_1,\ldots,x_k$, $y_1,\ldots,y_k \in \mathcal H$,  
				\begin{equation*}
					\lim_{j\to \infty}\mathbb{E} f\left(\langle A_jx_1,y_1 \rangle,\ldots,\langle A_jx_k,y_k \rangle \right)  = \mathbb{E} f\left(\langle Ax_1,y_1 \rangle,\ldots,\langle Ax_k,y_k\rangle \right)
				\end{equation*}  for any bounded continuous function $f$. When $A_j \to_w A$ (resp.\ $A_j \to_{wd} A$), 
				$A_j$ is said to converge weakly (resp.\ weakly in distribution) to $A$; these modes of convergence are introduced in detail by \cite{skorohod2001}. Note that $A_j \to _{wd} A$ means that the finite dimensional distribution of $A_j$ converges to that of $A$. Furthermore, it can be shown that $$A_j \to_{\mathcal L_{\mathcal H}} A \quad \Rightarrow \quad A_j \to_{w} A \quad \Rightarrow \quad A_j \to_{wd} A.$$ 
				

				\subsection{Convergence of random bounded linear operators} \label{app:conv}
				We collect some useful results on convergence of random bounded linear operators. 
				\begin{lemma}\label{applem1} Let $\{A_j\}_{j \geq 1}$ and $\{B_j\}_{j \geq 1}$ be sequences of random bounded linear operators. Then the following hold.
					\begin{itemize}  \setlength\itemsep{0em}
						\item[$\mathrm{(i)}$] If $A_j \to_{wd} A$, then $\sup_j\|A_j\| = \mathrm{O}_p(1)$.	 
						\item[$\mathrm{(ii)}$] If $A_j \to_{\mathcal L_{\mathcal H}} A$ and $B_j \to_{\mathcal L_{\mathcal H}} B$ , then $A_jB_j \to_{\mathcal L_{\mathcal H}} AB$. 
						\item[$\mathrm{(iii)}$] If $A_j \to_{\mathcal L_{\mathcal H}} A$, then $A_j^\ast \to_{\mathcal L_{\mathcal H}} A^\ast$.
					\end{itemize}
				\end{lemma}	
				\begin{proof}  (i) is given in \cite{skorohod2001}; see, the proof of Theorem 3 in Ch 1.  
					
					Note  that $\|A_jB_j-AB\|_{\mathcal L_{\mathcal H}} \leq \|A_j-A\|_{\mathcal L_{\mathcal H}} \|B_j\|_{\mathcal L_{\mathcal H}} + \|A\|_{\mathcal L_{\mathcal H}}\|B_j-B\|_{\mathcal L_{\mathcal H}}$. Using (i), we may deduce that the right hand side converges to zero, hence (ii) is established. 
					
					(iii) immediately follows from the fact that $\|A_j - A\|_{\mathcal L_{\mathcal H}} = \|A_j^\ast - A^\ast\|_{\mathcal L_{\mathcal H}}$.			
				\end{proof}
				
				\begin{lemma}\label{applem2} Let $\{A_j\}_{j \geq 1}$ be a random sequence in $\mathcal L_{\mathcal H}$ and for $j \geq 1$ $\{\Psi_{A_j}\}_{j \geq 1}$ be their characteristic functionals that are given by $\Psi_{A_j} (\mathcal U) = \mathbb{E} \exp\left( \ii \Tr(A_j \mathcal U) \right)$, where $\mathcal U$ is any operator of the form 
					\begin{equation}
						\mathcal U = \sum_{i=1}^k x_i \otimes y_i, \quad x_i, y_i \in \mathcal H, \quad i=1,\ldots,k,\quad k=1,2, \ldots. \label{equps01}
					\end{equation}
					Then the following hold.
					\begin{itemize} \setlength\itemsep{0em}
						\item[$\mathrm{(i)}$]	$\Psi_{A_j} (\mathcal U) \to \Psi_{A} (\mathcal U) \, \text { for all $\mathcal U$ of the form \eqref{equps01}}$ if and only if  $A_j \to_{wd} A$.  
						\item[$\mathrm{(ii)}$] (Cram\'{e}r-Wold device) $A_j \to_{wd} A$ if and only if  
						$\sum_{i=1}^k \langle A_jy_i,x_i \rangle  \to_d \sum_{i=1}^k \langle Ay_i,x_i \rangle$ 
						for any $k$, $x_1,\ldots,x_k \in \mathcal H$ and $y_1,\ldots,y_k \in \mathcal H$.  
						\item[$\mathrm{(iii)}$] If $\Psi_{A}(\alpha (x\otimes y))$ is continuous in $\alpha$  for all $x,y \in \mathcal H$, then there exists a sequence of random operators $\{\widetilde{A}_j\}_{j \geq 1}$ and $\widetilde{A}$ such that (i) $	\Psi_{\widetilde{A}_j}(\cdot) = \Psi_{A_j}(\cdot)$, (ii) $\Psi_{\widetilde{A}}(\cdot) = \Psi_{A}(\cdot)$, and $A_j \to_w  \widetilde{A}$.
					\end{itemize}
				\end{lemma}
				\begin{proof}
					The results can be found in or deduced from Chapter 3.3 of \cite{skorohod2001}. 
				\end{proof}
				
				\section{Mathematical Proofs}\label{appproof}
				Let $\{A_j\}_{j\geq 1}$ be a sequence in the space equipped with norm $\|\cdot\|_{\mathcal B}$ and satisfy that $\|A_j-A\|_{\mathcal B} = \mathrm{O}_p(n)$. 	As is common in the literature, we sometimes, for convenience, write this as  $A_j = A + \mathrm{O}_p(n)$. We similarly write $A_j = A + \mathrm{o}_p(n)$ to denote $\|A_j-A\|_{\mathcal B} = \mathrm{o}_p(n)$.  
				\subsection{Mathematical proofs for the results in Sections \ref{sec:fpca} and  \ref{sec:harris}}  \label{appproof1}
				\subsubsection{Preliminary results} \label{appsecproof1}
				We provide useful lemmas.  
			
				\begin{lemma}\label{lem1} Under Assumptions \ref{assumm0} for every $k\geq 1$ and $x_1,\ldots,x_k \in \mathcal H$, 
					\begin{equation}
						\frac{1}{\sqrt{T}} \sum_{t=1}^{\lfloor Ts\rfloor} \mathcal E_{k,t}   \quad \to_d \quad   W_k(s) \label{eqweakconv}
					\end{equation}
					in the Skorohod space $\mathcal D[0,1]^k$, where  $\mathcal E_{k,t} = (\langle \mathcal E_{t}, x_1 \rangle, \ldots, \langle \mathcal E_{t}, x_k \rangle)'$ and $W_k(s) = (\langle W(s),x_1 \rangle,$ $\ldots,\langle W(s),x_k \rangle)'$.
					\begin{proof}
						Under the summability conditions $\sum_{j=1}^{\infty}j\|{\Phi}_j\|_{\mathcal L_{\mathcal H}}<\infty$ and $\sum_{j=1}^{\infty}j\|\widetilde{\Phi}_j\|_{\mathcal L_{\mathcal H}}<\infty$,  	$\{\mathcal E_t^N\}_{t\in \mathbb{Z}}$ and 	$\{\mathcal E_t^S\}_{t\in \mathbb{Z}}$ are, respectively, so-called $L^2$-m-approximable \citep[Proposition 2.1]{hormann2010}, and then $\{\mathcal E_t\}_{t\in \mathbb{Z}}$ is also $L^2$-m-approximable \citep[Lemma 2.1]{hormann2010}. Then the desired weak convergence result follows from Theorem 1.1. of \cite{berkes2013weak}.	
					\end{proof} 
				\end{lemma}

				\begin{lemma}  \label{lem2} Under Assumptions \ref{assumm0} and \ref{assum2}, $\widehat{\Omega}_{\varphi}$ and $\widehat{\Gamma}_{\varphi}$  in \eqref{omegae}  are consistent, i.e.,	$	\widehat{\Omega}_{\varphi} \to_{\mathcal L_{\mathcal H}} \Omega$ and $\widehat{\Gamma}_{\varphi} \to_{\mathcal L_{\mathcal H}}  \Gamma.$	Moreover, let $\overline{\Omega}_{\varphi}$ and  $\overline{\Gamma}_{\varphi}$ be defined as in Section \ref{sec:deterministic}. Under the same assumptions,   $\overline{\Omega}_{\varphi}$ and   $\overline{\Gamma}_{\varphi}$ are consistent in the same sense. 
				\end{lemma}
				\begin{proof}
					We know from Theorem \ref{prop1} that $\widehat{P}^N_{\varphi}  = P^N + \mathrm{O}_p(T^{-1})$ and  $\widehat{P}^S_{\varphi}  = P^S + \mathrm{O}_p(T^{-1})$, hence \begin{equation*} 
						\widehat{\Omega}_{\varphi} = \widehat{\Omega}^{NN}_{\varphi} + \widehat{\Omega}^{NS}_{\varphi} + \widehat{\Omega}^{SN}_{\varphi} + \widehat{\Omega}^{SS}_{\varphi}  =  \widehat{\Omega}_{0,\varphi}  + \mathrm{O}_p(T^{-1}),
					\end{equation*} where $\widehat{\Omega}_{0,\varphi}  = P^N\widehat{\Omega}_{\varphi}  P^N + P^N\widehat{\Omega}_{\varphi}  P^S  + P^S\widehat{\Omega}_{\varphi}  P^N +  P^S\widehat{\Omega}_{\varphi} P^S$. Note that $\widehat{\Omega}_{0,\varphi}$ is the sample long-run covariance of $\{\mathcal E_t\}_{t =1}^T$, where $\mathcal E_t = \mathcal E^N_t + \mathcal E^S_t$, $\mathcal E^N_t = \sum_{j=0}^\infty P^N \Phi_j \varepsilon_{t-j}$, and $\mathcal E^S_t =  \sum_{j=0}^\infty P^S \widetilde{\Phi}_j \varepsilon_{t-j}$. 
					We know from our proof of Lemma \ref{lem1} that the summability conditions $\sum_{j=1}^{\infty}j\|{\Phi}_j\|_{\mathcal L_{\mathcal H}}<\infty$ and $\sum_{j=1}^{\infty}j\|\widetilde{\Phi}_j\|_{\mathcal L_{\mathcal H}}<\infty$ imply that   
					$\{\mathcal E_t\}_{t\in \mathbb{Z}}$ is $L^2$-m-approximable. We then apply Theorem 2 of \cite{horvath2013estimation} to obtain $\widehat{\Omega}_{0,\varphi}   \to_{\mathcal L_{\mathcal H}} \Omega$, which in turn establishes $\widehat{\Omega}_{\varphi}  \to_{\mathcal L_{\mathcal H}} \Omega$. 
					To show $\widehat{\Gamma}_{\varphi}  \to_{\mathcal L_{\mathcal H}}  \Gamma$, we note that $\widehat{\Gamma}_{\varphi}  = \widehat{\Gamma}_{0,\varphi}  + \mathrm{O}_p(T^{-1})$ where $\widehat{\Gamma}_{0,\varphi} = P^N\widehat{\Gamma}_{\varphi}  P^N + P^N\widehat{\Gamma}_{\varphi}  P^S  + P^S\widehat{\Gamma}_{\varphi}  P^N +  P^S\widehat{\Gamma}_{\varphi}  P^S$. Then  $\widehat{\Gamma}_{\varphi}  \to_{\mathcal L_{\mathcal H}} \Gamma$ may be deduced from the proof of Theorem 2 in \cite{horvath2013estimation}.
					
					For the cases with deterministic terms, the desired results are deduced from Theorem 2 of \cite{horvath2013estimation} (when Model D1 is true), Theorem 5.3 of \cite{kokoszka2016kpss} (when Model D2 is true),   and our previous proof for the case without deterministic terms.
				\end{proof}
				
				
				\subsubsection{Proofs of the main results} \label{appsecproof2}
				\paragraph*{\normalfont \textbf{Proof of Theorem \ref{prop1}.}}  We note the identity 
				\begin{equation} \label{eq001}
					\widehat{P}^N_{\varphi}  - P^N = P^S\widehat{P}^N_{\varphi} - P^N\widehat{P}^S_{\varphi}.
				\end{equation}
				Since $\widehat{P}^N_{\varphi}$ is the projection onto the first $\varphi$ leading eigenvectors, $P^N\widehat{C}\widehat{P}^S_{\varphi}= 	P^N\widehat{C}P^N\widehat{P}^S_{\varphi} + P^N\widehat{C}P^S\widehat{P}^S_{\varphi}  = P^N\widehat{C}(I-\sum_{j=1}^{\varphi} \langle \hat{v}_{j},\cdot \rangle \hat{v}_{j} ) =P^N\widehat{P}^S_{\varphi} \widehat{\Lambda}$,
				where $\widehat{\Lambda} = \sum_{j=1}^{\infty}  \hat{\lambda}_{j} \hat{v}_{j} \otimes \hat{v}_{j}$. 
				We  thus find that
				\begin{equation}
					T P^N\widehat{P}^S_{\varphi}  = -\left(T^{-1}P^N  \widehat{C} P^N\right)^\dag  P^N  \widehat{C} P^S + \left(T^{-1}P^N  \widehat{C} P^N\right)^\dag   P^N\widehat{P}^S_{\varphi}  \widehat{\Lambda}, \label{eq002}
				\end{equation}	
				where $(T^{-1}P^N  \widehat{C} P^N)^\dag$ is well defined since $T^{-1}P^N  \widehat{C} P^N$ is a finite rank operator, hence has a closed range. We first deduce from the weak convergence result given in \eqref{eqweakconv} and the continuous mapping theorem that the following holds under Assumption \ref{assumm0}: for $k \geq 1$, $x_1,\ldots,x_k \in \mathcal H$ and $y_1,\ldots,y_k \in \mathcal H$,
				\begin{equation*} 
					\langle T^{-1} P^N\widehat{C}P^N x_j, y_j \rangle = \frac{1}{T^2} \sum_{t=1}^T \langle P^N X_t, x_j \rangle  \langle P^N X_t,y_j \rangle   \to_d \int \langle W^N(r), x_j \rangle \langle W^N(r),y_j \rangle, \quad \text{jointly.}
				\end{equation*}
				This result implies that 
				\begin{equation} \label{eqref0}
					T^{-1} P^N\widehat{C}P^N   \to_{wd}   \int W^N(r) \otimes W^N(r),
				\end{equation}
				due to the Cram\'{e}r-Wold device (Lemma \ref{applem2}-(ii)). Let $X_{k,t}^S = (\langle P^SX_t, x_1 \rangle, \ldots, \langle P^SX_t, x_k \rangle)'$,   $X_{k,t}^N = (\langle P^NX_t, y_1 \rangle, \ldots, \langle P^NX_t, y_k \rangle)'$. We  deduce from the  convergence result given by Assumption \ref{assumW} that 
				
				\begin{equation*} 
					\frac{1}{T} \sum_{t=1}^T X_{k,t}^S  X_{k,t}^{N}\,'  \to_d  \int_{0}^1 d W_k^S(s)  W_k^N(s)' + \sum_{j=0}^\infty \mathbb{E}[\mathcal E^S_{k,t-j}\mathcal E^N_{k,t}\,'],
				\end{equation*} 
				where $\mathcal E^S_{k,t} = (\langle \mathcal E_{t}^S, x_1 \rangle, \ldots, \langle \mathcal E_{t}^S, x_k \rangle)'$, $\mathcal E^N_{k,t} = (\langle \mathcal E_{t}^N, y_1 \rangle, \ldots, \langle \mathcal E_{t}^N, y_k \rangle)'$, $W_k^S(s) =  (\langle W^S(s), x_1 \rangle,$ $\ldots, \langle W^S(s), x_k \rangle)'$ and $W_k^N(s) =  (\langle W^N(s), y_1 \rangle, \ldots, \langle W^N(s), y_k \rangle)'$. For any $x,y \in \mathcal H$, $\langle P^N\widehat{C}P^S x, y \rangle = T^{-1} \sum_{t=1}^T  \langle P^S X_t, x \rangle \langle P^N X_t,y \rangle $, from which we find that $\sum_{j=1}^k \langle P^N\widehat{C}P^S x_j, y_j \rangle = \tr(T^{-1}\sum_{t=1}^T X_{k,t}^S  X_{k,t}^{N}\,')$. From the continuous mapping theorem, we have
				\begin{equation}
					\sum_{j=1}^k	\langle P^N\widehat{C}P^S x_j, y_j \rangle \to_d \sum_{j=1}^k \left( \int \langle d W^S(r), x_j \rangle \langle W^N(r),y_j \rangle + \langle \Gamma^{NS}x_j,y_j \rangle\right). \label{eqadd0001}
				\end{equation}
				From \eqref{eqadd0001} and the Cram\'{e}r-Wold device (Lemma \ref{applem2}-(ii)), we deduce that 
				\begin{equation} \label{eqref1}
					P^N\widehat{C}P^S   \to_{wd}   \int dW^S(r) \otimes W^N(r) +  \Gamma^{NS}.
				\end{equation}
				It will be shown that the convergence results given by \eqref{eqref0} and \eqref{eqref1} can be strengthened as follows: \\ \vspace{-0.7em}
				\begin{claim}
					$T^{-1}P^N  \widehat{C} P^N \to_{\mathcal L_{\mathcal H}} \mathcal A_1=_{\fdd} \int W^N(r) \otimes W^N(r)$.	
				\end{claim}\\\vspace{-0.7em}
				\begin{claim2}
					$P^N  \widehat{C} P^S \to_{\mathcal L_{\mathcal H}} \mathcal A_2=_{\fdd} \int dW^S(r) \otimes W^N(r) +  \Gamma^{NS}$.		
				\end{claim2}\\\vspace{-0.7em}
				
				\noindent Almost surely, the eigenvalues of $\mathcal A_1$ are   distinct and each eigenvalue is associated with one-dimensional eigenspace. We then know from Claim 1 and Lemma 4.3 of \cite{Bosq2000}  that the eigenvectors of $T^{-1}P^N  \widehat{C} P^N$ converge to those of $\mathcal A_1$. From similar arguments to the proofs of Proposition 3.2 and Theorem 3.3 of \cite{Chang2016152}, we find that  $P^N \widehat{P}^S_{\varphi} = \mathrm{O}_p (T^{-1})$, $P^S \widehat{P}^S_{\varphi} - P^S= \mathrm{O}_p(T^{-1})$ and $\hat{\lambda}_{\varphi+k} = \mathrm{O}_p(1)$ for $k \geq 1$. Given these results, the following can be additionally proved. \\ \vspace{-0.7em}
				
				\begin{claim3}
					$TP^S\widehat{P}^N_{\varphi} = - T(P^N\widehat{P}^S_{\varphi})^\ast + \mathrm{O}_p(T^{-1})$.
				\end{claim3}\\ \vspace{-0.7em}
				
				\noindent Since $P^N \widehat{P}^S_{\varphi} = \mathrm{O}_p (T^{-1})$, \eqref{eq002} can be written as $T P^N\widehat{P}^S_{\varphi} = -\left(T^{-1}P^N  \widehat{C} P^N\right)^\dag  P^N  \widehat{C} P^S + \mathrm{O}_p(T^{-1}).$ From this equation, Claim 3 and \eqref{eq001}, we find that  
			\begin{equation*}
				T(\widehat{P}^N_{\varphi} - P^N) = \left(T^{-1}P^N  \widehat{C} P^N\right)^\dag  P^N  \widehat{C} P^S + P^S  \widehat{C} P^N \left(T^{-1}P^N  \widehat{C} P^N\right)^\dag  + \mathrm{O}_{p}(T^{-1}).
			\end{equation*}
			We know from Claims 1-2 and Lemma \ref{applem1} that our proof becomes complete if we show
			\begin{equation}
				\left(T^{-1}P^N  \widehat{C} P^N\right)^\dag   \to_{\mathcal L_{\mathcal H}}  \mathcal A_1^\dag.\label{show1} 
			\end{equation}
			Note that both $T^{-1}P^N \widehat{C} P^N$ and $\mathcal A_1$ are self-adjoint positive definite operators of rank $\varphi$, almost surely. Consider the spectral representations of $T^{-1}P^N \widehat{C} P^N$  and  $\mathcal A_1$ as follows:
			\begin{equation*}
				T^{-1}P^N \widehat{C} P^N = \sum_{j=1}^\varphi \hat{\gamma}_{j} \hat{u}_{j} \otimes \hat{u}_{j}, \quad \mathcal A_1= \sum_{j=1}^\varphi \gamma_j u_j \otimes u_j.
			\end{equation*}	
			$T^{-1}P^N \widehat{C} P^N \to_{\mathcal L_{\mathcal H}} \mathcal A_1$ implies that the eigenvalues of $T^{-1}P^N \widehat{C} P^N$ converge to those of $\mathcal A_1$ \citep[Lemma 4.2]{Bosq2000}, i.e., 
			\begin{equation}\label{argeq1a}
				\sup_{1\leq j\leq \varphi} |\hat{\gamma}_{j}- \gamma_j|   \to_p   0. 
			\end{equation}
			Note that $\gamma_j > 0$ and the associated eigenspace is one-dimensional for each $j=1,\ldots,\varphi$ almost surely. It is deduced from Lemma 4.3 of \cite{Bosq2000} that the eigenvectors satisfy
			\begin{equation} \label{argeq2}
				\sup_{1\leq j\leq \varphi}\| \hat{u}_{j} - \sgn(\langle \hat{u}_{j},u_j \rangle)u_j \|   \to_p    0.
			\end{equation}
			Let $\bar{u}_j=\sgn(\langle \hat{u}_{j},u_j \rangle)u_j$. Then it can be shown that $\sup_{\|x\|\leq 1}\|(T^{-1}P^N \widehat{C} P^N)^\dag x - \mathcal A_1^\dag x \|$ is bounded above by 
			\begin{equation}
				\sum_{j=1}^\varphi |\gamma_j^{-1}-\hat{\gamma}_{j}^{-1}|  + 2 \sum_{j=1}^\varphi \hat{\gamma}_{j}^{-1} \|\bar{u}_j-\hat{u}_{j}\|,   \label{eqadd01}
			\end{equation}
			We then deduce from \eqref{argeq1a} and \eqref{argeq2} that \eqref{eqadd01} converges in probability to zero as desired, so \eqref{show1} holds.\\
			

			\noindent \textbf{Proofs of Claims 1 \& 2} : We only provide our proof of Claim 2; the arguments to be given can be applied to prove Claim 1 with a minor modification. To simplify expressions, we let $\widehat{A} = P^N\widehat{C}P^S$ and $ \widehat{A}^\ast = P^S\widehat{C}P^N$. From \eqref{eqref1} and Lemma \ref{applem2}-(iii), we may assume that $\widehat{A} \to_{w} \mathcal A_2 =_{\fdd} \int dW^S(r) \otimes W^N(r) +  \Gamma^{NS}$ and $\widehat{A}^\ast \to_w \mathcal A_2^\ast$. Let $\{u_j\}_{j=1}^{\varphi}$ (resp.\ $\{u_j\}_{j=\varphi+1}^{\infty}$) denote an orthonormal basis of $\mathcal H^N$ (resp.\ $\mathcal H^S$). Since $\widehat{A} \to_{w} \mathcal A_2$,  $\ran \widehat{A}\subset {\mathcal H}^N$ and $\ran \mathcal A_2 \subset \mathcal H^N$, we have for any $x \in \mathcal H$
			\begin{equation}
				\|(\widehat{A}-\mathcal A_2)x\|^2 = \sum_{j=1}^\varphi |\langle (\widehat{A}-\mathcal A_2)x, u_j  \rangle|^2 \to 0, \label{lemeqh00}
			\end{equation} 
			where the equality follows from  Parseval's identity. From the properties of operator norm, we have   
			\begin{equation}\label{lemeqh}
				\|\widehat{A} - \mathcal A_2\|_{\mathcal L_{\mathcal H}}^2 = \|\widehat{A} \widehat{A}^\ast - \mathcal A_2 \widehat{A}^\ast - \widehat{A} \mathcal A_2^\ast + \mathcal A_2 \mathcal A_2^\ast\|_{\mathcal L_{\mathcal H}}. 
			\end{equation}
			Note that $\widehat{A} \widehat{A}^\ast - \mathcal A_2 \widehat{A}^\ast - \widehat{A} \mathcal A_2^\ast + \mathcal A_2 \mathcal A_2^\ast$ is self-adjoint, positive semidefinite and has finite rank. Moreover, its operator norm is bounded above by the trace norm  (see equation (1.55) in \citeauthor{Bosq2000}, \citeyear{Bosq2000}), which is in turn bounded by $\sum_{j=1}^\infty | \langle (\widehat{A} \widehat{A}^\ast - \mathcal A_2 \widehat{A}^\ast - \widehat{A} \mathcal A_2^\ast + \mathcal A_2 \mathcal A_2^\ast)u_j,u_j \rangle |$. Since  $\ran (\widehat{A} \widehat{A}^\ast - \mathcal A_2 \widehat{A}^\ast - \widehat{A} \mathcal A_2^\ast + \mathcal A_2 \mathcal A_2^\ast)$ is orthogonal to $\mathcal H^S$, this upper bound is simplified to
			\begin{equation}\label{lemeqh0}
				\sum_{j=1}^\varphi | \langle (\widehat{A} \widehat{A}^\ast - \mathcal A_2 \widehat{A}^\ast - \widehat{A} \mathcal A_2^\ast + \mathcal A_2 \mathcal A_2^\ast)u_j,u_j \rangle |.
			\end{equation}
			We note that $|\langle \widehat{A} \widehat{A}^\ast u_j, u_j \rangle - \langle \mathcal A_2 \mathcal A_2^\ast u_j, u_j \rangle |\leq  |\langle \widehat{A}^\ast u_j, (\widehat{A}^\ast-\mathcal A_2^\ast)u_j \rangle| + |\langle (\widehat{A}^\ast-\mathcal A_2^\ast) u_j, \mathcal A_2^\ast u_j \rangle |$, hence
			\begin{equation}
				|\langle \widehat{A} \widehat{A}^\ast u_j, u_j \rangle - \langle \mathcal A_2 \mathcal A_2^\ast u_j, u_j \rangle | \leq \|\widehat{A}^\ast u_j\|\|(\widehat{A}^\ast -\mathcal A_2^\ast )u_j\| + \|\mathcal A_2^\ast u_j\| \|(\widehat{A}^\ast -\mathcal A_2^\ast )u_j\| =\mathrm{o}_p(1), \notag
			\end{equation}
			where the last equality follows from the fact that \eqref{lemeqh00} and $\sup_T \|\widehat{A}^\ast\| = \mathrm{O}_p(1)$ (Lemma \ref{applem1}-(i)). We therefore conclude that 
			\begin{equation}\label{lemeqh1}
				\widehat{A} \widehat{A}^\ast   \to_w   \mathcal A_2 \mathcal A_2^\ast.
			\end{equation}
			From parallel arguments,
			\begin{equation}\label{lemeqh2}
				\mathcal A_2 \widehat{A}^\ast    \to_w   \mathcal A_2 \mathcal A_2^\ast, \quad  \widehat{A} \mathcal A_2^\ast   \to_w   \mathcal A_2 \mathcal A_2^\ast.
			\end{equation}
			Equations \eqref{lemeqh1} and \eqref{lemeqh2} imply that \eqref{lemeqh0} is $\mathrm{o}_p(1)$, which in turn implies that \eqref{lemeqh} is $\mathrm{o}_p(1)$.  \\
			
			\noindent \textbf{Proof of Claim 3}: 
			Note that $\widehat{P}^N_{\varphi} = P^S \widehat{P}^N_{\varphi} + P^N \widehat{P}^N_{\varphi}$ and $\widehat{P}^S_{\varphi} \widehat{P}^N_{\varphi}= 0$ by construction, we thus have $	T \widehat{P}^S_{\varphi} P^S \widehat{P}^N_{\varphi} = -  T \widehat{P}^S_{\varphi} P^N \widehat{P}^N_{\varphi}$. From this equation and the identity $I=\widehat{P}^N_{\varphi} + \widehat{P}^S_{\varphi}$, the following can be shown:		\begin{equation*}
				T  P^S \widehat{P}^N_{\varphi} - T\widehat{P}^N_{\varphi}P^S \widehat{P}^N_{\varphi} =  - T (P^N\widehat{P}^S_{\varphi})^\ast + T(\widehat{P}^S_{\varphi}P^N\widehat{P}^S_{\varphi})^\ast. 
			\end{equation*} 
			Since $P^N \widehat{P}^S_{\varphi} = \mathrm{O}_p (T^{-1})$ and $P^S \widehat{P}^S_{\varphi} - P^S= \mathrm{O}_p(T^{-1})$,  we find that $\|T\widehat{P}^N_{\varphi}P^S \widehat{P}^N_{\varphi}\|_{\mathcal L_{\mathcal H}} = \|T(I-\widehat{P}^S_{\varphi})P^S P^S (I-\widehat{P}^S_{\varphi})\|_{\mathcal L_{\mathcal H}} = \mathrm{O}_p(T^{-1})$ and $\|T\widehat{P}^S_{\varphi}P^N\widehat{P}^S_{\varphi}\|_{\mathcal L_{\mathcal H}} = \|T\widehat{P}^S_{\varphi}P^NP^N\widehat{P}^S_{\varphi}\|_{\mathcal L_{\mathcal H}} = \mathrm{O}_p(T^{-1})$.
			We therefore conclude that $TP^S\widehat{P}^N_{\varphi} = - T(P^N\widehat{P}^S_{\varphi})^\ast = \mathrm{O}_p(T^{-1})$.  \qed
			
			\paragraph*{\normalfont \textbf{Proof of Theorem \ref{prop2}.}} \hspace{0.1em}\\
			We first show that the desired result can be obtained from the following eigenvalue problem:
			\begin{equation} \label{eigenp2}
				\left(\widehat{C}_{\varphi} - \widehat{\Upsilon}_{\varphi}\right) \hat{w}_{j} = \hat{\mu}_{j} \hat{w}_{j}, \quad j=1,2,\ldots.
			\end{equation} 
			It can be shown that $P^N (\widehat{C}_{\varphi} - \widehat{\Upsilon}_{\varphi})P^S$ and $P^S(\widehat{C}_{\varphi}  - \widehat{\Upsilon})P^S$ weakly converge in distribution to some elements in $\mathcal L_{\mathcal H}$,  then Lemma \ref{applem1}-(i) implies that each of $\|P^N (\widehat{C}_{\varphi} - \widehat{\Upsilon}_{\varphi})P^S\|_{\mathcal L_{\mathcal H}}$ and $\|P^S (\widehat{C}_{\varphi} - \widehat{\Upsilon}_{\varphi})P^S\|_{\mathcal L_{\mathcal H}}$ is $\mathrm{O}_p(1)$. We know from Lemma \ref{lem2} that $\widehat{\Upsilon}_{\varphi} = \mathrm{O}_p(1)$, and also deduce from our construction of ${X}_{\varphi,t}$ that  $T^{-1}P^N \widehat{C}_{\varphi} P^N  =  T^{-1}P^N\widehat{C}P^N + \mathrm{O}_p(T^{-1})$. Combining all these results and arguments used in our proof of Theorem \ref{prop1}, we have 
			\begin{equation}
				T^{-1} (\widehat{C}_{\varphi}  - \widehat{\Upsilon}_{\varphi}) =  T^{-1}P^N \widehat{C}P^N + \mathrm{O}_p(T^{-1}) \,\,\to_{\mathcal L_{\mathcal H}} \,\,  \mathcal A_1 =_{\fdd} \int W^N(r)\otimes W^N(r). \label{eq0001}
			\end{equation}
			Then from similar arguments to the proofs of Proposition 3.2 and Theorem 3.3 of \cite{Chang2016152}, we find that  $P^N \widehat{\Pi}^S_{\varphi} = \mathrm{O}_p (T^{-1})$, $P^S \widehat{\Pi}^S_{\varphi} - P^S= \mathrm{O}_p(T^{-1})$ and $\hat{\mu}_{\varphi+k} = \mathrm{O}_p(1)$ for $k \geq 1$.  Since $P^N(\widehat{C}_{\varphi}  - \widehat{\Upsilon}_{\varphi}) \widehat{\Pi}^S_{\varphi} =P^N\widehat{\Pi}^S \widehat{M}$ for $\widehat{M} = \sum_{j=1}^{\infty}  \hat{\mu}_{j} \hat{w}_{j} \otimes \hat{w}_{j}$ and $\widehat{\Pi}^S_{\varphi} = P^N\widehat{\Pi}^S_{\varphi} + P^S\widehat{\Pi}^S_{\varphi}$,  
			\begin{equation*}
				P^N (\widehat{C}_{\varphi}  - \widehat{\Upsilon}_{\varphi}) P^N\widehat{\Pi}^S_{\varphi} + P^N (\widehat{C}_{\varphi}  - \widehat{\Upsilon}_{\varphi}) P^S\widehat{\Pi}^S_{\varphi} =    P^N\widehat{\Pi}^S_{\varphi} \widehat{M}. 
			\end{equation*}	
			We know from $\widehat{\Pi}^S_{\varphi}P^N=\mathrm{O}_p(T^{-1})$ and $\widehat{P}^S_\varphi P^N=\mathrm{O}_p(T^{-1})$ that   $P^N \widehat{\Pi}^S_{\varphi} \widehat{M}= \mathrm{O}_p(T^{-1})$ and $P^N \widehat{\Upsilon}_{\varphi} P^N = \mathrm{O}_p(T^{-1})$. This in turn implies that 
			
			\begin{equation} \label{eqprop01}
				T P^N\widehat{\Pi}^S_{\varphi} = -\left(T^{-1}P^N \widehat{C}_{\varphi} P^N\right)^\dag \left(P^N \widehat{C}_{\varphi} P^S -  P^N\widehat{\Upsilon}_{\varphi} P^S\right) + \mathrm{O}_p(T^{-1}),
			\end{equation}		
			where note that $(T^{-1}P^N  \widehat{C}_{\varphi} P^N)^\dag$ is well define since $T^{-1}P^N  \widehat{C}_{\varphi} P^N$ is a finite rank operator, so has a closed range.
			Using similar results to \eqref{eq001} and Claim 3 in our proof of Theorem \ref{prop1}, we find that
			\begin{align}
				T(\widehat{\Pi}^N_{\varphi} - P^N) = &\left(P^S \widehat{C}_{\varphi} P^N -  P^S \widehat{\Upsilon}_{\varphi}^\ast P^N\right)  \left(T^{-1}P^N \widehat{C}_{\varphi} P^N\right)^\dag  \notag \\ &+  \left(T^{-1}P^N \widehat{C}_{\varphi} P^N\right)^\dag \left(P^N \widehat{C}_{\varphi} P^S -  P^N \widehat{\Upsilon}_{\varphi} P^S\right) + \mathrm{O}_{p}(T^{-1}). \label{eqqpp1}
			\end{align}
			From \eqref{eq0001} and nearly identical arguments to those used to derive \eqref{show1}, 
			\begin{equation} \label{prop2com1}
				\left(T^{-1}P^N \widehat{C}_{\varphi} P^N\right)^\dag  \to_{\mathcal L_{\mathcal H}}   \mathcal A_1^\dag.  
			\end{equation}	We will derive the limit of $\widehat{A} = P^N\widehat{C}_{\varphi}P^S - P^N \widehat{\Upsilon}_{\varphi} P^S$. First, it may be deduced from Lemma \ref{lem2} that  
			\begin{equation} \label{eqeq01}
				\widehat{\Gamma}^{NN}_{\varphi} \to_{\mathcal L_{\mathcal H}} \Gamma^{NN}, \quad \widehat{\Gamma}^{NS}_{\varphi}  \to_{\mathcal L_{\mathcal H}} \Gamma^{NS}, \quad \widehat{\Gamma}^{SS}_{\varphi}   \to_{\mathcal L_{\mathcal H}} \Gamma^{SS}.
			\end{equation}
			Moreover, from the fact that $\rank \widehat{\Omega}^{NN} = \rank {\Omega}^{NN}= \varphi$ almost surely, $B \mapsto B^{-1}$ is a continuous map for every positive definite matrix $B\in \mathbb{R}^{\varphi\times\varphi}$, and  $\widehat{\Omega}^{NN}_{\varphi}\to_{\mathcal L_{\mathcal H}} {\Omega}^{NN}_{\varphi}$ (Lemma \ref{lem2}) and this may be understood as a convergence in probability in  $\mathbb{R}^{\varphi\times \varphi}$, we find that
			\begin{equation}\label{eqeq02}
				\widehat{\Omega}^{NN}_{\varphi}{\mid_{\varphi}^\dagger}  \to_{\mathcal L_{\mathcal H}}  ({\Omega}^{NN})^\dag.
			\end{equation}
			Combining \eqref{eqeq01} and \eqref{eqeq02}, we have, for any $k \geq 1$, $x_1,\ldots,x_k \in \mathcal H$ and $y_1,\ldots,y_k \in \mathcal H$, $\langle P^N \widehat{\Upsilon}_{\varphi} P^Sx_j,y_j \rangle  \to_{p}   \langle \Upsilon x_j,y_j\rangle$,
			where $\Upsilon = \Gamma^{NS}- \Gamma^{NN} \left({\Omega}^{NN} \right)^\dag \Omega^{NS}$.
			We then note that
			\begin{equation*}
				\langle \widehat{A} x_j, y_j \rangle =  \frac{1}{T} \sum_{t=1}^T \langle P^S X_t - {\Omega}^{SN} \left({\Omega}^{NN}\right)^\dag P^N \Delta X_t , x_j \rangle  \langle P^N X_t,y_j \rangle -  \langle \Upsilon x_j,y_j\rangle + \mathrm{o}_p(1),
			\end{equation*}
			where the first term converges in distribution to $\int \langle d {W}^{S.N}(r), x_j \rangle \langle W^N(r),y_j \rangle + \langle \Upsilon x_j,y_j\rangle$; see e.g.\ the proofs of Theorems 1 and 2 of \cite{harris1997principal}. Using this result and nearly identical arguments used to derive \eqref{eqref1}, it is quite obvious to establish that $\widehat{A} \to_{wd} \int dW^{S.N}(r) \otimes W^N(r)$. Furthermore, from similar  arguments to those used to prove Claim 2 in our proof of Theorem \ref{prop1},  we find that
			\begin{equation}\label{prop2com2}
				\widehat{A}  \to_{\mathcal L_{\mathcal H}}    {\mathcal A}_2 =_{\fdd}  \int dW^{S.N}(r) \otimes W^N(r).
			\end{equation}
			Combining \eqref{eqqpp1}, \eqref{prop2com1}, \eqref{prop2com2}, and Lemma \ref{applem1}-(ii) and (iii), we obtain the desired result. 
			
			We now consider the eigenvalue problem \eqref{eigenp2a}. On top of the previous proof, it is easy to show that $\widehat{\Upsilon}_{\varphi}^\ast= \mathrm{O}_p(1)$ and $P^N \widehat{\Upsilon}_{\varphi}^\ast P^S  = \mathrm{O}_{p}(T^{-2})$, so \eqref{eqprop01} still holds for $\widehat{\Pi}^S_{\varphi}$ obtained from \eqref{eigenp2a}. The rest of proof is almost identical to our proof for the eigenvalue problem  \eqref{eigenp2}.
			
			Independence of ${W}^{S.N}$ and $W^N$ is deduced from the fact that $\mathbb{E}[{W}^{S.N} \otimes W^N] = 0$.
			
			\paragraph*{\normalfont \textbf{Proof of Proposition \ref{prop2a}}.} 
			Since $\dim(\mathcal H) < \infty$ and the minimum eigenvalue of $\mathbb{E}[\mathcal E_t \otimes \mathcal E_t]$ is strictly positive, we may deduce from the proof of Theorem 2 in \cite{harris1997principal} and arguments similar to those used in Theorems \ref{prop1} and \ref{prop2} that 
			\begin{align}
				&T^{-1} \ddot{C}_\varphi = T^{-1}P^N \ddot{C}_\varphi P^N + \mathrm{O}_p(T^{-1})  \,\,\to_{\mathcal L_{\mathcal H}}\,\, \mathcal A_1 =_{\fdd} \int W^N(r) \otimes W^N(r), \label{hreq01} \\
				&  P^N\ddot{C}_\varphi P^S\,\,\to_{\mathcal L_{\mathcal H}}\,\, \mathcal A_2 =_{\fdd} \int dW^{S.N}(r) \otimes W^N(r).\label{hreq02}
			\end{align}
			As in our proof of Theorem \ref{prop2}, we have $P^N  \ddot{\Pi}^S_\varphi = \mathrm{O}_p(T^{-1})$, $P^S  \ddot{\Pi}^S_\varphi - P^S = \mathrm{O}_p(T^{-1})$ and $\ddot{\lambda}_{\varphi+k} = \mathrm{O}_p(1)$ for $k \geq 1$. From similar arguments to those used to derive \eqref{eqqpp1}, we have 
			\begin{equation}
				T(\ddot{\Pi}^N_\varphi - P^N) = \left(T^{-1}P^N  \ddot{C}_\varphi P^N\right)^\dag  P^N  \ddot{C}_\varphi P^S + P^S  \ddot{C}_\varphi P^N \left(T^{-1}P^N  \ddot{C}_\varphi P^N\right)^\dag  + \mathrm{O}_{p}(T^{-1}).  \label{hreq03}
			\end{equation}
			Then  the desired result is deduced from \eqref{hreq01}-\eqref{hreq03}. \qed

		\paragraph*{\normalfont \textbf{A detailed discussion on Remark \ref{rem2}}.} \label{addremark}
		We here consider the simple case when $k=1$. From Theorem 3.1 of \cite{saikkonen1991asymptotically}, it can be shown that 
		\begin{equation}
			\lim_{T\to \infty} \text{Prob.\ }  \{|T \langle P^N \widehat{\Pi}^S_{\varphi}x,y \rangle|<\delta\} \geq \lim_{T\to \infty} \text{Prob.\ }  \{|T \langle P^N \widehat{P}^S_{\varphi}x,y \rangle|<\delta\}. \label{eq001lem1}
		\end{equation}
		If $\Omega^{NS}= \Gamma^{NS}=0$, then $T \langle P^N \widehat{P}^S_\varphi\rangle$ and $T \langle P^N \widehat{\Pi}^S_{\varphi}\rangle$ have the same limiting distribution. Moreover, if $y \in \mathcal H^S$ 
		both sides of \eqref{eq001lem1} are equal to zero regardless of the limiting behaviors of $P^N \widehat{P}^S_\varphi$ and $ P^N \widehat{\Pi}^S_{\varphi}$. Therefore the inequality given in \eqref{eq001lem1} is not always strict. However, if $\Omega^{NS}= \Gamma^{NS} = 0$ is not true and our choice of $x$, $y$ and $\delta$ makes the left hand side of  \eqref{eq001lem1} positive, then the inequality is strict; see Theorem 3.1 of \cite{saikkonen1991asymptotically}.  Moreover, the left hand side can be made positive by  choosing $x,y$ and $\delta$ appropriately. For example, suppose that  $x\in \mathcal H^S \setminus\{0\}$ and  $y\in \mathcal H^N\setminus\{0\}$. As shown in our proof of Theorem \ref{prop1}, we may assume that $\int W^N(r) \otimes W^N (r)= \sum_{j=1}^\varphi \gamma_j u_j \otimes u_j$, $\{\gamma_j\}_{j=1}^\varphi$ are all positive, and $\spn(\{u_j\}_{j=1}^\varphi) = \mathcal H^N$. Then from Theorem \ref{prop2}, we have $T \langle P^N \widehat{\Pi}^S_{\varphi}x,y \rangle   \to_d   	\sum_{j=1}^\varphi \gamma_j^{-1} \int \langle dW^{S.N}(r),x \rangle \langle u_j, W^N(r) \rangle\langle u_j, y \rangle$.
		That is, $T \langle P^N \widehat{\Pi}^S_{\varphi}x,y \rangle $ converges to a nondegenerate distiribution centered at $0 \in \mathbb{R}$. Therefore, the left hand side of \eqref{eq001lem1} becomes positive for any strictly positive $\delta$. 	A similar result can be obtained for $T^{-1}  P^S \widehat{\Pi}^N_{\varphi}$ with only a minor modification of the above arguments.

		
		\subsection{Mathematical proofs for the results in Sections \ref{sec:test} and \ref{sec_general_test}}  \label{appproof2}
		\subsubsection{Preliminary results}
		\begin{lemma}  \label{lem3} Suppose that Assumptions \ref{assumm0}, \ref{assumW} and \ref{assum2} hold.
			\begin{itemize}
				\item[$\mathrm{(i)}$]  For any $\varphi_0 \leq \varphi$ and $K > \varphi$, $\hat{w}_j$ given in \eqref{eigenptest} satisfies
				\begin{align} 
					\|\hat{w}_j - \sgn(\langle\hat{w}_j, w_j\rangle ) w_j\|  \to_p 0, \quad j=1,\ldots,K,  \label{eqlem301}
				\end{align} 
				where $\{w_j\}_{j=1}^\varphi$ (resp.\ $\{w_j\}_{j=\varphi+1}^K$) is an orthonormal set of $\mathcal H^N$  (resp.\ $\mathcal H^S$).
				\item[$\mathrm{(ii)}$] 	 Then \eqref{eqlem301} still holds if we replace $\hat{w}_j$ with  $\overline{w}_j$ given in \eqref{eigenptest2} when Model D1 given by \eqref{eqdeter1} or Model D2 given by \eqref{eqdeter2} is true.
			\end{itemize}
		\end{lemma}
		\begin{proof}
			\noindent We first show (i). The proof is trivial when $\varphi=0$, so we hereafter assume $\varphi \geq 1$. Note that 
			\begin{equation}
				{{X}_{\varphi_0,t}}/{\sqrt{T}} = 	{X_t}/{\sqrt{T}} - 	\widehat{\Omega}^{SN}_{\varphi_0}\left(\widehat{\Omega}^{NN}_{\varphi_0}{\mid_{\varphi_0}^\dagger}\right) 	\widehat{P}^{N}_{\varphi_0}\Delta X_t/\sqrt{T}. \label{eqlem01}
			\end{equation}  
			It can be shown that  $\widehat{\Omega}^{SN}_{\varphi_0}(\widehat{\Omega}^{NN}_{\varphi_0}{\mid_{\varphi_0}^\dagger}) 	\widehat{P}^{N}_{\varphi_0}$ converges to a well defined limit if $\varphi_0 \leq \varphi$. Thus the second term on the right hand side of \eqref{eqlem01} is $O_p(T^{-1/2})$ and this implies that $T^{-1} \widehat{C}_{\varphi_0} = T^{-1} \widehat{C} + \mathrm{O}_p(T^{-1})$. Moreover, we deduce from Lemma \ref{lem2} that  $\widehat{\Upsilon}_{\varphi_0} = \mathrm{O}_p(1)$. From similar arguments to our proof of Theorem \ref{prop2}, we have
			\begin{equation} \label{pfeqclaim}
				T^{-1} (\widehat{C}_{\varphi_0} -\widehat{\Upsilon}_{\varphi_0} -\widehat{\Upsilon}_{\varphi_0}^\ast)   \to_{\mathcal L_{\mathcal H}}    \mathcal A_1  =_{\fdd} \int W^N (r)\otimes W^N(r).
			\end{equation}
			We note that the eigenvalues of the limiting operator are distinct and each eigenvalue is associated with  one-dimensional eigenspace, almost surely.  Then Lemma 4.3 of \cite{Bosq2000} and \eqref{pfeqclaim} jointly imply that the eigenvectors of $T^{-1} (\widehat{C}_{\varphi} -\widehat{\Upsilon}_{\varphi_0} -\widehat{\Upsilon}_{\varphi_0}^\ast)$ converge to those of $\int W^N (r)\otimes W^N(r)$. Note also that $\int W^N (r)\otimes W^N(r)$ is a  positive definite operator  of rank $\varphi = \dim(\mathcal H^N)$ almost surely. This implies that the eigenvectors of $T^{-1} (\widehat{C}_{\varphi_0} -\widehat{\Upsilon}_{\varphi_0} -\widehat{\Upsilon}_{\varphi_0}^\ast)$ corresponding to the largest $\varphi$ eigenvalues converge to an orthonormal basis of $\mathcal H^N$. 	We now let $\widehat{Q}^N_{\varphi}=\sum_{j=1}^\varphi \hat{w}_j \otimes\hat{w}_j$ and let $\widehat{Q}^S_{\varphi} = I-\widehat{Q}^N_{\varphi}$. Then from \eqref{pfeqclaim} and similar arguments used in the proof of Proposition 3.2 in \cite{Chang2016152}, it can be shown that $\|\widehat{Q}^N_{\varphi}-P^N\|_{\mathcal L_{\mathcal H}} = \mathrm{O}_p(T^{-1})$ and $\|\widehat{Q}^S_{\varphi}-P^S\|_{\mathcal L_{\mathcal H}} = \mathrm{O}_p(T^{-1})$. Then from a nearly identical argument used in the proof of Theorem 3.3 in \cite{Chang2016152}, it can be shown that 
			\begin{equation*}
				\|\hat{w}_j- \sgn(\langle\hat{w}_j, w_j\rangle ) w_j\|   \to_p    0, \quad j = \varphi+1,\ldots,
			\end{equation*} 
			where $w_{\varphi+1},\ldots$ are the eigenvectors of $\mathbb{E}[\mathcal E^S_t \otimes \mathcal E^S_t]$. This completes our proof of (i). 
			
			We now show (ii). From our proof of Theorem \ref{prop3}, we may similarly deduce that $T^{-1} (\overline{C}_{\varphi_0}-\overline{\Upsilon}_{\varphi_0} -\overline{\Upsilon}_{\varphi_0}^\ast)   \to_{\mathcal L_{\mathcal H}} \overline{ \mathcal A}_{1}  =_{\fdd} \int  \overline{W}^N (r)\otimes  \overline{W}^N(r)$. 
			Then the rest of proof is similar to that for the case without deterministic terms.
		\end{proof}
		\subsubsection{Proofs of the main results}
		\paragraph*{\normalfont \textbf{Proof of Proposition \ref{propkpss} and Theorem \ref{prop3a}}.} Since Proposition \ref{propkpss} is a special case of Theorem \ref{prop3a}, we here only prove the latter. For notational convenience, we let $Y_{\varphi_0,t} =  \sum_{s=1}^t {X}_{\varphi_0,t}$. Then $\sum_{s=1}^t z_{\varphi_0,s} = (\langle Y_{\varphi_0,t},\hat{w}_{\varphi_0+1}\rangle, \ldots, \langle Y_{\varphi_0,t}, \hat{w}_{\varphi}\rangle)'$.  
		We first consider the case when $\varphi_0 \geq 1$. Note that 
		\begin{equation}\label{proofthmeq01}
			T^{-1/2}\langle Y_{\varphi_0,t}, \hat{w}_{j} \rangle = 	\langle T^{-1/2}{Y_{\varphi_0,t}}, P^S \widehat{\Pi}^S_{\varphi_0}  \hat{w}_{j} \rangle + 	\left\langle T^{-3/2}{Y_{\varphi_0,t}}, T P^N \widehat{\Pi}^S_{\varphi_0} \hat{w}_{\varphi_0+1} \right\rangle 
		\end{equation}
		for $j=\varphi_0+1,\ldots,K$.
		From Theorem \ref{prop1} and Lemma \ref{lem3}-(i), 
		\begin{align}
			& P^S \widehat{\Pi}^S_{\varphi_0}  \to_{\mathcal L_{\mathcal H}}   P^S, \label{proofthmeq00} \\
			&\|\hat{w}_{j}- \sgn(\langle \hat{w}_{j},w_j \rangle) w_{j}\| = \mathrm{o}_p(1), \quad j=\varphi_0+1,\ldots,K, \label{proofthmeq00a}
		\end{align}
		where $\{w_{j}\}_{j=\varphi_0+1}^K$ is an orthonormal set included in $\mathcal H^S$.	Moreover, we may deduce the following from our proof of Theorem \ref{prop2}: for any $v \in \mathcal H^S$ and $w \in \mathcal H^N$,
		\begin{equation} 
			\left(\left\langle T^{-1/2}{Y_{\varphi_0,t}}, v \right\rangle, \left\langle T^{-3/2}{Y_{\varphi_0,t}}, w \right\rangle\right) \to_d \left(\left\langle W^{S.N}(r),v \right\rangle, \left\langle \int_{0}^r W^N(s),w \right\rangle  \right).  \label{proofthmeq00aa} 
		\end{equation}
		From \eqref{proofthmeq00} and \eqref{proofthmeq00a}, we conclude that \text{for any $v \in \mathcal H$,}
		\begin{equation}
			\langle v,  P^S \widehat{\Pi}^S_{\varphi_0} \hat{w}_{j}\rangle  \to_p    -  \sgn(\langle \hat{w}_{j},w_j \rangle)  \left\langle v,w_{j}\right\rangle. \label{proofthmeq00aa2} 
		\end{equation}
		From Theorem \ref{prop2} and \eqref{proofthmeq00a}, we may deduce that   for any $v \in \mathcal H^N$ and $j=\varphi_0+1,\ldots,K$,
		\begin{equation}
			\langle v, T P^N \widehat{\Pi}^S_{\varphi_0} \hat{w}_{j}\rangle \to_d   -  \sgn(\langle \hat{w}_{j},w_j \rangle)  \left\langle v,A w_{j}\right\rangle =_d  - \sgn(\langle\hat{w}_{j},w_j\rangle) \left\langle A^\ast v, w_{j}\right\rangle , \quad  \label{proofthmeq00aa3} 
		\end{equation}
		where $A^\ast =_{\fdd} \left(\int W^N(r) \otimes  d {W}^{S.N}(r)  \right)\left(\int W^N(r) \otimes W^N(r)\right)^\dag$. 	Combining \eqref{proofthmeq01}, \eqref{proofthmeq00aa}, \eqref{proofthmeq00aa2}, \eqref{proofthmeq00aa3}, and the Cram\'{e}r-Wold device, we obtain the following convergence result:
		\begin{equation}
			T^{-1/2}\langle Y_{\varphi_0,t}, \hat{w}_{j}\rangle  \to_d	\left\langle {W}^{S.N}(r) - A^\ast \int_{0}^r W^N(s), {w}_{j} \right\rangle, \quad  \text{jointly for $j=\varphi_0+1,\ldots,K$},\label{proofthmeq06}
		\end{equation} 
		where we used the property that the limiting distribution does not depend on $\sgn(\langle\hat{w}_{j},w_j\rangle)$.
		We let 
		\begin{equation*}
			{B}^N(r) = (\langle {W}^N(r),w_{1} \rangle,  \ldots,\langle {W}^N(r), w_{\varphi_0} \rangle)', \quad {B}^{S.N}(r) = (\langle {W}^{S.N}(r), w_{\varphi_0+1} \rangle, \ldots, {W}^{S.N}(r), w_K \rangle)'.
		\end{equation*}
		Due to isomorphism between $n$-dimensional subspace of $\mathcal H$ and $n$-dimensional Euclidean space for any finite integer $n$, the joint limiting distribution in \eqref{proofthmeq06} may be understood as
		\begin{equation}
			V^S(r) = {B}^{S.N}(r) - \int d	{B}^{S.N}(s) \mathsf{B}(s)' \left(\int \mathsf{B}(s)\mathsf{B}(s)'\right)^{-1} \int_{0}^r\mathsf{B}(s), \label{eqfinal01}
		\end{equation}	
		where we used the distributional identity given by $$B^N(s)'\left(\int B^N(s) {B^N}(s)'\right)^{-1} \int_{0}^r B^N(s) =_d \mathsf{B}(s)' \left(\int \mathsf{B}(s)\mathsf{B}(s)'\right)^{-1} \int_{0}^r\mathsf{B}(s).$$
		We next show that $\LRV(z_{\varphi_0,t})$ converges to the covariance matrix of $B^{S.N}$. For  time series $\{V_t\}_{t=1}^T$ and $\{W_t\}_{t=1}^T$, we   let $\mathcal G(V_t,W_t)$ denote the operator given by 
		\begin{equation}
			\mathcal G(V_t,W_t)=\frac{1}{T}\sum_{t=1}^T V_t \otimes W_t + 	\frac{1}{T}\sum_{s=1}^{T-1} \mathrm{k}\left(\frac{s}{h}\right) \sum_{t=s+1}^T \left\{V_t \otimes  W_{t-s} + W_{t-s} \otimes  V_t\right\}. 
		\end{equation}
		We will show that $\|\mathcal G(\widehat{\Pi}^S_{\varphi_0} X_{\varphi_0,t},\widehat{\Pi}^S_{\varphi_0} X_{\varphi_0,t}) - \Omega^{S.N}\|_{\mathcal L_{\mathcal H}} = \mathrm{o}_p(1)$, which is a stronger result implying that $\LRV(z_{\varphi_0,t})$ converges to the covariance matrix of $B^{S.N}$. Note that $\mathcal G(\widehat{\Pi}^S_{\varphi_0} X_{\varphi_0,t},\widehat{\Pi}^S_{\varphi_0} X_{\varphi_0,t})$ is equal to 
		\begin{align}
			&\mathcal G(\widehat{\Pi}^S_{\varphi_0} X_t,\widehat{\Pi}^S_{\varphi_0} X_t) -  \mathcal G(\widehat{\Pi}^S_{\varphi_0} X_t, \widehat{\Pi}^S_{\varphi_0} \widehat{\Omega}^{SN}(\widehat{\Omega}^{NN}_{\varphi_0}{\mid_{\varphi_0}^\dagger})\widehat{P}^N_{\varphi_0} \Delta X_t)  -  \mathcal G(\widehat{\Pi}^S_{\varphi_0} \widehat{\Omega}^{SN}(\widehat{\Omega}^{NN}_{\varphi_0}{\mid_{\varphi_0}^\dagger})\widehat{P}^N_{\varphi_0} \Delta X_t, \widehat{\Pi}^S_{\varphi_0} X_t)\notag \\ & +  \mathcal G(\widehat{\Pi}^S_{\varphi_0} \widehat{\Omega}^{SN}(\widehat{\Omega}^{NN}_{\varphi_0}{\mid_{\varphi_0}^\dagger})\widehat{P}^N_{\varphi_0} \Delta X_t, \widehat{\Pi}^S_{\varphi_0} \widehat{\Omega}^{SN}(\widehat{\Omega}^{NN}_{\varphi_0}{\mid_{\varphi_0}^\dagger})\widehat{P}^N_{\varphi_0} \Delta X_t).\notag
		\end{align}
		From Lemma \ref{lem2}, Theorems \ref{prop1} and \ref{prop3}, and \eqref{eqeq02}, we find that $\mathcal G(\widehat{\Pi}^S_{\varphi_0} X_t,\widehat{\Pi}^S_{\varphi_0} X_t)\to_{\mathcal L_{\mathcal H}}\Omega^{SS}$ and each of the other terms converges to $\Omega^{SN}({\Omega}^{NN})^\dagger \Omega^{NS}$ in the same sense. This implies that $	\mathcal G(\widehat{\Pi}^S_{\varphi_0} X_{\varphi_0,t},\widehat{\Pi}^S_{\varphi_0} X_{\varphi_0,t})  \to_{\mathcal L_{\mathcal H}}   \Omega^{S.N}$.
		From this result and \eqref{eqfinal01}, we find that under $H_0$,
		\begin{equation}
			\LRV(z_{\varphi_0,t})^{-1/2} V^S(r)  \to_d    V(r).  \label{eqfinal02}
		\end{equation}
		From \eqref{proofthmeq06},  \eqref{eqfinal01}, \eqref{eqfinal02} and the continuous mapping theorem, we may conclude that $\widehat{\mathcal Q}(K,\varphi_0)  \to_d  \int V(r)'V(r)$, which establishes the desired result under $H_0$.

		Under $H_1$, Lemma \ref{lem3}-(i) implies that $\widehat{\Pi}^N_{\varphi_0}$ converges to a projection onto a strict subspace of $\mathcal H^N$. Therefore for some $j$, we have $\widehat{\Pi}^N_{\varphi_0}\hat{w}_{j}  \to_p  w \in \mathcal H^N $.
		We then deduce from \eqref{proofthmeq00aa} that for such $j$ $T^{-3/2} \langle {Y_{\varphi_0,t}}, P^S \widehat{\Pi}^S_{\varphi_0}\hat{w}_{j} \rangle$ converges to a functional of $W^N$, so $T^{-3/2}(\langle Y_{\varphi_0,t}, \hat{w}_{\varphi_0+1} \rangle, \ldots, \langle Y_{\varphi_0,t},\hat{w}_{K} \rangle)'$ converges in distribution to a nondegenerate limit. It is also  deduced from the unnumbered equation between (A.10) and (A.11) of \cite{phillips1988spectral} that $(\kappa hT)^{-1}\LRV(z_{\varphi_0,t})$ converges in probability to a nonzero limit.  Combining all these results,  $\widehat{\mathcal Q}(K,\varphi_0) \to_p \infty$ is deduced. 
		
		It remains to prove the case when $\varphi_0 = 0$. In this case, the second term in  \eqref{proofthmeq01} is equal to zero, $A^\ast=0$ in \eqref{proofthmeq06}, $W^N = 0$ and $W^{S.N}$ is understood as $W^S$. Then the rest of the proof is similar to  that for the case when  $\varphi_0 \geq 1$.	\qed

		\subsection{Mathematical proofs of the results given in Section \ref{sec:deterministic}} 
		
		\paragraph*{\normalfont \textbf{Proof of Theorem \ref{prop3}}.} From Lemma 3 of \cite{SS2019} and our proof for the case without deterministic terms, we may deduce that $	T^{-1}\overline{C} =T^{-1}P^N\overline{C} P^N + \mathrm{O}_p(T^{-1})$ 
		and thus
		\begin{align*}
			T^{-1}\overline{C} \to_{\mathcal L_{\mathcal H}} \overline{\mathcal A}_{1}=_{\fdd} \int \overline{W}^N(r) \otimes \overline{W}^N(r), 
		\end{align*}
		moreover, $P^N \overline{P}^{S}_\varphi = \mathrm{O}_p (T^{-1})$, $P^S \overline{P}^{S}_\varphi - P^S= \mathrm{O}_p(T^{-1})$, 
		and   $\overline{\lambda}_{\varphi+k} = \mathrm{O}_p(1)$ 
		for $k \geq 1$. With a standard modification to allow deterministic terms from our proof of Theorem \ref{prop1}, we deduce that  $P^N \overline{C} P^S   \to_{\mathcal L_{\mathcal H}}  \overline{\mathcal A}_{2}=_{\fdd} \int d W^S(r) \otimes \overline{W}^N(r)  +  \Gamma^{NS}$ 
		From these results and our proof of Theorem \ref{prop1}, the limit of $T(\overline{P}^N_{\varphi} - P^N)$ 
		is obtained as desired.
		
		Moreover, as we did in our proof of Theorem \ref{prop2}, it can be shown that $T^{-1}(\overline{C}_\varphi-\overline{\Upsilon}_\varphi -\overline{\Upsilon}_\varphi^\ast)  = T^{-1}P^N \overline{C} P^N + \mathrm{O}_p(T^{-1}) \to_{\mathcal L_{\mathcal H}}  \overline{\mathcal A}_{1}=_{\fdd} \int \overline{W}^N(r) \otimes \overline{W}^N(r)$, 
		(ii) $P^N  \overline{\Pi}^S_{\varphi} = \mathrm{O}_p (T^{-1})$, $P^S \overline{\Pi}^S_{\varphi} - P^S= \mathrm{O}_p(T^{-1})$, 
		and $\overline{\mu}_{\varphi+k} = \mathrm{O}_p(1)$ 
		for $k \geq 1$.
		The rest of proof is similar to that of Theorem \ref{prop2} concerning with the case without deterministic terms; that is, it can be shown that (i) $
		T(\overline{\Pi}^N_{\varphi} - P^N) = (P^S \overline{C}_{\varphi} P^N -  P^S \overline{\Upsilon}_{\varphi}^\ast P^N)  (T^{-1}P^N \overline{C}_{\varphi} P^N)^\dag  +  (T^{-1}P^N \overline{C}_{\varphi} P^N)^\dag (P^N \overline{C}_{\varphi} P^S -  P^N \overline{\Upsilon}_{\varphi} P^S) + \mathrm{O}_{p}(T^{-1})$ and (ii) $(T^{-1}P^N \overline{C}_{\varphi} P^N)^\dag \to_{\mathcal L_{\mathcal H}} \overline{\mathcal A}_{1}^\dag$ and $P^N \overline{C}_{\varphi} P^S -  P^N \overline{\Upsilon}_{\varphi} P^S \to_{\mathcal L_{\mathcal H}} \overline{\mathcal A}_{2} =_{\fdd} (\int d W^{S.N}(r) \otimes \overline{W}^N(r))$.
		From these, the desired result for $\overline{\Pi}^N_{\varphi}$ is obtained. 
		\qed
		
		\paragraph*{\normalfont \textbf{Proof of Theorem \ref{prop3aa}}.}
		From Theorem \ref{prop3}, Lemma \ref{lem3}-(ii) and a slight modification of the arguments used in our proof of Theorem \ref{prop3a}, the desired results may be easily deduced.  \qed
		
		%

		\bibliographystyle{apalike}
			\bibliographystyle{apalike}

\begin{thebibliography}{}
			
			\bibitem[\protect\citeauthoryear{Aue, Norinho, and H{\"o}rmann}{Aue
				et~al.}{2015}]{aue2015prediction}
			Aue A, Norinho DD, H{\"o}rmann S. 2015.
			\newblock On the prediction of stationary functional time series.
			\newblock {\em Journal of the American Statistical Association\/}~{\em
				110\/}(509), 378--392.
			
			\bibitem[\protect\citeauthoryear{Aue and van~Delft}{Aue and
				van~Delft}{2020}]{aue2017testing}
			Aue A, van~Delft A. 2020.
			\newblock Testing for stationarity of functional time series in the frequency
			domain.
			\newblock {\em Annals of Statistics\/}~{\em 48\/}(5), 2505--2547.
			
			\bibitem[\protect\citeauthoryear{Beare, Seo, and Seo}{Beare
				et~al.}{2017}]{BSS2017}
			Beare BK, Seo J, Seo WK. 2017.
			\newblock Cointegrated linear processes in {H}ilbert space.
			\newblock {\em Journal of Time Series Analysis\/}~{\em 38\/}(6), 1010--1027.
			
			\bibitem[\protect\citeauthoryear{Beare and Seo}{Beare and Seo}{2020}]{BS2018}
			Beare BK, Seo WK. 2020.
			\newblock Representation of {I}(1) and {I}(2) autoregressive {H}ilbertian
			processes.
			\newblock {\em Econometric Theory\/}~{\em 36\/}(5), 773--802.
			
			\bibitem[Berkes et~al., 2013]{berkes2013weak}
		Berkes I, Horv{\'a}th L, Rice G. 2013.
		\newblock Weak invariance principles for sums of dependent random functions.
		\newblock {\em Stochastic Processes and their Applications}, 123(2), 385--403.
		
		
		\bibitem[\protect\citeauthoryear{Bosq}{Bosq}{2000}]{Bosq2000}
		Bosq D. 2000.
		\newblock {\em Linear Processes in Function Spaces}.
		\newblock Springer, New York.
		
	
			
			\bibitem[\protect\citeauthoryear{Chang, Kaufmann, Kim, Miller, Park, and
				Park}{Chang et~al.}{2020}]{chang2020evaluating}
			Chang Y, Kaufmann RK, , Kim CS, Miller JI , Park JY, and Park S. 2020.
			\newblock Evaluating trends in time series of distributions: A spatial
			fingerprint of human effects on climate.
			\newblock {\em Journal of Econometrics\/}~{\em 214\/}(1), 274--294.
			
			\bibitem[\protect\citeauthoryear{Chang, Kim, and Park}{Chang
				et~al.}{2016}]{Chang2016152}
			Chang Y, Kim CS, and Park JY. 2016.
			\newblock Nonstationarity in time series of state densities.
			\newblock {\em Journal of Econometrics\/}~{\em 192\/}(1), 152 -- 167.
			
			\bibitem[\protect\citeauthoryear{Chen and White}{Chen and
				White}{1998}]{chen1998}
			Chen X, White H. 1998.
			\newblock Central limit and functional central limit theorems for
			Hilbert-valued dependent heterogeneous arrays with applications.
			\newblock {\em Econometric Theory\/}~{\em 14\/}(2), 260--284.
					
			\bibitem[\protect\citeauthoryear{Choi and Ahn}{Choi and
				Ahn}{1995}]{choi1995testing}
			Choi I, Ahn BC. 1995.
			\newblock Testing for cointegration in a system of equations.
			\newblock {\em Econometric Theory\/}~{\em 11\/}(5), 952--983.
			
			\bibitem[\protect\citeauthoryear{Delicado}{Delicado}{2011}]{delicado2011dimensionality}
			Delicado P. 2011.
			\newblock Dimensionality reduction when data are density functions.
			\newblock {\em Computational Statistics \& Data Analysis\/}~{\em 55\/}(1),
			401--420.
			
			\bibitem[\protect\citeauthoryear{Egozcue, D{\'i}az-Barrero, and
				Pawlowsky-Glahn}{Egozcue et~al.}{2006}]{Egozcue2006}
			Egozcue JJ, D{\'i}az-Barrero JL,  Pawlowsky-Glahn V. 2006.
			\newblock {H}ilbert space of probability density functions based on {A}itchison
			geometry.
			\newblock {\em Acta Mathematica Sinica\/}~{\em 22\/}(4), 1175--1182.
			
			\bibitem[\protect\citeauthoryear{Engl and Nashed}{Engl and
				Nashed}{1981}]{Engl1981}
			Engl HW,  Nashed M. 1981.
			\newblock Generalized inverses of random linear operators in {B}anach spaces.
			\newblock {\em Journal of Mathematical Analysis and Applications\/}~{\em
				83\/}(2), 582 -- 610.
			
			\bibitem[\protect\citeauthoryear{Flood, King, Rodgers, Ruggles, and
				Warren}{Flood et~al.}{2018}]{IPUMS2020}
			Flood S, King M, Rodgers R, Ruggles S, and Warren JR. 2018.
			\newblock {\em Integrated Public Use Microdata Series, Current Population
				Survey: Version 8.0 [dataset]}.
			\newblock Minneapolis, MN: IPUMS.
			
			\bibitem[\protect\citeauthoryear{Franchi and Paruolo}{Franchi and
				Paruolo}{2020}]{Franchi2017b}
			Franchi M, Paruolo P. 2020.
			\newblock Cointegration in funcational autoregressive processes.
			\newblock {\em Econometric Theory\/}~{\em 36\/}(5), 803--839.
			
			\bibitem[\protect\citeauthoryear{Hansen}{Hansen}{1992}]{hansen1992convergence}
			Hansen BE. 1992.
			\newblock Convergence to stochastic integrals for dependent heterogeneous
			processes.
			\newblock {\em Econometric Theory\/}~{\em 8\/}(4), 489--500.
			
			\bibitem[\protect\citeauthoryear{Harris}{Harris}{1997}]{harris1997principal}
			Harris D. 1997.
			\newblock Principal components analysis of cointegrated time series.
			\newblock {\em Econometric Theory\/}~{\em 13\/}(4), 529--557.
			
				
			\bibitem[H{\"o}rmann and Kokoszka, 2010]{hormann2010}
			H{\"o}rmann S, Kokoszka P. 2010.
			\newblock Weakly dependent functional data.
			\newblock {\em The Annals of Statistics}, 38(3):1845--1884.
			
			
			\bibitem[\protect\citeauthoryear{Horv{\'a}th, Kokoszka, and Reeder}{Horv{\'a}th
				et~al.}{2013}]{horvath2013estimation}
			Horv{\'a}th L, Kokoszka P, Reeder R. 2013.
			\newblock Estimation of the mean of functional time series and a two-sample
			problem.
			\newblock {\em Journal of the Royal Statistical Society: Series B (Statistical
				Methodology)\/}~{\em 75\/}(1), 103--122.
			
			\bibitem[\protect\citeauthoryear{Horv{\'a}th, Kokoszka, and Rice}{Horv{\'a}th
				et~al.}{2014}]{horvath2014test}
			Horv{\'a}th L, Kokoszka P, and Rice G. 2014.
			\newblock Testing stationarity of functional time series.
			\newblock {\em Journal of Econometrics\/}~{\em 179\/}(1), 66--82.
			
				\bibitem[\protect\citeauthoryear{Hron, Menafoglio, Templ,
			Hr$\mathring{\mathrm{u}}$zov{\'a}, and Filzmoser}{Hron
			et~al.}{2016}]{Hron2016330}
		Hron K, Menafoglio A, Templ M, Hr$\mathring{\mathrm{u}}$zov{\'a} K,
		Filzmoser P. 2016.
		\newblock Simplicial principal component analysis for density functions in
		Bayes spaces.
		\newblock {\em Computational Statistics \& Data Analysis\/}~{\em 94}, 330--350.
			
			\bibitem[\protect\citeauthoryear{Hyndman and Ullah}{Hyndman and
				Ullah}{2007}]{hyndman2007robust}
			Hyndman RJ, Ullah MS. 2007.
			\newblock Robust forecasting of mortality and fertility rates: a functional
			data approach.
			\newblock {\em Computational Statistics \& Data Analysis\/}~{\em 51\/}(10),
			4942--4956.
			
			\bibitem[\protect\citeauthoryear{Johansen}{Johansen}{1995}]{Johansen1996}
			Johansen S. 1995.
			\newblock {\em Likelihood-Based Inference in Cointegrated Vector Autoregressive
				Models}.
			\newblock Oxford University Press, Oxford.
			
			\bibitem[\protect\citeauthoryear{Klepsch, Kl{\"u}ppelberg, and Wei}{Klepsch
				et~al.}{2017}]{klepsch2017prediction}
			Klepsch J, Kl{\"u}ppelberg C, Wei T. 2017.
			\newblock Prediction of functional ARMA processes with an application to
			traffic data.
			\newblock {\em Econometrics and Statistics\/}~{\em 1}, 128--149.
			
			\bibitem[\protect\citeauthoryear{Kokoszka, Miao, Petersen, and Shang}{Kokoszka
			et~al.}{2019}]{kokoszka2019forecasting}
		Kokoszka P, Miao H, Petersen A, Shang HL. 2019.
		\newblock Forecasting of density functions with an application to
		cross-sectional and intraday returns.
		\newblock {\em International Journal of Forecasting\/}~{\em 35\/}(4),
		1304--1317.
	
			
			\bibitem[\protect\citeauthoryear{Kokoszka and Young}{Kokoszka and
				Young}{2016}]{kokoszka2016kpss}
			Kokoszka P,  Young G. 2016.
			\newblock {KPSS} test for functional time series.
			\newblock {\em Statistics\/}~{\em 50\/}(5), 957--973.
			
			
			\bibitem[\protect\citeauthoryear{Kwiatkowski, Phillips, Schmidt, and
				Shin}{Kwiatkowski et~al.}{1992}]{kwiatkowski1992testing}
			Kwiatkowski D, Phillips PCB, Schmidt P, Shin Y. 1992.
			\newblock Testing the null hypothesis of stationarity against the alternative
			of a unit root: How sure are we that economic time series have a unit root?
			\newblock {\em Journal of Econometrics\/}~{\em 54\/}(1-3), 159--178.
			
			
			\bibitem[\protect\citeauthoryear{Li, Robinson, and Shang}{Li
				et~al.}{2020}]{li2020long}
			Li D, Robinson PM, Shang HL. 2020.
			\newblock Long-range dependent curve time series.
			\newblock {\em Journal of the American Statistical Association\/}~{\em
				115\/}(530), 957--971.
			
			\bibitem[\protect\citeauthoryear{Li, Robinson, and Shang}{Li
				et~al.}{2022}]{LRS}
			Li D, Robinson PM, Shang HL. 2022.
			\newblock Nonstationary fractionally integrated functional time series. 	\newblock {\em Bernoulli\/}, in press.
			
			\bibitem[Loader, 1996]{loader1996}
		Loader CR. 1996.
		\newblock Local likelihood density estimation.
		\newblock {\em Annals of Statistics}, 24(4):1602--1618.
		
		\bibitem[Loader, 2006]{Loader2006}
		Loader CR. 2006.
		\newblock {\em Local Regression and Likelihood}.
		\newblock Springer, New York.
		
			\bibitem[\protect\citeauthoryear{Nielsen, Seo, and Seong}{Nielsen
				et~al.}{2022}]{SS2019}
			Nielsen, M{\O}, Seo WK, Seong D. 2022.
			\newblock Inference on the dimension of the nonstationary subspace in
			functional time series.
			\newblock {\em Econometric Theory\/}, in press.
			
			\bibitem[\protect\citeauthoryear{Nyblom and Harvey}{Nyblom and
				Harvey}{2000}]{nyblom2000tests}
			Nyblom J,  Harvey A. 2000.
			\newblock Tests of common stochastic trends.
			\newblock {\em Econometric Theory\/}~{\em 16\/}(2), 176--199.
			
			
			\bibitem[\protect\citeauthoryear{Park and Qian}{Park and
				Qian}{2012}]{Park2012397}
			Park JY, Qian J. 2012.
			\newblock Functional regression of continuous state distributions.
			\newblock {\em Journal of Econometrics\/}~{\em 167\/}(2), 397 -- 412.
			
			\bibitem[\protect\citeauthoryear{Petersen and M\"{u}ller}{Petersen and
			M\"{u}ller}{2016}]{petersen2016}
		Petersen A, M\"{u}ller HG. 2016.
		\newblock Functional data analysis for density functions by transformation to a
		{H}ilbert space.
		\newblock {\em Annals of Statistics\/}~{\em 44\/}(1), 183--218.
		
		\bibitem[Phillips, 1991]{phillips1988spectral}
		Phillips PCB. 1991.
		\newblock Spectral regression for cointegrated time series.
		\newblock In Barnett, W.~A., Powell, J., and Tauchen, G.~E., editors, {\em
			Nonparametric and semiparametric methods in econometrics and statistics:
			proceedings of the Fifth International Symposium in Economic Theory and
			Econometrics}, 413--435. Cambridge University Press.
		
	
			\bibitem[\protect\citeauthoryear{Phillips and Hansen}{Phillips and
				Hansen}{1990}]{phillips1990statistical}
			Phillips PCB, Hansen BE. 1990.
			\newblock Statistical inference in instrumental variables regression with
			{I}(1) processes.
			\newblock {\em Review of Economic Studies\/}~{\em 57\/}(1), 99--125.
			
			\bibitem[\protect\citeauthoryear{Phillips and Solo}{Phillips and
				Solo}{1992}]{Phillips1992}
			Phillips PCB, Solo V. 1992.
			\newblock Asymptotics for linear processes.
			\newblock {\em Annals of Statistics\/}~{\em 20\/}(2), 971--1001.
			
			\bibitem[Rice and Shang, 2017]{rice2017plug}
			Rice G, Shang HL. 2017.
			\newblock A plug-in bandwidth selection procedure for long-run covariance
			estimation with stationary functional time series.
			\newblock {\em Journal of Time Series Analysis}, 38(4), 591--609.
			
		
			\bibitem[\protect\citeauthoryear{Saikkonen}{Saikkonen}{1991}]{saikkonen1991asymptotically}
			Saikkonen P. 1991.
			\newblock Asymptotically efficient estimation of cointegration regressions.
			\newblock {\em Econometric Theory\/}~{\em 7\/}(1), 1--21.
			
			\bibitem[\protect\citeauthoryear{Seo}{Seo}{2022}]{seo_2022}
			Seo WK. 2022.
			\newblock Cointegration and representation of cointegrated autoregressive
			processes in {B}anch spaces.
			\newblock {\em Econometric {T}heory\/}, in press.
			
				\bibitem[\protect\citeauthoryear{Seo and Beare}{Seo and Beare}{2019}]{SEO2019}
			Seo WK,  Beare BK. 2019.
			\newblock Cointegrated linear processes in Bayes Hilbert space.
			\newblock {\em Statistics \& Probability Letters\/}~{\em 147}, 90 -- 95.
		
			
			\bibitem[\protect\citeauthoryear{Seo and Shang}{Seo and Shang}{2022}]{seoshang2022}
			Seo WK, Shang HL. 2022.
			\newblock {\em Fractionally integrated curve time series with cointegration.}
			\newblock {arXiv preprint arXiv:2212.04071}.
			
			\bibitem[\protect\citeauthoryear{Seong and Seo}{Seong and Seo}{2022}]{seong2022}
			Seong D, Seo WK. 2022.
			\newblock {\em Functional instrumental variable regression with an application to estimating the impact of immigration on native wages.}
			\newblock {arXiv preprint arXiv:2110.12722}.
				
			\bibitem[\protect\citeauthoryear{Shin}{Shin}{1994}]{shin1994residual}
			Shin Y. 1994.
			\newblock A residual-based test of the null of cointegration against the
			alternative of no cointegration.
			\newblock {\em Econometric Theory\/}~{\em 10\/}(1), 91--115.
			
			\bibitem[Skorohod, 1983]{skorohod2001}
			Skorohod A. 1983.
			\newblock {\em Random Linear Operators (Mathematics and its Applications)}.
			\newblock Springer, Dordrecht.
				
			\bibitem[\protect\citeauthoryear{Zhang, Kokoszka, and Petersen}{Zhang
				et~al.}{2021}]{zhang2020wasserstein}
			Zhang C, Kokoszka P, Petersen A. 2021.
			\newblock Wasserstein autoregressive models for density time series.
			\newblock {\em Journal of Time Series Analysis\/}~{\em 43\/}(1), 30--52.
			
			
			
			
		\end{thebibliography}


\end{document}